\documentclass[reqno,12pt]{article}
\usepackage{amsfonts, amsmath, amssymb, amsthm}
\usepackage{braket}
\usepackage[T1]{fontenc}
\usepackage{authblk}
\usepackage{xcolor,eucal,mathrsfs}
\usepackage[latin1]{inputenc}
\usepackage[hidelinks]{hyperref}
\usepackage{needspace}

\usepackage{mathtools}

\newcommand{\fc}{f_{\mathrm{cov}}}
\newcommand{\gc}{g_{\mathrm{cov}}}
\newcommand{\ccov}{\mathcal{C}_{\mathrm{cov}}}
\newcommand{\lcc}{\mathcal{L}_{\mathrm{cov}}} 
\newif\ifshowrevisions
\showrevisionsfalse
\ifshowrevisions
  \newcommand{\rev}[1]{{\color{blue}#1}}
  \newenvironment{revision}{\begingroup\color{blue}}{\endgroup}
\else
  \newcommand{\rev}[1]{#1}
  \newenvironment{revision}{}{}
\fi

\newcommand{\proj}[1]{\ket{#1}\!\bra{#1}}

\DeclareMathOperator{\Sp}{Sp}

\DeclareMathOperator{\Dom}{Dom}
\DeclareMathOperator{\Ran}{Ran}

\DeclareMathOperator{\I}{\mathbb{I}}

\DeclareMathOperator{\Tr}{\mathrm{Tr}}

\DeclareMathOperator{\QW}{\mathcal{QW}}
\DeclareMathOperator{\C}{\mathbb{C}}
\DeclareMathOperator{\rank}{rank}

\newcommand{\cL}{\mathcal{L}}

\newcommand{\cB}{\mathcal{B}}

\newcommand{\cH}{\mathcal{H}}

\newcommand{\cA}{\mathcal{A}}

\newcommand{\cU}{\mathcal{U}}
\newcommand{\cM}{\mathcal{M}}

\newcommand{\cS}{\mathcal{S}}
\newcommand{\cT}{\mathcal{T}}

\newcommand{\cK}{\mathcal{K}}

\newcommand{\smat}[1]{\left(\begin{smallmatrix}#1\end{smallmatrix}\right)}

\def\R{\mathbb R}

\usepackage{enumitem}
\newcommand{\restr}[1]{\lower3pt\hbox{$|_{#1}$}}

\newcommand{\N}{\mathbb{N}}
\newcommand{\Real}{\operatorname{Re}}
\newcommand{\Imag}{\operatorname{Im}}
\theoremstyle{plain}
\newtheorem{theorem}{Theorem}[section]
\newtheorem{theo}[theorem]{Theorem}

\newtheorem{proposition}[theorem]{Proposition}
\newtheorem{prop}[theorem]{Proposition}

\newtheorem{deff}[theorem]{Definition}

\newtheorem{remark}[theorem]{Remark}
\newtheorem{lemma}[theorem]{Lemma}
\newtheorem{coro}[theorem]{Corollary}
\newtheorem{corollary}[theorem]{Corollary}
\newtheorem{cor}[theorem]{Corollary}

\numberwithin{equation}{section}
\theoremstyle{remark}

\newtheorem{example}[theorem]{Example}
\DeclareFontFamily{U}{mathx}{\hyphenchar\font45}
\DeclareFontShape{U}{mathx}{m}{n}{
	<5> <6> <7> <8> <9> <10>
	<10.95> <12> <14.4> <17.28> <20.74> <24.88>
	mathx10
}{}
\DeclareSymbolFont{mathx}{U}{mathx}{m}{n}
\DeclareMathSymbol{\bigtimes}{1}{mathx}{"91}

\title{Quantum Optimal Transport Barycenters: \\ Existence, Duality, and Gaussian Rigidity}
\author[1]{Augusto Gerolin}
\author[2]{Zhiyi Lin}

\affil[1]{Instituto de Matem\'atica Pura e Aplicada, Brazil}
\affil[2]{Department of Mathematics and Statistics, University of Ottawa}
\date{}
\hypersetup{
pdftitle={Quantum Optimal Transport Barycenters: Existence, Duality, and Gaussian Rigidity},
pdfauthor={Augusto Gerolin and Zhiyi Lin},
pdfsubject={Quantum optimal transport barycenters},
pdfkeywords={quantum optimal transport, Wasserstein barycenters, Gaussian states, quantum Wasserstein distance}
}

\begin{document}
\maketitle
\begin{abstract}
We develop a Quantum Optimal Transport (QOT) barycenter framework for quantum states, as an analog of Wasserstein barycenters~\cite{AguCar}, and establish existence and duality results for a broad class of possibly unbounded transport costs on separable Hilbert spaces. Our framework provides, in particular, a unified treatment of $2$-quantum Wasserstein (QW) barycenters in both the quantum-state and quantum-channel formulations by specializing to the canonical quadratic cost operators associated with the $2$-quantum Wasserstein distances of Caglioti--Golse--Mouhot--Paul~\cite{caglioti2021towards,Golse-Mouhot-Paul:2016} and De Palma--Trevisan~\cite{DPaTre19}.

The central results concern Gaussian input states and, in particular, Gaussian rigidity: whether Gaussian input states force the 2-QW barycenter itself to be Gaussian and uniquely determined when the minimization is taken over all quantum states. We first show that the barycenter problem admits a Gaussian minimizer and reduces to a finite-dimensional convex optimization problem over covariance matrices. 

The main difficulty is that uniqueness of the optimal covariance does not, in general, imply uniqueness of the underlying quantum state. \rev{We bridge this gap through a state-reconstruction principle under covariance complementary slackness that upgrades uniqueness of the optimal covariance to uniqueness of the full quantum state, and thereby prove a global rigidity theorem:} if at least one Gaussian input is faithful, then the barycenter is unique among all quantum states and is necessarily Gaussian. \rev{Faithfulness is sufficient but not necessary: families of pure inputs still determine a unique barycenter, whereas partially pure nonfaithful Gaussian inputs may admit multiple barycenters.}

Finally, within the same Gaussian framework, we also prove semiclassical convergence to the classical Gaussian $2$-Wasserstein barycenter. As an explicit class of examples, we obtain characterizations of the 2-QW barycenters for displaced isotropic thermal inputs.
\end{abstract}

\begingroup
\footnotesize
\linespread{0.95}\selectfont
\tableofcontents
\endgroup

\newpage

\section{Introduction}
In this paper, we introduce the quantum analog of the celebrated Wasserstein barycenter for probability measures~\cite{AguCar}: let $[N]$ be an index set of $N\in\mathbb{N}$ elements, $\cH_0$ and $(\cH_s)_{s\in[N]}$ be separable Hilbert spaces,
$\boldsymbol{\sigma}:=(\sigma_s)_{s\in[N]}$ be quantum states $\sigma_s\in\cS(\cH_s)$, $\mathbf{C}:=(C_s)_{s\in[N]}$ be a family of nonnegative self-adjoint cost operators $C_s \in \cL_{\geq}\bigl(\cH_0 \otimes \cH_s\bigr)$,
$\boldsymbol{\alpha}:=(\alpha_s)_{s\in[N]}$ be positive real numbers satisfying
$\sum_{s=1}^N \alpha_s = 1$.

The Quantum Optimal Transport barycenter problem is then defined by
\begin{equation}\label{intro:QOT-barycenter}
\cB^{\mathbf C,\boldsymbol{\alpha}}_{\mathbf A,\mathbf E}(\boldsymbol{\sigma}) 
:=
\inf_{\rho\in\cS_{\mathbf A,\mathbf E}(\cH_0)}
\sum_{s=1}^N
\alpha_s\,
\operatorname{QOT}\!\big(\rho,\sigma_s\mid C_s\big),
\end{equation}
where $\operatorname{QOT}(\rho,\sigma_s\mid C_s)$ is a Quantum Optimal Transport value between quantum states $\rho$ and $\sigma_s$ with cost operator $C_s$; see \eqref{eq:QOT-def-setting}. The constrained state space $\cS_{\mathbf A,\mathbf E}(\cH_0)$ is defined as
\begin{equation}\label{eq:constrain}
    \cS_{\mathbf A,\mathbf E}(\cH_0)
    :=
    \left\{
        \rho\in\cS(\cH_0):
        \Tr(A_\ell\rho)\le E_\ell,
        \ \ell\in[k]\right\}, \, 
\end{equation}
where $\mathbf E=(E_1,\ldots,E_k)\in[0,+\infty)^k$ and $\{A_\ell\}_{\ell\in[k]}\subset\cL_\geq(\cH_0)$ is a finite family of nonnegative self-adjoint
operators on $\cH_0$. The constraints \( \Tr(A_\ell \rho) \leq E_\ell \) allow us to impose finite-energy or moment-type bounds on the barycentric state; the unconstrained problem is recovered by omitting these conditions.

An element
\begin{equation}\label{intro:rhobar}
\overline{\rho}\in \mathrm{Bar}^{\mathbf C,\boldsymbol{\alpha}}_{\mathbf A,\mathbf E}(\boldsymbol{\sigma})
:=\underset{\rho\in\cS_{\mathbf A,\mathbf E}(\cH_0)}{\operatorname{argmin}}\ \sum_{s=1}^N \alpha_s\,
\operatorname{QOT}\!\big(\rho,\sigma_s\mid C_s\big)
\end{equation}
is called a QOT barycenter of $(\sigma_s)_{s\in[N]}$ associated with the costs $(C_s)_{s\in[N]}$.

This abstract framework already includes, as important examples, the two
quadratic transport models considered in this paper: the
$2$-quantum Wasserstein (QW) transport problem for quantum states introduced by
Caglioti--Golse--Mouhot--Paul~\cite{caglioti2021towards,Golse-Mouhot-Paul:2016}, and its
counterpart for quantum channels introduced by De Palma and
Trevisan~\cite{DPaTre19}. 

To treat these two models simultaneously, we use the language of Weyl systems as a common framework for the canonical observables and their
associated quadratic costs. Accordingly, we write both barycenter problems in the
unified form
\begin{equation}\label{eq:tau-barycenter-intro}
\cB_\tau(\hat{\boldsymbol\sigma}\mid\boldsymbol\alpha)
:=
\inf_{\hat\rho\in\cS(L_2(\mathbb R^m))}
\sum_{s=1}^N \alpha_s\,\QW_{2,\tau}^2(\hat\rho,\hat\sigma_s)
=
\inf_{\hat\rho\in\cS(L_2(\mathbb R^m))}
\inf_{\substack{\forall s\in[N]\\ \hat\pi_s\in\Gamma_\tau(\hat\rho,\hat\sigma_s)}}
\sum_{s=1}^N \alpha_s\,\Tr\!\big(\hat C^{(\tau)}\hat\pi_s\big).
\end{equation}
Here $\QW_{2,\tau}$ denotes the corresponding $2$-quantum Wasserstein
transport cost, $\hat C^{(\tau)}$ is the quadratic transport cost generated by the canonical
observables, and \rev{the associated coupling sets are $\Gamma_{+1}(\hat\rho,\hat\sigma_s):=\Gamma(\hat\rho,\hat\sigma_s)$ and $\Gamma_{-1}(\hat\rho,\hat\sigma_s):=\Gamma(\hat\rho,\hat\sigma_s^\top)$}; see the precise definitions in equations~\eqref{def:quantumcoupling}, \eqref{eq:tau-cost}, and~\eqref{eq:qw}.

The parameter \rev{\(\tau\in\{\pm1\}\)} distinguishes the two transport models considered in this paper. 
The choice \(\tau=+1\) gives the quantum-state formulation, while the choice 
\(\tau=-1\) gives the quantum-channel formulation. In particular, any element
\[
\hat{\overline\rho}_\tau \in \mathrm{Bar}_\tau(\hat{\boldsymbol\sigma}\mid\boldsymbol\alpha):=
\operatorname{argmin}_{\hat\rho\in\cS(L_2(\mathbb R^m))}
\sum_{s=1}^N \alpha_s\,\QW_{2,\tau}^2(\hat\rho,\hat\sigma_s)
\]
is called a $2$-QW barycenter in the quantum-state formulation $(\tau=+1)$ or the quantum-channel formulation $(\tau=-1)$.

\noindent
\textbf{\rev{Gaussian quantum states:}} \rev{In optimal transport theory}, for nondegenerate Gaussian measures, the
quadratic Wasserstein barycenter~\cite{AguCar} is unique and Gaussian, and its
computation reduces to a finite-dimensional problem involving means
and covariance matrices -- see also \cite{MalGerMin20}. In the quantum setting, the analogous
reduction raises a further question: does an optimal covariance
matrix determine the minimizing quantum state?

We study this question for barycenters associated with two static
quadratic quantum optimal transport models. The expectation of a
quadratic coupling cost depends only on the first and second
canonical moments of the joint state. For Gaussian input states,
Gaussian replacement therefore gives an exact reduction of the
minimum to a covariance optimization and guarantees a Gaussian
minimizer. This does not, however, classify all minimizers: distinct
quantum states may have the same first and second moments. In
particular, uniqueness within the Gaussian family does not by
itself exclude additional non-Gaussian barycenters. The central
question of this paper is therefore:
\begin{quote}
\itshape
Under what conditions do Gaussian input states force Gaussianity
and uniqueness of the barycenter among all quantum states?
\end{quote}
We prove that a single faithful input within a family of Gaussian
states suffices, in both transport formulations considered here.
The key step is a reconstruction principle that determines an
optimal coupling from its prescribed Gaussian marginal and a
covariance slackness relation. This turns a finite-dimensional
optimality certificate into a rigidity statement for the entire
density operator.

The terminology ``quantum optimal transport'' encompasses several
distinct noncommutative extensions of classical transport.
These include the dynamical approach of
Carlen--Maas~\cite{CarMas14}, which constructs a quantum analogue
of the Benamou--Brenier formulation through a noncommutative
continuity equation; and the transport of positive semidefinite
matrix-valued measures developed in
\cite{PeyreChizatVialardSolomon2019}. In the latter framework,
marginal discrepancies (probability densities) are penalized using quantum relative
entropy, and quantum-entropic regularization leads to an
iterative scaling algorithm for transport and tensor-field
barycenter problems.

In contrast, the present work follows a static, Kantorovich-type
approach to transport between density operators on Hilbert
spaces; see, for example,
\cite{CapGerMonPel26,CapGerMonPor24,Golse-Mouhot-Paul:2016,
FelGerPor23,GerMon2023,GerPel25}.
Transport costs are minimized over bipartite states with
prescribed partial traces, rather than over paths satisfying
a continuity equation. In the quantum-channel formulation,
for a fixed source state, channels defined on the trace-class
operators over its support correspond bijectively to couplings
on $\cH\otimes\cH^*$ whose marginals are the output state and
the transpose of the source state
\cite[Def.~3 and Prop.~2]{DPaTre19}.

\subsection*{Main contributions and \rev{organization} of the paper}

\noindent
\textbf{Existence and duality.} We develop a variational theory of Quantum Optimal Transport barycenters on separable
Hilbert spaces, covering finite-dimensional and continuous-variable
systems and providing a common framework for the quantum-state~\cite{Golse-Mouhot-Paul:2016} and
quantum-channel~\cite{DPaTre19,DPaTre2023} formulations. Under confinement and finite-cost
feasibility assumptions, we prove existence and Kantorovich duality,
with or without barycentric energy constraints
(Theorems~\ref{thm:existenceQOTbary_intro}--\ref{thm:energy-constrained-barycenter-duality} and Corollary~\ref{cor:unconstrained-barycenter-duality}). For quadrature-generated
costs, we also establish an \emph{a posteriori} energy estimate
under the hypotheses of \rev{Proposition~\ref{prop:finitesecondemomentrho}}.

\noindent
\textbf{Gaussian reduction and global uniqueness.} For the canonical quadratic costs, Gaussian inputs admit a Gaussian
barycenter, and the minimum is characterized by a finite-dimensional
convex covariance optimization (Theorem~\ref{thm:gaussian-barycenter-intro}). Our central result (Gaussian rigidity)
upgrades this reduction to uniqueness in the full state space:
if all inputs are Gaussian and at least one is faithful, the
barycenter is unique among all quantum states and is Gaussian,
in both formulations (Theorem~\ref{thm:faithful-global-uniqueness-intro}). The underlying reconstruction
principle is established in Propositions~\ref{prop:dual-support-gaussian-source}--\ref{prop:saturated-face-reconstruction}.
Finally, faithfulness is sufficient but not necessary: families of pure
states with finite canonical second moments have a unique
coherent-state barycenter (Lemma~\ref{lem:pure-input-coherent-barycenter}), whereas partially pure
Gaussian inputs can admit Gaussian and non-Gaussian barycenters
with identical first and second moments (Example~\ref{ex:nonfaithful-boundary}).

\noindent
\textbf{Semiclassical limit and thermal inputs.}
For Gaussian means and positive-definite covariances \rev{fixed independently of $\hbar$}, Theorem~\ref{thm:gaussian-semiclassical-intro} proves convergence of
the quantum barycenter covariances to the covariance of the
classical Gaussian $2$-Wasserstein barycenter as
$\hbar\to0^+$. The associated Wigner distributions converge
weakly to that classical Gaussian measure. The proof also
gives an $O(\hbar)$ bound for the difference between the
quantum and classical minimum values. These conclusions
hold for both transport formulations. For displaced
isotropic thermal inputs, Corollary~\ref{cor:thermal-barycenter-intro} gives a unique
thermal barycenter, with an explicit covariance formula
in the quantum-state formulation and a one-dimensional
convex characterization in the quantum-channel formulation.

\noindent
\textbf{Organization of the paper.}
Section~\ref{intro:statementofresults} states the main results and
introduces the two classes of cost operators considered throughout the
paper. Section~\ref{sec:preliminary} collects the operator-theoretic,
Weyl-system, and continuous-variable preliminaries.
Section~\ref{sec:qot-bary} proves existence, duality, and the
finite-energy estimate for general QOT barycenters.
Section~\ref{sec:gaussian-states} develops the Gaussian reduction, covariance duality and faithful-input uniqueness, state reconstruction under a covariance slackness relation, the semiclassical limit, the thermal case, and the nonuniqueness example on the partially pure boundary. The appendices collect the analytic tools used in the proofs,
including compactness for quantum states, essential self-adjointness of quadratic Hamiltonians, and
Fenchel--Rockafellar duality. We also develop the symplectic normal forms and Gaussian-channel representations underlying the rigidity argument, together with isotropic symmetrization
and the explicit calculations for thermal transport and the
nonuniqueness example.

 \section{Statement of the main results}\label{intro:statementofresults}
We introduce two classes of cost operators that cover all examples considered in this work. Let $\cH_1,\cH_2$ be separable Hilbert spaces and $C\in\cL_{\geq}(\cH_1\otimes\cH_2)$ be a nonnegative self-adjoint operator. 

\noindent
\textbf{Compact-resolvent cost operator.}
We say that \(C\) is a compact-resolvent cost operator if \(C\) has a compact resolvent in the sense recalled in Definition~\ref{def:compact-resolvent}. This assumption covers many discrete-spectrum energy operators, including confining
Hamiltonians, and is automatically satisfied in finite dimensions: every positive semidefinite cost operator on a finite-dimensional Hilbert space is a compact-resolvent cost operator.

\noindent
\textbf{Quadrature-generated costs.}
The second class consists of noncommutative analogues of the quadratic Euclidean cost $|\mathbf x-\mathbf y|^2=\sum_j(x_j-y_j)^2$
from classical optimal transport~\cite{FigalliGlaudobook,Santabook,VilON}.
The guiding idea is to replace the classical coordinate functions by two
families of quantum observables and to measure their quadratic discrepancy. Let $R_1,\ldots,R_{2m}\in\cL_S(\cH_1)$ and $S_1,\ldots,S_{2m}\in\cL_S(\cH_2)$ be self-adjoint operators, called quadratures. Formally, the
corresponding quadrature-generated cost $C$ on \(\cH_1\otimes \cH_2\) is
\[
  C
  =
  \sum_{j=1}^{2m}
  \left(
    R_j\otimes \I_{\cH_2}
    -
    \I_{\cH_1}\otimes S_j
  \right)^2, 
\]
see
also \cite[Def.~5]{DPaTre19} for a related quadratic cost. Since the quadratures are typically unbounded, this expression requires a self-adjoint
realization. In this paper, the quadratures $\{R_j\}_{j\in[2m]}$ and $\{S_j\}_{j\in[2m]}$ are taken to be the canonical generators of
regular Weyl systems
\[
  W_1:\mathbb R^{2m}\to\mathcal U(\cH_1),
  \quad
  W_2:\mathbb R^{2m}\to\mathcal U(\cH_2),
\]
\rev{Here \(\mathcal U(\cH_i)\) denotes the unitary group on \(\cH_i\); see Subsection~\ref{subsec:op-space} for our operator-theoretic notation. The systems \(W_1\) and \(W_2\) have nondegenerate skew-symmetric commutation matrices \(\Theta_1\) and \(\Theta_2\), respectively. Regular Weyl systems, their commutation matrices, and their canonical generators are defined in Subsection~\ref{pre:gauss}.} The cost \(C\) is understood as the unique self-adjoint closure specified in
Remark~\ref{rem:quadratic-hamiltonian-closure-convention}.

For compactness arguments, we assume that the Weyl system $W_1$ on $\cH_1$ is irreducible. Then the associated one-body square-sums are
\begin{equation}
\label{intro:eq:one-body-H}
\mathsf H_{\mathbf R}:=\sum_{j=1}^{2m}R_j^2,
\qquad
\mathsf H_{\mathbf S}:=\sum_{j=1}^{2m}S_j^2.
\end{equation}
The operator $\mathsf H_{\mathbf R}$ has compact resolvent; see
Corollary~\ref{coro:nelson-square-compact-resolvent}.

\subsection{QOT barycenters: Existence and duality}
The existence of QOT barycenters is established by the direct method of \rev{the calculus of variations}. In infinite dimensions, prescribing only one marginal does not in general make the coupling set compact in trace norm. We therefore require one cost \(C_{s_0}\) to belong to one of the two classes of costs described above to obtain the compactness of the minimizing sequences.
\begin{theorem}[Existence of QOT barycenters]
\label{thm:existenceQOTbary_intro} Let $[N]$ be an index set of $N\in\mathbb N$ elements and
$\boldsymbol{\alpha}:=(\alpha_s)_{s\in[N]}$ be positive real numbers satisfying
$\sum_{s=1}^N \alpha_s = 1$. Let 
$\cH_0$ and $(\cH_s)_{s\in[N]}$ be separable Hilbert spaces, and
$\boldsymbol{\sigma}:=(\sigma_s)_{s\in[N]}$ with $\sigma_s\in\cS(\cH_s)$ be a family of quantum states. 

Let $\mathbf{C}:=(C_s)_{s\in[N]}$ be a family of nonnegative self-adjoint cost operators $C_s\in\cL_{\geq}(\cH_0\otimes\cH_s)$. Assume that there exists $s_0\in[N]$ such that one of the following conditions holds:
\begin{itemize}
\item[(i)] $C_{s_0}$ has compact resolvent on
$\cH_0\otimes\cH_{s_0}$;

\item[(ii)] $C_{s_0}$ is a quadrature-generated cost, the associated one-body square-sum $H_0^{(s_0)}$ on $\cH_0$ has compact resolvent, and the one-body square-sum $H^{(s_0)}$ on
$\cH_{s_0}$ satisfies $\Tr\bigl(H^{(s_0)}\sigma_{s_0}\bigr)<+\infty$.
\end{itemize}

Let $k\in\N$, $\mathbf E=(E_1,\ldots,E_k)\in[0,+\infty)^k$ and $\mathbf A:=(A_\ell)_{\ell\in[k]}$, where
$A_\ell\in\cL_{\geq}(\cH_0)$ for every $\ell\in[k]$. Assume that the constrained barycentric state set $\cS_{\mathbf A,\mathbf E}(\cH_0)$, defined in \eqref{eq:constrain}, is nonempty and that there exists $\rho_0\in\cS_{\mathbf A,\mathbf E}(\cH_0)$ such that
\begin{equation}\label{eq:exas}
\sum_{s=1}^{N}\alpha_s\,\operatorname{QOT}\!\big(\rho_0,\sigma_s\mid C_s\big)<+\infty.
\end{equation}
Then the QOT barycenter problem~\eqref{intro:QOT-barycenter} admits a barycenter. 
\end{theorem}

For duality, we adapt the Fenchel--Rockafellar approach for the two-marginal problem of
Caglioti--Golse--Paul~\rev{\cite[Theorem~2.2]{caglioti2021towards}} to the shared-marginal barycenter problem,
including barycentric energy constraints. The resulting dual
formula uses bounded self-adjoint operator potentials. We assert
equality of the primal and dual values, without asserting
attainment in this class of potentials. 

We next state the corresponding strong duality result, which is a consequence of the Fenchel-Rockafellar theorem~\ref{app:FRTheorem}.
\begin{theorem}[Duality with barycentric energy constraints]
\label{thm:energy-constrained-barycenter-duality}
Under the same hypotheses as in Theorem~\ref{thm:existenceQOTbary_intro}, the following duality formula holds for the QOT barycenter problem~\eqref{intro:QOT-barycenter}\rev{:}
\begin{equation}
\label{eq:energy-constrained-duality-barycenter}
\cB_{\mathbf A,\mathbf E}^{\mathbf C,\boldsymbol\alpha}(\boldsymbol{\sigma})
=
\sup
\left\{
\sum_{s=1}^N \alpha_s\Tr(V_s\sigma_s)
-
\sum_{\ell\in[k]}\eta_\ell E_\ell
:
((U_s,V_s)_{s\in[N]},(\eta_\ell)_{\ell\in[k]})
\in
\mathsf{Pot}^{\mathbf C,\boldsymbol\alpha}_{\mathbf A,\mathbf E}
(\boldsymbol\sigma)
\right\}, 
\end{equation}
where 
$\mathsf{Pot}^{\mathbf C,\boldsymbol\alpha}_{\mathbf A,\mathbf E}
(\boldsymbol\sigma)$
is the set of admissible energy-constrained potentials defined by
\begin{equation}\label{eq:pot}
\begin{aligned}
\mathsf{Pot}^{\mathbf C,\boldsymbol\alpha}_{\mathbf A,\mathbf E}
(\boldsymbol\sigma)
:=
\Bigl\{ &((U_s,V_s)_{s\in[N]},(\eta_\ell)_{\ell\in[k]})
:
\, U_s\in\cB_S(\cH_0),\,
V_s\in\cB_S(\cH_s),\, U_s\oplus V_s\le C_s,
\\
&
\text{there exist }
U_s^{(\ell)}\in\cB_S(\cH_0) 
\text{ such that }
U_s=\sum_{\ell=1}^k U_s^{(\ell)}\, \text{for all }
s\in[N], \,
\\
&
-\sum_{s=1}^N\alpha_s U_s^{(\ell)}
\le
\eta_\ell A_\ell \text{ and }\eta_\ell\ge0\, \text{for all } \ell\in[k]  \, 
\Bigr\}.
\end{aligned}
\end{equation}
\end{theorem}

A particular case of interest is the unconstrained QOT barycenter problem,
obtained from problem~\eqref{intro:QOT-barycenter} by removing the energy
constraint. It can be written as
\begin{equation}
\label{eq:unconstrained-barycenter-primal}
\begin{aligned}
\cB^{\mathbf C,\boldsymbol\alpha}_0(\boldsymbol{\sigma})
&:=
\inf_{\rho\in\cS(\cH_0)}
\sum_{s=1}^N
\alpha_s\,
\operatorname{QOT}\!\big(\rho,\sigma_s\mid C_s\big),\\
\mathrm{Bar}^{\mathbf C,\boldsymbol\alpha}_0(\boldsymbol{\sigma})
&:=
\arg\min_{\rho\in\cS(\cH_0)}
\sum_{s=1}^N
\alpha_s\,
\operatorname{QOT}\!\big(\rho,\sigma_s\mid C_s\big).
\end{aligned}
\end{equation}

As a consequence, we derive the unconstrained duality
formula.

\begin{coro}[Unconstrained barycenter duality]
\label{cor:unconstrained-barycenter-duality}
Assume the cost assumptions of Theorem~\ref{thm:existenceQOTbary_intro} and assume that
there exists \(\rho_0\in\cS(\cH_0)\) such that $\sum_{s=1}^N\alpha_s \operatorname{QOT}(\rho_0,\sigma_s\mid C_s)<+\infty$. Then we have the following dual formula for \eqref{eq:unconstrained-barycenter-primal}:
\begin{equation}
\label{eq:unconstrained-barycenter-duality}
\cB^{\mathbf C,\boldsymbol\alpha}_0(\boldsymbol{\sigma})
=
\sup
\left\{
\sum_{s=1}^N \alpha_s \Tr(V_s\sigma_s)
:
(U_s,V_s)_{s\in[N]}
\in
\mathsf{Pot}_0^{\mathbf C,\boldsymbol\alpha}(\boldsymbol\sigma)
\right\},
\end{equation}
where the admissible potential set is
\[
\mathsf{Pot}_0^{\mathbf C,\boldsymbol\alpha}(\boldsymbol\sigma)
:=
\left\{
(U_s,V_s)_{s\in[N]}:
\begin{array}{l}
U_s\in\cB_S(\cH_0),\quad V_s\in\cB_S(\cH_s),
\\[1mm]
\displaystyle \sum_{s=1}^N \alpha_s U_s=0, \quad U_s\oplus V_s\le C_s,\quad s\in[N]
\end{array}
\right\}.
\]
\end{coro}

Finally, although the unconstrained QOT barycenter problem imposes no a priori energy constraint on the barycenter, quadrature-generated costs still provide a natural a posteriori one-body energy estimate.

\begin{prop}[Finite energy of the barycenter for quadrature-generated costs]\label{prop:finitesecondemomentrho}
Assume the cost assumptions of Theorem~\ref{thm:existenceQOTbary_intro} in the
unconstrained case, and assume in addition that every \(C_s\) is quadrature-generated. Let \(H^{(s)}_0\) and \(H^{(s)}\) be the one-body square-sums associated with $C_s$ as in~\eqref{intro:eq:one-body-H}. Assume, in addition, that $\operatorname{Tr}\bigl(H^{(s)}\sigma_s\bigr)<+\infty$ for all $s\in[N]$. Let $\overline{\rho}\in
\mathrm{Bar}^{\mathbf C,\boldsymbol\alpha}_0(\boldsymbol\sigma)$ be a QOT barycenter. If 
\[
M:=\inf_{\rho\in\cS(\cH_0)}\sum_{s=1}^N \alpha_s\Tr\!\Big(\rho H_0^{(s)}\Big) < \infty,
\]
then
\begin{equation}
\sum_{s=1}^N \alpha_s\Tr\!\Big(\overline{\rho} H_0^{(s)}\Big)
 \le \frac83M+\frac{20}{3}\sum_{s=1}^N\alpha_s\Tr(\sigma_s H^{(s)})<+\infty.
\end{equation}

Consequently, the minimum in~\eqref{eq:unconstrained-barycenter-primal}
is unchanged if the barycentric search space is restricted to $
\bigcap_{s=1}^N
\left\{
\rho\in\cS(\cH_0):
\Tr\!\bigl(H_0^{(s)}\rho\bigr)<+\infty
\right\}$.
\end{prop}

\subsection{\texorpdfstring{$2$}{2}-QW barycenters for Gaussian states: reduction, uniqueness and semiclassical limits}

\rev{Our general framework provides a unified treatment of $2$-quantum Wasserstein barycenters in both the quantum-state formulation
introduced by Golse--Mouhot--Paul~\cite{Golse-Mouhot-Paul:2016}
and developed further by Caglioti--Golse--Paul
\cite{caglioti2021towards}, and the quantum-channel
formulation of De Palma--Trevisan~\cite{DPaTre19}.}

We now specialize to the $2$-quantum Wasserstein barycenter problems
\eqref{eq:tau-barycenter-intro} with Gaussian states; see Definition~\ref{def:weyl-characteristic-gaussian-state}. Throughout
this subsection, $m,N\in\N$, $\alpha_s>0$ with
$\sum_{s=1}^N\alpha_s=1$, and
$\hat{\boldsymbol\sigma}^G
=(\hat\sigma_s^G)_{s\in[N]}$ is a family of Gaussian states on
$L_2(\R^m)$ with mean vectors
$\mathbf m_s\in\R^{2m}$ and covariance matrices $\Sigma_s$.

Let
\begin{equation}\label{eq:covqm}
\Omega_m:=\bigoplus_{j=1}^m{}_D
\begin{pmatrix}0&1\\-1&0\end{pmatrix},
\qquad
\mathrm{Cov}_m^{\mathrm Q}
:=
\left\{
\Sigma\in\cM_{2m,\mathrm{sym}}(\R):
\Sigma+\frac{i\hbar}{2}\Omega_m\succeq0
\right\}.
\end{equation}
For $\tau\in\{\pm1\}$ and
$\Sigma_1,\Sigma_2\in\mathrm{Cov}_m^{\mathrm Q}$, set
\begin{equation}\label{eq:corr}
\mathrm{Corr}_{\tau}(\Sigma_1,\Sigma_2)
:=
\left\{
A\in\cM_{2m}(\R):
\begin{pmatrix}
\Sigma_1&A\\
A^\top&\Sigma_2
\end{pmatrix}
+
\frac{i\hbar}{2}
\bigl(\Omega_m\oplus_D(\tau\Omega_m)\bigr)
\succeq0
\right\}.
\end{equation}

The preceding general theory guarantees the existence of QOT barycenters and provides a dual variational formulation, but it does not in itself characterize the barycenters or ensure their uniqueness. These questions become particularly natural for the 2-QW costs and Gaussian input states, where one may ask whether the infinite-dimensional barycenter problem inherits the finite-dimensional structure of Gaussian quantum mechanics.

For nondegenerate classical Gaussian measures, the quadratic
Wasserstein barycenter is unique and Gaussian, and is
characterized by its mean and covariance~\cite{AguCar}.
The quantum problem admits a related finite-dimensional
reduction.

\begin{theorem}[Gaussian reduction for $2$-QW barycenters]
\label{thm:gaussian-barycenter-intro}
For each $\tau\in\{\pm1\}$, consider the $2$-QW barycenter problem 
\eqref{eq:tau-barycenter-intro} with Gaussian inputs
$\hat{\boldsymbol\sigma}^G$:
\begin{equation}
\label{eq:gaussianqw}
\min_{\hat\rho\in\cS(L_2(\R^m))}
\sum_{s=1}^N
\alpha_s\,\QW_{2,\tau}^2(\hat\rho,\hat\sigma_s^G).
\end{equation}
Then \eqref{eq:gaussianqw} admits a Gaussian barycenter. Every Gaussian
barycenter has mean
\[
\overline{\mathbf m}
=
\sum_{s=1}^N\alpha_s\mathbf m_s.
\]
Define the Gaussian barycenter covariance problem by
\begin{equation}
\mathsf P_{\tau,\mathrm{cov}}
:=
\min_{\substack{
\Sigma\in\mathrm{Cov}_m^{\mathrm Q},\\
A_s\in\mathrm{Corr}_{\tau}(\Sigma,\Sigma_s),\ s\in[N]}}
\left\{
\Tr(\Sigma)
+
\sum_{s=1}^N\alpha_s
\bigl(\Tr(\Sigma_s)-2\Tr(A_s)\bigr)
\right\}.
\label{eq:qw2cov}
\end{equation}
The covariance matrix of every Gaussian barycenter is the first
component of an optimizer of \eqref{eq:qw2cov}. Conversely, every
optimizer of \eqref{eq:qw2cov}, together with the mean
$\overline{\mathbf m}$, determines a Gaussian barycenter. Here
$\tau=+1$ gives the quantum-state formulation and $\tau=-1$ the
quantum-channel formulation.
\end{theorem}

\paragraph{Idea of the Gaussian reduction.}
The expectation of a quadratic transport cost depends only on
the first and second canonical moments of the joint state.
Gaussian projection preserves these moments and is compatible
with partial traces. Applied to an optimal family of couplings,
whose canonical second moments are finite by the preceding
energy estimates, it therefore leaves the costs and prescribed
Gaussian marginals unchanged, while replacing the common
barycentric marginal by its Gaussian projection. This produces
a Gaussian barycenter and shows that restricting the search to
Gaussian states does not change the minimum. For the channel
formulation, the underlying replacement argument for two
prescribed Gaussian marginals is due to
De Palma--Trevisan~\cite[Prop.~7]{DPaTre19}.

Within the Gaussian class, the contribution of the means can
be minimized independently, giving $\overline{\mathbf m}= \sum_{s=1}^N\alpha_s\mathbf m_s$.
The remaining optimization is exactly~\eqref{eq:qw2cov}:
its objective is linear in the covariance and cross-covariance
matrices, and its constraints express quantum realizability
of the corresponding joint Gaussian states.

\paragraph{From uncertainty saturation to state reconstruction.}

Theorem~\ref{thm:gaussian-barycenter-intro} guarantees a Gaussian
minimizer, but does not exclude additional non-Gaussian ones.
Even uniqueness of the optimal covariance would not resolve
this issue, because first and second moments do not, in general,
determine a quantum state. The next theorem rules out this
ambiguity when at least one prescribed Gaussian state is faithful.

\begin{theorem}[Global uniqueness with a faithful Gaussian input]
\label{thm:faithful-global-uniqueness-intro}
Under the assumptions of
Theorem~\ref{thm:gaussian-barycenter-intro}, suppose that at least one
input Gaussian state is faithful, that is, for some $s_0\in[N]$,
\[
\Sigma_{s_0}
+
\frac{i\hbar}{2}\Omega_m
\succ0.
\]
Then, for each $\tau\in\{\pm1\}$, the $2$-QW barycenter problem
\eqref{eq:gaussianqw} for the Gaussian inputs
$\hat{\boldsymbol\sigma}^G$ admits a unique barycenter among all
quantum states. This barycenter is Gaussian; its mean and covariance
are characterized by
Theorem~\ref{thm:gaussian-barycenter-intro}.
\end{theorem}

\paragraph{Proof strategy: covariance uniqueness and state reconstruction.}
The proof, developed in
Section~\ref{subsec:faithful-gaussian-global-uniqueness},
has two distinct steps.

\emph{Covariance uniqueness.}
Under the faithful-input hypothesis,
Proposition~\ref{prop:covdual} establishes strong duality and
attainment for the covariance problem~\eqref{eq:qw2cov}.
A fixed dual optimizer supplies a common complementary-slackness
relation for every primal optimizer. The dual constraints,
together with this relation, imply uniqueness of the barycenter
covariance and of the cross-covariances associated with each
faithful input. Attainment here concerns the finite-dimensional
covariance dual; it is distinct from the equality of optimal
values in the operator-potential duality of
Theorem~\ref{thm:energy-constrained-barycenter-duality}.
At this stage, however, only the optimal moments have been
determined.

\emph{State reconstruction.}
Fix a faithful input. After identifying the corresponding
bipartite system with the standard CCR representation, the
optimal couplings with this input satisfy
\[
\left(
\Sigma_{\hat\pi}+\frac{i\hbar}{2}\Omega_{2m}
\right)Z=0,
\qquad Z\succeq0,
\]
for a matrix $Z$ obtained from the fixed covariance-dual
optimizer. We prove that, at fixed first moments, all
finite-second-moment states satisfying this relation can be
represented as
\[
\hat\pi
=
\rev{W_{\mathrm{CCR}}\!\left(
\frac{\Omega_{2m}\mathbf m}{\hbar}
\right)\hat U_S\bigl(
|0\rangle\langle0|^{\otimes q}\otimes\hat\eta
\bigr)\hat U_S^\dagger
W_{\mathrm{CCR}}\!\left(
\frac{\Omega_{2m}\mathbf m}{\hbar}
\right)^\dagger}
=:\Phi_G(\hat\eta),
\]
where \rev{$\mathbf m\in\R^{4m}$ is the fixed mean vector,
$W_{\mathrm{CCR}}(\Omega_{2m}\mathbf m/\hbar)$ is the corresponding
Weyl displacement, and $\hat U_S$ is a fixed Gaussian unitary for the
entire class}, and $\hat\eta$ has zero mean and finite
canonical second moments. The slackness relation forces
certain modes, in these coordinates, to be in the vacuum;
their purity makes them factorize from the residual system.
The phase-space transformation matrix of $\Phi_G$ satisfies
\[
\rev{\Ran(\mathbb X_\Phi)}=\ker(\Real Z).
\]

A Gaussian channel $\Psi_G$ whose phase-space transformation
matrix satisfies
\[
\ker(\mathbb X_\Psi)\cap\ker(\Real Z)=\{0\}
\]
is injective on this fixed-mean slackness class. In addition,
a Gaussian output forces the original state to be Gaussian.
For the certificate obtained from a faithful input, the
partial trace onto the prescribed input satisfies this
injectivity condition.

Consequently, an optimal coupling and its Gaussian projection
have the same first moments, satisfy the same slackness
relation, and have the same prescribed Gaussian marginal.
They must therefore coincide. The optimal coupling, and hence
its barycentric marginal, is Gaussian. Covariance uniqueness
then identifies the barycenter uniquely. The argument uses
the \emph{full prescribed marginal}, not merely its moments.
The representation and reconstruction results apply beyond
optimal transport, to arbitrary finite-second-moment states
satisfying the stated conditions. Here reconstruction means
unique determination within the slackness class, not the
existence of a physical inverse channel.

\noindent
\textit{Pure inputs and nonuniqueness at the boundary.}
Faithfulness is sufficient but not necessary for uniqueness.
If all inputs are pure and have finite canonical second
moments, the unique barycenter is the coherent state with
mean $\overline{\mathbf m}=\sum_{s=1}^N\alpha_s\mathbf m_s$, even when the inputs are not Gaussian
(Lemma~\ref{lem:pure-input-coherent-barycenter}). In contrast, partially pure Gaussian inputs,
with both pure and mixed modes in their Williamson decomposition,
may admit a continuum of Gaussian barycenters and additional
non-Gaussian ones. Example~\ref{ex:nonfaithful-boundary} exhibits distinct Gaussian
and non-Gaussian barycenters with identical first and second
moments. Consequently, even a specified optimal covariance
need not determine the minimizing state when the
faithful-input hypothesis is absent.

\paragraph{\rev{Semiclassical limit.}} We next describe the semiclassical limit for general Gaussian inputs. Fix $\mathbf m_s\in\R^{2m}$ and
$\Sigma_s\in\cM_{2m,>}(\R)$, $s\in[N]$, independently of $\hbar$.
The classical quadratic Wasserstein barycenter of the Gaussian measures $\{\mathcal N(\mathbf m_s,\Sigma_s)\}_{s\in[N]}$ is again Gaussian
\cite{AguCar}, with mean $\overline{\mathbf m}
=
\sum_{s=1}^N\alpha_s\mathbf m_s$ and covariance $\overline\Sigma\in\cM_{2m,>}(\R)$, uniquely
characterized by $\overline\Sigma=\sum_{s=1}^N
\alpha_s
\left(
\overline\Sigma^{1/2}
\Sigma_s
\overline\Sigma^{1/2}
\right)^{1/2}$.
\begin{theorem}[Semiclassical limit of Gaussian $2$-QW barycenters]
\label{thm:gaussian-semiclassical-intro}
For sufficiently small $\hbar>0$, let
$\hat\sigma_s^\hbar$ be the Gaussian state with mean
$\mathbf m_s$ and covariance $\Sigma_s$, $s\in[N]$.
For each $\tau\in\{\pm1\}$, consider the corresponding $2$-QW
barycenter problem
$$
\min_{\hat\rho\in\cS(L_2(\R^m))}
\sum_{s=1}^N
\alpha_s\,
\QW_{2,\tau}^2(\hat\rho,\hat\sigma_s^\hbar).
$$
By Theorem~\ref{thm:faithful-global-uniqueness-intro}, its barycenter
$\hat{\overline\rho}_\tau^\hbar$ is uniquely defined and Gaussian,
with mean $\overline{\mathbf m}$. Writing its covariance as
$\overline\Sigma_\tau(\hbar)$, one has 
$\overline\Sigma_\tau(\hbar)
\rightarrow
\overline\Sigma$ as $\hbar\to0^+$.

Consequently, for both the quantum-state and quantum-channel formulations, the Wigner distributions
$\mathbb W_{\hat{\overline\rho}_\tau^\hbar}(\mathbf z)\,d\mathbf z$ converge weakly to
$\mathcal N(\overline{\mathbf m},\overline\Sigma)$, the classical
$2$-Wasserstein barycenter of $\{\mathcal N(\mathbf m_s,\Sigma_s)\}_{s\in[N]}$ with weights $\boldsymbol{\alpha}$.
\end{theorem}

\medskip
\noindent
\textbf{Isotropic thermal states.}
We next specialize the input to displaced thermal Gaussian states $\hat{\boldsymbol\sigma}^{\mathrm{th}}
=(\hat\sigma_s^{\mathrm{th}})_{s\in[N]}$, with mean vectors
$\mathbf m_s$ and covariance matrices
\[
\Sigma_s=w_s\I_{2m},
\qquad
w_s\geq\frac{\hbar}{2}.
\]
See Section~\ref{subsec:thermal}.

\begin{corollary}[Unique $2$-QW barycenters for isotropic thermal states]
\label{cor:thermal-barycenter-intro}
For each $\tau\in\{\pm1\}$, the \rev{$2$}-QW barycenter problem
\begin{equation}\label{eq:thermalqw}
\min_{\hat\rho\in\cS(L_2(\mathbb R^m))}
\sum_{s=1}^N \alpha_s\,\QW_{2,\tau}^2(\hat\rho,\hat\sigma_s^{\mathrm{th}})
\end{equation}
admits a unique barycenter
$\hat{\overline\rho}_\tau$. It is a displaced isotropic thermal
state with mean $\overline{\mathbf m}=\sum_{s=1}^N\alpha_s\mathbf m_s$ and covariance
$\overline\Sigma_\tau=\overline w_\tau\I_{2m}$.

For the quantum-state formulation,
\[
\overline w_{+1}
=
\left(
\sum_{s=1}^N
\alpha_s\sqrt{w_s-\frac{\hbar}{2}}
\right)^2
+
\frac{\hbar}{2}.
\]
For the quantum-channel formulation,
\[
\{\overline w_{-1}\}
=
\operatorname{argmin}_{w\geq\hbar/2}
\left\{
\begin{aligned}
&\sum_{w_s<w}
\alpha_s
\left[
\sqrt{w+\frac{\hbar}{2}}
-
\sqrt{w_s-\frac{\hbar}{2}}
\right]^2\\
&\quad+
\sum_{w_s\geq w}
\alpha_s
\left[
\sqrt{w_s+\frac{\hbar}{2}}
-
\sqrt{w-\frac{\hbar}{2}}
\right]^2
\end{aligned}
\right\}.
\]
\end{corollary}
In particular, Theorem~\ref{thm:gaussian-semiclassical-intro}
recovers the classical isotropic Gaussian Wasserstein barycenter in the limit $\hbar\to0^+$.
 \section{Preliminaries}\label{sec:preliminary}
\subsection{Operator-theoretic notation}\label{subsec:op-space}
Let $\mathbb{R}$ and $\mathbb{C}$ denote the fields of real numbers and complex numbers, respectively. Let $\mathbb{N}$ denote the set of natural numbers. For $N\in\N$, we set $[N]:=\{1,\ldots,N\}$. For $m\in\mathbb N$ and $\mathbb F\in\{\mathbb R,\mathbb C\}$, let $\cM_m(\mathbb F)$ denote the space of $m\times m$ matrices over $\mathbb F$. We denote by $\cM_{m,\mathrm{sym}}(\mathbb R)$ the subspace of real symmetric matrices and by $\cM_{m,\mathrm{sa}}(\mathbb C)$ the subspace of Hermitian matrices. Their positive semidefinite and positive definite cones are denoted by $\cM_{m,\geq}(\mathbb F)$, $\cM_{m,>}(\mathbb F)$. We write $\I_m$ for the $m\times m$ identity matrix. Matrix spaces are regarded as real inner-product spaces equipped with the real Frobenius pairing $\langle A,B\rangle_F:=\Real\Tr(A^\dagger B)$, and finite product spaces are equipped with the corresponding sum pairing.

Let $\cH$ be a separable Hilbert space. We denote by $\I_{\cH}$ the identity operator on $\cH$, by $\cU(\cH)$ the unitary operator group on $\cH$, by $\cB(\cH)$ the $C^*$-algebra
of bounded linear operators on $\cH$, by $\cK(\cH)$ the compact operators, and by
$\cT^{(1)}(\cH)$ the trace-class operators. Their \rev{self-adjoint parts} and positive
semidefinite cones are denoted by the subscripts $S$ and $\geq$, respectively. The operator norm and the trace norm are denoted by $\|\cdot\|_{\mathrm{op}}$ and $\|\cdot\|_1$, respectively.

There are canonical isometric isomorphisms $(\cT^{(1)}(\cH))^* \cong \cB(\cH)$ with pairing $\langle T,B\rangle := \Tr(TB)$ for
$T \in \cT^{(1)}(\cH)$ and $B \in \cB(\cH)$;
see, e.g., \cite[Thm.~2.1.7]{peterson2013notes}. Restricting the trace pairing to the self-adjoint parts and
viewing $\cT^{(1)}_S(\cH)$ and $\cB_S(\cH)$ as real Banach spaces
yields the canonical isometric identification $(\cT^{(1)}_S(\cH))^* \cong \cB_S(\cH)$ as an immediate corollary.

We will also use the real dual space $\cB_S(\cH)^*$. For $\Lambda\in\cB_S(\cH)^*$,
\[
\|\Lambda\|_* = \sup_{\|B\|_{\mathrm{op}}=1} |\Lambda(B)|,
\qquad
\Lambda\ge 0 \ \Longleftrightarrow\ \Lambda(B)\ge 0 \text{ for all } B\in\cB_S(\cH),\ B\ge 0.
\]

The set of density operators, or quantum states, is
\[
\cS(\cH):=\bigl\{\rho\in\cT^{(1)}(\cH):\ \rho\ge 0,\ \Tr\rho=1\bigr\}.
\]
A state $\rho\in\cS(\cH)$ is called pure if it has rank one, and
faithful if $\ker\rho=\{0\}$.

A densely defined operator \(A\) on \(\cH\) is symmetric if $\langle A\psi,\varphi\rangle=\langle \psi,A\varphi\rangle$ for all $\psi,\varphi\in\Dom(A)$ and is self-adjoint if $A=A^*$. It is essentially self-adjoint if its closure is self-adjoint. In this case, the closure is the
unique self-adjoint extension of \(A\); see, e.g., \cite[Sec.~3]{Schmudgen2012}.

Let $\cL_S(\cH)$ denote the class of densely defined self-adjoint operators
on $\cH$. We say that $A\in\cL_S(\cH)$ is \emph{lower semibounded} if there
exists $m\in\R$ such that
\[
    \langle x,Ax\rangle\ge m\|x\|^2,
    \qquad x\in\Dom(A).
\]
We denote by $\cL_\geq(\cH)$ the subclass of $\cL_S(\cH)$ consisting of
positive semidefinite operators, i.e. those $A\in\cL_S(\cH)$ satisfying
\[
    \langle x,Ax\rangle\ge 0,
    \qquad x\in\Dom(A).
\]

For lower semibounded $A,B\in\cL_S(\cH)$ and any common lower bound $c$ of $A$ and
$B$, set $\widetilde A:=A-c\I_\cH$ and $\widetilde B:=B-c\I_\cH$ so that
$\widetilde A,\widetilde B\ge 0$. We write $A\ge B$ if
\[
\Dom(\widetilde A^{1/2})\subset \Dom(\widetilde B^{1/2})
\quad\text{and}\quad
\|\widetilde A^{1/2}x\|^2\ge \|\widetilde B^{1/2}x\|^2
\quad\text{for all }x\in\Dom(\widetilde A^{1/2}).
\]
This definition is independent of the choice of $c$ and coincides with the
usual Loewner order when $A,B\in\cB_S(\cH)$.

Finally, let us recall the definition of compact resolvent; see, e.g., \cite{Schmudgen2012}.

\begin{deff}[\rev{Compact} resolvent]
\label{def:compact-resolvent}
Let $\cH$ be a separable Hilbert space, and let $H\in\cL_S(\cH)$ be a self-adjoint
operator on $\cH$ and $\mathrm{spec}(H)$ be the spectrum of $H$.
The resolvent set of $H$ is then $\mathrm{res}(H):=\mathbb C\setminus\mathrm{spec}(H)$. For $z\in\mathrm{res}(H)$, the resolvent of $H$ at $z$ is the bounded operator
\[
    R_z(H):=(H-z\I_{\cH})^{-1}.
\]
We say that $H$ has compact resolvent if $R_z(H)\in\cK(\cH)$ is compact for some,
equivalently for every, $z\in\mathrm{res}(H)$.
\end{deff}

\subsection{Spectral integral against a positive trace-class operator}
\label{subsec:spectral-trunc}

Let $A\in\cL_S(\cH)$ be a possibly unbounded self-adjoint operator on a separable Hilbert
space $\cH$, and let $E^A$ denote its spectral measure.
For $\rho\in\cT^{(1)}_{\geq}(\cH)$, the map $\mu_\rho^A(B)\ :=\ \Tr\!\big(E^A(B)\rho\big)$ for all Borel $B$ defines a finite positive Borel measure on $\R$ with total mass
$\mu_\rho^A(\R)=\Tr(\rho)$. We use this scalar spectral measure to define the extended expectation of \(A\) in $\rho$; see \cite[Sec.~4]{Schmudgen2012}.

\begin{deff}[Spectral integral against a \rev{positive trace-class operator}]
\label{def:spectral-integral-traceclass}
Let $A\in\cL_S(\cH)$ be self-adjoint and let $\rho\in\cT^{(1)}_{\geq}(\cH)$.
We define the spectral integral of $A$ against $\rho$ by
\[
\Tr(A\rho)
\ :=\
\int_{\R} \lambda\, d\mu_\rho^A(\lambda)
\ \in\ [-\infty,+\infty],
\]
whenever the integral is well-defined in the extended sense, i.e.\ whenever at
least one of the positive/negative parts is finite:
\[
\int_{\R} \lambda^+\, d\mu_\rho^A(\lambda)<+\infty
\quad \text{or}\quad
\int_{\R} \lambda^-\, d\mu_\rho^A(\lambda)<+\infty,
\]
where $\lambda^\pm:=\max\{\pm\lambda,0\}$.
\end{deff}

\begin{prop}[Truncation and monotone convergence for positive observables]
\label{prop:spectral-trunc-positive}
Let $A\in\cL_{\geq}(\cH)$ be self-adjoint and positive semidefinite, and let
$\rho\in\cS(\cH)$.
For $n\in\N$, define the bounded spectral truncations
\[
A^{(n)}
:=
\int_{[0,n]} \lambda\, dE^A(\lambda)
\ \in\ \cB_S(\cH).
\]
Then $0\le A^{(n)}\le A^{(n+1)}\le A$, and
\[
\Tr(A\rho)
=
\lim_{n\to\infty}\Tr\!\big(\,A^{(n)}\rho\big)
=
\sup_{n\in\N}\Tr\!\big(\,A^{(n)}\rho\big)
\ \in\ [0,+\infty],
\]
with the convention that the right-hand side may be $+\infty$.
Equivalently,
\[
\Tr(A\rho )
=
\sup\Bigl\{\Tr(B\rho ):\ B\in\cB_S(\cH),\ 0\le B\le A\Bigr\}.
\]
\end{prop}

\begin{proof}
By construction, $\mu_\rho^A$ is a finite positive measure on $\R$ and
\[
\Tr(A\rho )
=
\int_{\R}\lambda\, d\mu_\rho^A(\lambda)
=
\int_{[0,\infty)}\lambda\, d\mu_\rho^A(\lambda)\in[0,+\infty].
\]
Moreover, $\Tr\!\big(\,A^{(n)}\rho\big)
=
\int_{[0,n]} \lambda\, d\mu_\rho^A(\lambda)$, hence $\Tr( A^{(n)}\rho)\uparrow \Tr(A\rho)$ by monotone convergence for
nonnegative functions.
Finally, for any $B\in\cB_S(\cH)$ with $0\le B\le A$, one has
$\Tr(B\rho)\le \Tr(A\rho)$, and taking $B=A^{(n)}$ yields the supremum
representation.
\end{proof}
\subsection{Tensor products, partial traces, and quantum couplings}
\label{subsec:tensor-partial-trace-coupling}

Throughout the paper, \(\otimes\) denotes the Hilbert-space tensor product. If
\(\cH_1,\cH_2\) are separable Hilbert spaces and \(A\in \cB(\cH_1)\), \(B\in \cB(\cH_2)\), then
\(A\otimes B\in \cB(\cH_1\otimes \cH_2)\) is defined on simple tensors by
\[
(A\otimes B)(\psi\otimes\varphi)
=
A\psi\otimes B\varphi,
\qquad
\psi\in \cH_1,\quad \varphi\in \cH_2,
\]
and extended by linearity and continuity. Also, we write
\[
A\oplus B
:=
A\otimes \I_{\cH_2}+\I_{\cH_1}\otimes B
\in \cB(\cH_1\otimes \cH_2)
\]
for the Kronecker sum. This notation is used in the dual constraints
\(U\oplus V\le C\).

We reserve the symbol \(\oplus_D\) for block-diagonal direct sums of matrices. Thus, if
\(M_1\in\cM_{m_1}(\mathbb R)\) and \(M_2\in\cM_{m_2}(\mathbb R)\), then
\[
M_1\oplus_D M_2
:=
\begin{pmatrix}
M_1 & 0\\
0 & M_2
\end{pmatrix}
\in\cM_{m_1+m_2}(\mathbb R).
\]
In particular, for the \(m\)-mode symplectic matrix we write
\[
\Omega_m
=
\bigoplus_{k=1}^m{}_D\,\omega,
\]
where \(\omega=\smat{0&1\\-1&0}\) is the one-mode symplectic matrix.

\begin{deff}[Partial trace]\label{def:partialtrace}
Let $\cH_1,\cH_2$ be separable Hilbert spaces and $T\in\cT^{(1)}(\cH_1\otimes\cH_2)$.
The partial traces $\Tr_{\cH_2}(T)\in\cT^{(1)}(\cH_1)$ and
$\Tr_{\cH_1}(T)\in\cT^{(1)}(\cH_2)$ are the unique trace-class operators such that
\begin{equation}\label{eq:partialtrace}
\Tr\!\bigl(B_1\,\Tr_{\cH_2}(T)\bigr)=\Tr\!\bigl((B_1\otimes \I_{\cH_2})\,T\bigr),
\qquad
\Tr\!\bigl(B_2\,\Tr_{\cH_1}(T)\bigr)=\Tr\!\bigl((\I_{\cH_1}\otimes B_2)\,T\bigr),
\end{equation}
for all $B_1\in\cB(\cH_1)$ and $B_2\in\cB(\cH_2)$.
\end{deff}

\begin{lemma}[Trace-norm continuity of partial traces]\label{lem:pt-continuity}
For $T\in\cT^{(1)}(\cH_1\otimes\cH_2)$ one has
$\|\Tr_{\cH_2}(T)\|_1\le \|T\|_1$ and $\|\Tr_{\cH_1}(T)\|_1\le \|T\|_1$.
\end{lemma}

\begin{proof}
We prove the first inequality. By duality of $\cT^{(1)}(\cH_1)$ and $\cB(\cH_1)$,
\[
\begin{aligned}
\|\Tr_{\cH_2}(T)\|_1
&=\sup_{\|A\|_{\mathrm{op}}\le 1}
\big|\Tr\!\big(A\,\Tr_{\cH_2}(T)\big)\big|\\
&=\sup_{\|A\|_{\mathrm{op}}\le 1}
\big|\Tr\!\big((A\otimes\I_{\cH_2})T\big)\big|\\
&\le \sup_{\|B\|_{\mathrm{op}}\le 1}|\Tr(BT)|=\|T\|_1.
\end{aligned}
\]
The second inequality is analogous.
\end{proof}

The following compatibility follows directly from the defining property
of the partial trace and the spectral calculus.

\begin{prop}[Compatibility with unbounded observables]
\label{prop:partialtrace-unbounded}
Let $A\in\cL_S(\cH_1)$ and $\pi\in\cS(\cH_1\otimes\cH_2)$. Then $\Tr\!\left(A(\Tr_{\cH_2}\pi)\right)$ is well defined if and only if $\Tr\!\left((A\otimes\I_{\cH_2})\pi\right)$ is well defined, and in this case 
\[
\Tr\!\left(A(\Tr_{\cH_2}\pi)\right)=
\Tr\!\left((A\otimes\I_{\cH_2})\pi\right).
\]
\end{prop}
\begin{proof}
By the spectral calculus, $E^{A\otimes\I_{\cH_2}}(B)=E^A(B)\otimes\I_{\cH_2}$ for every Borel set $B\subset\R$. Hence, by the defining property of
the partial trace,
\[
\begin{aligned}
\mu_{\pi}^{A\otimes\I_{\cH_2}}(B)=
\Tr\!\left[
E^{A\otimes\I_{\cH_2}}(B)\pi
\right]=
\Tr\!\left[
(E^A(B)\otimes\I_{\cH_2})\pi
\right]=
\Tr\!\left[E^A(B)
\Tr_{\cH_2}\pi
\right]
=
\mu_{\Tr_{\cH_2}\pi}^{A}(B).
\end{aligned}
\]
Thus the two scalar spectral measures coincide. Therefore their
positive and negative spectral integrals are finite simultaneously,
and, whenever the corresponding extended expectations are well defined,
\[
\Tr\!\left(A\Tr_{\cH_2}\pi\right)
=
\int_{\R}\lambda\,d\mu_{\Tr_{\cH_2}\pi}^{A}(\lambda)
=
\int_{\R}\lambda\,
d\mu_{\pi}^{A\otimes\I_{\cH_2}}(\lambda)
=
\Tr\!\left((A\otimes\I_{\cH_2})\pi\right).
\]
\end{proof}

\begin{deff}[Quantum couplings]\label{def:quantumcouplings}
Let $\rho\in\cS(\cH_1)$ and $\sigma\in\cS(\cH_2)$.
A \emph{coupling} of $(\rho,\sigma)$ is a state $\pi\in\cS(\cH_1\otimes\cH_2)$ such that $\Tr_{\cH_2}(\pi)=\rho$ and $\Tr_{\cH_1}(\pi)=\sigma$. We denote the set of all such couplings by
\begin{equation}\label{def:quantumcoupling}
\Gamma(\rho,\sigma):=\bigl\{\pi\in\cS(\cH_1\otimes\cH_2):\ \Tr_{\cH_2}\pi=\rho,\ \Tr_{\cH_1}\pi=\sigma\bigr\}.
\end{equation}
\end{deff}

\subsection{Regular Weyl systems and quadratic Hamiltonians}\label{pre:gauss}
Let \(m\in\N\) and let $\Theta\in\cM_{2m}(\R)$ satisfy
$\Theta^\top=-\Theta$ and $\det\Theta\neq0$. A
\emph{regular \(\Theta\)-Weyl system} on a Hilbert space
\(\mathcal H\) is a strongly continuous map $
  W:\mathbb R^{2m}\to \mathcal U(\mathcal H)$ satisfying the Weyl relations
\[
  W(\xi)W(\eta)
  =
  e^{-\frac{i}{2}\xi^\top\Theta\eta}W(\xi+\eta),
  \qquad
  \xi,\eta\in\mathbb R^{2m}.
\]
Here strong continuity means that
\(\xi\mapsto W(\xi)\psi\) is continuous for every \(\psi\in\mathcal H\).
The operators \(W(\xi)\) are called the Weyl operators of the system. See \cite[Sec.~3.1]{DerezinskiCCR} for the standard Weyl-form convention for
CCR representations, and \cite[Sec.~4]{AraiGCCR} for generalized CCR systems
associated with a general skew-symmetric commutation matrix \(\Theta\).

The Weyl relations imply \(W(0)=\I_{\mathcal H}\) and
\(W(\xi)^*=W(-\xi)\). For every
\(\xi\in\mathbb R^{2m}\), the map $t\longmapsto W(t\xi)$ is a strongly continuous one-parameter unitary group. Hence, there exists a unique self-adjoint operator \(R(\xi)\) such that $W(t\xi)=e^{itR(\xi)}$ for $t\in\mathbb R$. We call \(R(\xi)\) the field operator associated with \(\xi\); see \cite[Sec.~3.3]{DerezinskiCCR}.

Let \(\{e_j\}_{j\in[2m]}\) denote the standard basis of
\(\mathbb R^{2m}\), and set $R_j:=R(e_j)$ for $j\in[2m]$. We call $\mathbf{R}=(R_1,\ldots,R_{2m})$ the associated canonical generators. Thus \(\mathbf R\) is a Weyl representation of the \(\Theta\)-GCCR in the sense of
Arai~\cite[Def.~11]{AraiGCCR}.

Equivalently, \(W\) lifts to a strongly continuous unitary representation of the Heisenberg--Weyl group associated with \(\Theta\). We therefore use the standard smooth-vector theory for unitary Lie-group representations; see \cite[Props.~1.7.3, 1.7.6--1.7.7]{FischerRuzhansky2016}.
\[
  \mathcal H_R^\infty
  :=
  \left\{
    \psi\in\mathcal H:
    \mathbb R^{2m}\ni\xi\longmapsto W(\xi)\psi\in\mathcal H
    \text{ is } C^\infty
  \right\}.
\]
It is dense in \(\mathcal H\) and invariant under the Weyl operators and under
the derived operators \(R_j\). Moreover, for every
\(\psi\in\mathcal H_R^\infty\) and
\(\xi=(\xi_1,\ldots,\xi_{2m})\in\mathbb R^{2m}\), $R(\xi)\psi
  =
  \sum_{j=1}^{2m}\xi_jR_j\psi$. Thus every real linear combination of the canonical generators is well defined
on \(\mathcal H_R^\infty\).

On \(\mathcal H_R^\infty\), the canonical generators satisfy the Heisenberg
commutation relations
\[
  [R(\xi),R(\eta)]\psi
  =
  i\,\xi^\top\Theta\eta\,\psi,
  \qquad
  \xi,\eta\in\mathbb R^{2m},\quad
  \psi\in\mathcal H_R^\infty.
\]
Equivalently, $[R_i,R_j]\psi
  =
  i\Theta_{ij}\psi$ for all $i,j\in[2m]$ and $\psi\in\mathcal H_R^\infty$.
\begin{proposition}[Quadratic Hamiltonians generated by a regular Weyl system]
\label{prop:quadratic-hamiltonian-esa}
Let \(\Theta\in\cM_{2m}(\mathbb R)\) be a nondegenerate skew-symmetric matrix, $W:\mathbb R^{2m}\to\mathcal U(\mathcal H)$ be a regular \(\Theta\)-Weyl system on a Hilbert space \(\mathcal H\), and \(\mathbf{R}=(R_1,\ldots,R_{2m})\) be its canonical
generators.

Let \(G\in\cM_{2m,\mathrm{sym}}(\R)\) be real symmetric, \(\mathbf{a}\in\mathbb R^{2m}\), and \(c\in\mathbb R\).
Define, on the common smooth domain \(\mathcal H_R^\infty\),
\[
  H_{G,\mathbf{a},c}^{(0)}
  :=
  \frac12\sum_{j,k=1}^{2m}G_{jk}R_jR_k
  +\sum_{j=1}^{2m}a_jR_j
  +c\I_\cH=\frac12\mathbf{R}^\top G\mathbf{R}+\mathbf{a}^\top\mathbf{R}
  +c\I_\cH.
\]
Then \(H_{G,\mathbf{a},c}^{(0)}\) is a well-defined symmetric operator on
\(\mathcal H_R^\infty\) and is essentially self-adjoint on
\(\mathcal H_R^\infty\).
\end{proposition}

See the proof in Appendix~\ref{app:quadratic-hamiltonians}.

\begin{remark}[Closure convention for quadratic Hamiltonians] \label{rem:quadratic-hamiltonian-closure-convention} In what follows, whenever a Hamiltonian can be written as \[ H_{G,\mathbf{a},c} = \frac12\sum_{j,k=1}^{2m}G_{jk}R_jR_k +\sum_{j=1}^{2m}a_jR_j +c\I_\cH, \] it is understood as the unique self-adjoint extension, equivalently the closure, of the operator \(H_{G,\mathbf{a},c}^{(0)}\) initially defined on \(\mathcal H_R^\infty\). 

In particular, all quadrature-generated costs and one-body square-sums defined in Section~\ref{intro:statementofresults} are understood in this closure sense. Thus the operators are well-defined self-adjoint operators.
\end{remark}

We also introduce tensor products of Weyl systems. Let \(W_a:\mathbb R^{2m}\to \cU(\cH_a)\), \(a=1,2\), be regular \(\Theta_a\)-Weyl systems with canonical generators \(\{R_{j}\}_{j\in[2m]}\) on \(\cH_1\) and \(S_1,\ldots,S_{2m}\) on \(\cH_2\), respectively. Then \[ W_1\otimes W_2: \mathbb R^{2m}\oplus\mathbb R^{2m} \to \cU(\cH_1\otimes \cH_2), \qquad (W_1\otimes W_2)(\xi,\eta) := W_1(\xi)\otimes W_2(\eta), \] is a regular \((\Theta_1\oplus_D\Theta_2)\)-Weyl system. Its canonical generators are $\{R_j\otimes \I_{\cH_2},\I_{\cH_1}\otimes S_j\}_{j\in[2m]}$.

\paragraph{The standard CCR representation and the Stone--von Neumann theorem.}
\label{subsec:standard-ccr-representation}

The model example of a regular Weyl system is the \(m\)-degree-of-freedom canonical commutation relation (CCR) system on \(L_2(\mathbb R^m)\), which we call the standard CCR representation. We use this model as the reference system in the Stone--von Neumann theorem. For this representation, see \cite[Sec.~12.1.2]{Holevo2019QSCI} and \cite[Sec.~5.2]{holevo2011probabilistic}.

Let $m\in\N$. The Hilbert space for \(m\) degrees of freedom is $L_2(\mathbb R^m)$. For \(\psi\in\mathscr S(\mathbb R^m)\subset L_2(\mathbb R^m)\), define
\[
(\hat{Q}_j\psi)(x)=x_j\psi(x),
\qquad
(\hat{P}_j\psi)(x)=-i\hbar\partial_{x_j}\psi(x),
\qquad j\in[m],
\]
where $\hbar$ is the reduced Planck constant. We collect the canonical variables as $\hat{\mathbf R}$ and use the hat symbol to distinguish them from the canonical generators for the general Weyl system:
\[
\hat{\mathbf R}=(\hat{R}_1,\dots,\hat{R}_{2m})
:=
(\hat{Q}_1,\hat{P}_1,\dots,\hat{Q}_m,\hat{P}_m).
\]
Then, on \(\mathscr S(\mathbb R^m)\), $[\hat R_a,\hat R_b]=i\hbar(\Omega_m)_{ab}$, where $\Omega_m=\bigoplus_{j\in[m]}{}_D\begin{pmatrix}
    0&1\\
    -1&0
\end{pmatrix}$ is called the symplectic matrix for $m$ degrees of freedom. Equivalently,
\[
[\hat Q_j,\hat Q_k]=0,
\qquad
[\hat P_j,\hat P_k]=0,
\qquad
[\hat Q_j,\hat P_k]=i\hbar\delta_{jk}\I.
\]

The associated standard Weyl system is obtained by exponentiating real linear combinations of the canonical variables: 
\begin{equation}\label{eq:weyl-ccr}
W_{\mathrm{CCR}}:\mathbb R^{2m}\to \cU(L_2(\mathbb R^m)), \qquad W_{\mathrm{CCR}}(\xi) := \exp\!\bigl(i\xi^\top\hat{\mathbf R}\bigr),
\end{equation}
which is a $\hbar\Omega_m$-Weyl system on $L_2(\R^m)$. The canonical variables are recovered by $W_{\mathrm{CCR}}(t e_a) = e^{it\hat R_a}$ for $a\in[2m]$. This is the standard regular Weyl system associated with the usual CCR; see \cite[Sec.~3.1]{DerezinskiCCR} for the Weyl-form convention.

The following form of the Stone--von Neumann theorem says that the
standard Weyl system is universal among regular Weyl systems over finite-dimensional phase spaces, up to multiplicity.
\begin{theorem}[Stone--von Neumann theorem with multiplicity]\label{theo:stone}
Let \(\Theta\in\cM_{2m}(\mathbb R)\) be a nondegenerate skew-symmetric matrix, $W:\mathbb R^{2m}\to\mathcal U(\mathcal H)$ be a regular \(\Theta\)-Weyl system on a Hilbert space \(\mathcal H\), and \(\mathbf{R}=(R_1,\ldots,R_{2m})\) be its canonical generators as defined above. Choose \(L\in \cM_{2m}(\mathbb R)\) such that
\[
L^\top(\hbar\Omega_m)L=\Theta.
\]
Then there exist a Hilbert space \(\mathcal K\) and a unitary $U:L_2(\mathbb R^m)\otimes\mathcal K\to\mathcal H$ such that
\[
U^\dagger W(\xi)U
=
W_{\mathrm{CCR}}(L\xi)\otimes \I_{\cK}=e^{i\xi^\top L^\top\hat{\mathbf R}}\otimes \I_{\cK},\quad
U^\dagger R(\xi)U
=(\xi^\top L^\top\hat{\mathbf R})\otimes \I_{\mathcal K}.
\]
The representation is irreducible if and only if \(\mathcal K=\mathbb C\).
\end{theorem}

\begin{remark}[References for Theorem~\ref{theo:stone}]
\label{rem:stone-reference}
Theorem~\ref{theo:stone} is the finite-dimensional Stone--von Neumann theorem with
multiplicity, written in coordinates adapted to the skew-symmetric matrix \(\Theta\). The
linear change of variables reduces the \(\Theta\)-Weyl relations to the standard
\((\hbar\Omega_m)\)-Weyl relations. We refer to
\cite[Props.~15--16, Thm.~19, Cor.~20]{AraiGCCR} for generalized CCR systems and to
\cite[Thm.~4]{DerezinskiCCR} for the regular finite-dimensional CCR Stone--von Neumann
theorem with multiplicity.
\end{remark}
\begin{coro}[Compact resolvent of the square-sum]
\label{coro:nelson-square-compact-resolvent}
Let \(\Theta\in\cM_{2m}(\mathbb R)\) be a nondegenerate skew-symmetric matrix, $W:\mathbb R^{2m}\to\mathcal U(\mathcal H)$ be an irreducible regular \(\Theta\)-Weyl system on a Hilbert space \(\mathcal H\), and \(\mathbf{R}=(R_1,\ldots,R_{2m})\) be its canonical generators. Let $\mathsf H_{\mathbf R}
  :=\sum_{j=1}^{2m}R_j^2$ be the one-body square-sum on $\cH$. Then \(\mathsf H_{\mathbf R}\in\cL_S(\cH)\) is a nonnegative self-adjoint operator
with compact resolvent.
\end{coro}
\begin{proof}
By Theorem~\ref{theo:stone} and irreducibility, the multiplicity space is
\(\mathbb C\). Hence \(\mathsf H_{\mathbf R}\) is unitarily equivalent to the operator
\[
  \sum_{j=1}^{2m}
  \bigl((L^\top\hat{\mathbf R})_j\bigr)^2
  =
  \hat{\mathbf R}^\top LL^\top\hat{\mathbf R},
\]
where \(L\in \cM_{2m}(\mathbb R)\) satisfies
\(L^\top(\hbar\Omega_m)L=\Theta\). Since \(LL^\top\) is positive
definite, this is an elliptic positive quadratic Hamiltonian in the standard canonical variables. By the standard spectral theory of positive quadratic Hamiltonians, it has
discrete spectrum with finite multiplicities, equivalently compact resolvent;
see \cite[Sec.~26]{Shubin2001}. Hence \(\mathsf H_{\mathbf R}\) has compact
resolvent.
\end{proof}
 \subsection{Dual Hilbert spaces and the transpose CCR representation}
\label{subsec:dual-transpose-CCR}
Let \(\cH\) be a separable Hilbert space, and \(\cH^*\) denote its continuous dual. This subsection fixes the transpose convention on dual Hilbert spaces. The motivation comes from the channel-oriented transport framework of
\cite{DPaTre19}, where quantum transport plans are quantum channels and are in one-to-one correspondence with quantum couplings on $\cH\otimes\cH^*$; see \cite[Prop.~2]{DPaTre19}.

We use Dirac notation for elements of \(\cH^*\), writing them as bras \(\langle\varphi|\). For
\(B\in \cB(\cH)\), define its transpose \(B^\top\in \cB(\cH^*)\) by 
\[
B^\top\langle\varphi|=\langle\varphi|B,
\]
following the convention of \cite[Def.~1]{DPaTre19}. This definition is canonical and does not require choosing a basis. It reverses products $(AB)^\top=B^\top A^\top$ for $A,B\in \cB(\cH)$. Consequently,
\[
[A^\top,B^\top]=-[A,B]^\top.
\]

Let \(\cH=L_2(\mathbb R^m)\) and $W_{\mathrm{CCR}}$ defined in \eqref{eq:weyl-ccr} be the associated regular \((\hbar\Omega_m)\)-Weyl system. We define the transpose Weyl system $W_{\mathrm{CCR}}^\top$ on \(\cH^*=L_2(\mathbb R^m)^*\) by
\[
W_{\mathrm{CCR}}^\top(\xi)
:=
W_{\mathrm{CCR}}(\xi)^\top,
\qquad
\xi\in\mathbb R^{2m}.
\]
Since transposition reverses products, \(W_{\mathrm{CCR}}^\top\) satisfies
\[
W_{\mathrm{CCR}}^\top(\xi)W_{\mathrm{CCR}}^\top(\eta)
=
\exp\!\left(
\frac{i\hbar}{2}\xi^\top\Omega_m\eta
\right)
W_{\mathrm{CCR}}^\top(\xi+\eta).
\]
Equivalently, \(W_{\mathrm{CCR}}^\top\) is a regular \((-\hbar\Omega_m)\)-Weyl system. We denote the canonical generators of $W_{\mathrm{CCR}}^\top$ by 
\begin{equation}
\hat{\mathbf{R}}^\top
=
(\hat R_1^\top,\ldots,\hat R_{2m}^\top)
=
(\hat Q_1^\top,\hat P_1^\top,\ldots,\hat Q_m^\top,\hat P_m^\top),
\end{equation}
so that $W_{\mathrm{CCR}}^\top(te_j)=e^{it\hat R_j^\top}=(e^{it\hat R_j})^\top$ for $j\in[2m]$. In particular, the superscript $\top$ on $\hat R_j^\top$ is understood through the transposed Weyl system. The corresponding commutation relations are
\[
[\hat R_j^\top,\hat R_k^\top]
=
-i\hbar(\Omega_m)_{jk},
\qquad
j,k\in[2m].
\]
This is the \(\hbar\)-scaled version of the transpose-quadrature commutation relation of \cite[Eq.~(72)]{DPaTre19}.

For notational convenience, we encode the two second-factor conventions by a parameter
\(\tau\in\{\pm1\}\), which is consistent with \eqref{eq:tau-barycenter-intro}. Set
\begin{equation}\label{eq:tau}
L_2^{(\tau)}(\mathbb R^m)
:=
\begin{cases}
L_2(\mathbb R^m), & \tau=+1,\\
L_2(\mathbb R^m)^*, & \rev{\tau=-1}
\end{cases},\quad \hat{\mathbf R}^{(\tau)}
:=
\begin{cases}
\hat{\mathbf R}, & \tau=+1,\\
\hat{\mathbf R}^\top , & \rev{\tau=-1}
\end{cases}\rev{.}
\end{equation}

On the bipartite space $L_2(\mathbb R^m)\otimes L_2^{(\tau)}(\mathbb R^m)$
we use the unified bipartite canonical vector with the components indexed by
\begin{equation}\label{eq:bitau}
\hat{\mathbf R}_\tau
:=
\left(
\hat{\mathbf R}\otimes \I_{L_2^{(\tau)}(\mathbb R^m)},\,
\I_{L_2(\mathbb R^m)}\otimes \hat{\mathbf R}^{(\tau)}
\right)^\top,\quad
\hat R_{\tau,j}
=
\begin{cases}
\hat R_j\otimes \I_{L_2^{(\tau)}(\mathbb R^m)}, & 1\le j\le 2m,\\[1mm]
\I_{L_2(\mathbb R^m)}\otimes \hat R^{(\tau)}_{j-2m}, & 2m+1\le j\le 4m.
\end{cases}
\end{equation}
Then
\[
[\hat R_{\tau,j},\hat R_{\tau,k}]
=
i\hbar
\bigl(
\Omega_m\oplus_D(\tau\Omega_m)
\bigr)_{jk},
\qquad
j,k\in[4m].
\]

We finally fix a unitary identification of the transpose CCR representation with the standard CCR representation. Let $\{|\mathbf n\rangle:\mathbf n\in\N_0^m\}$ be the standard
$m$-mode Fock basis of $L_2(\R^m)$. Using this fixed basis, define the unitary
\[
J_m:L_2(\R^m)^*\to L_2(\R^m),
\qquad
J_m\langle\mathbf n|:=|\mathbf n\rangle.
\]
On the finite-particle subspace, $\hat Q_j$ has a real symmetric matrix
and $\hat P_j$ has a purely imaginary skew-symmetric matrix with
respect to the Fock basis. Hence
\[
J_m\hat Q_j^\top J_m^\dagger=\hat Q_j,
\qquad
J_m\hat P_j^\top J_m^\dagger=-\hat P_j,
\qquad j\in[m].
\]

For $\tau\in\{\pm1\}$, define
\[
J_{\tau,m}
:=
\begin{cases}
\I_{L_2(\R^m)},&\tau=+1,\\
J_m,&\tau=-1
\end{cases},
\qquad
F_{\tau,m}
:=
\begin{cases}
\I_{2m},&\tau=+1,\\
\displaystyle\bigoplus_{j=1}^m{}_D
\smat{1&0\\0&-1},&\tau=-1
\end{cases}.
\]

\begin{prop}[Unitary equivalence of the $\tau$-CCR and standard CCR representations]
\label{prop:tau-standard-CCR-identification}
We have
\begin{equation}
\bigl(\I_{L_2(\R^m)}\otimes J_{\tau,m}\bigr)^\dagger
\hat{\mathbf R}_{+1}
\bigl(\I_{L_2(\R^m)}\otimes J_{\tau,m}\bigr)
=
(\I_{2m}\oplus_DF_{\tau,m})
\hat{\mathbf R}_\tau,
\label{eq:tau-standard-canonical-identification}
\end{equation}
where the left unitary conjugation is understood
componentwise, while the right real matrices act on its phase-space index. In particular,
\begin{equation}
(\I_{2m}\oplus_DF_{\tau,m})
\bigl(
\Omega_m\oplus_D(\tau\Omega_m)
\bigr)
(\I_{2m}\oplus_DF_{\tau,m})^\top
=
\Omega_{2m}.
\label{eq:tau-standard-symplectic-identification}
\end{equation}
\end{prop}
\begin{proof}
Equation~\eqref{eq:tau-standard-canonical-identification} follows
componentwise from the definition of $J_{\tau,m}$ and the identities
for $\hat Q_j^\top$ and $\hat P_j^\top$ above. Moreover,
$F_{\tau,m}(\tau\Omega_m)F_{\tau,m}^\top=\Omega_m$, which gives
\eqref{eq:tau-standard-symplectic-identification}.
\end{proof}
 \section{Quantum Optimal Transport (QOT) barycenters}\label{sec:qot-bary}

Throughout this section, the barycenter $\rho\in\cS(\cH_0)$ is a state on $\cH_0$, while the given states $\sigma_s\in\cS(\cH_s)$ are states on $(\cH_s)_{s\in[N]}$. For every $s\in[N]$, let
$C_s\in \cL_{\geq}(\cH_0\otimes \cH_s)$ be a positive semidefinite (possibly unbounded)
self-adjoint operator.
The \emph{quantum
optimal transport (QOT)} cost associated with $C_s$ is
\begin{equation}\label{eq:QOT-def-setting}
\operatorname{QOT}(\rho,\sigma\mid C_s)
:= \inf_{\pi\in\Gamma(\rho,\sigma)} \Tr(C_s\pi).
\end{equation}
By Theorem~\ref{theo:compact} and
Proposition~\ref{prop:Hclosed}, the infimum in
\eqref{eq:QOT-def-setting} is attained whenever its value is finite.
When the cost operator $C_s$ is unbounded, $\Tr(C_s\pi)$ is understood in the sense of spectral truncations. The QOT barycenter problem \eqref{intro:QOT-barycenter} can be written as
\begin{equation}\label{eq:joint-primal-setting}
 \cB_{\mathbf A,\mathbf E}^{\mathbf C,\boldsymbol{\alpha}}(\boldsymbol{\sigma})
:=
\inf_{\rho\in\cS_{\mathbf A,\mathbf E}(\cH_0)}
\sum_{s=1}^N
\alpha_s\,
\operatorname{QOT}\!\big(\rho,\sigma_s\mid C_s\big)=\inf_{(\pi_s)_{s\in[N]}\in\Gamma_{\mathrm{joint}}^{\mathbf A,\mathbf E}(\boldsymbol{\sigma})}
\ \sum_{s=1}^N \alpha_s\,\Tr\!\big(C_s\pi_s\big),
\end{equation}

where $\Gamma_{\mathrm{joint}}^{\mathbf A,\mathbf E}(\boldsymbol{\sigma})$ denotes the joint feasible set encoding the
\rev{common barycentric-marginal constraint and the expectation constraints associated with $\mathbf{A}$ and $\mathbf{E}$:}
\begin{equation}\label{eq:joint-feasible-co}
\begin{aligned}
\Gamma_{\mathrm{joint}}^{\mathbf A,\mathbf E}(\boldsymbol{\sigma})
:={}&
\Bigl\{
(\pi_s)_{s\in[N]} \ \Big|\ &
\pi_s\in\cS(\cH_0\otimes\cH_s),\quad
\Tr_{\cH_0}(\pi_s)=\sigma_s,\\
&\Tr[A_\ell\Tr_{\cH_s}(\pi_s)]\leq E_\ell,
\quad \forall s\in[N],\ \ell\in[k],\\
&\Tr_{\cH_s}(\pi_s)=\Tr_{\cH_{s'}}(\pi_{s'}),
\quad \forall s,s'\in[N]
\Bigr\}.
\end{aligned}
\end{equation}

Any minimizer of the joint problem has a common marginal belonging to the barycenter set. Conversely, whenever the pairwise QOT minima are attained for a barycenter \(\rho\), a choice of pairwise optimal couplings gives a minimizer of the joint problem.
This formulation makes the shared-marginal constraint explicit and will serve as
the starting point for compactness and for the derivation of the dual
problem.

\paragraph{Compactness of constrained energy sublevels}

\begin{prop}[Compactness of constrained energy sublevels]
\label{prop:compact1}
Let $\cH_0,\cH$ be separable Hilbert spaces, let
$\sigma\in\cS(\cH)$, and let $c\ge0$. Then the following statements hold.
\begin{enumerate}[label=(\roman*)]
\item
Let $C\in\cL_{\geq}(\cH_0\otimes\cH)$ have compact resolvent. Then
\[
\cA(\sigma,c)
:=
\Bigl\{
\pi\in\cS(\cH_0\otimes\cH):
\Tr_{\cH_0}\pi=\sigma,\ 
\Tr(C\pi)\le c
\Bigr\}
\]
is compact in the trace-norm topology.

\item
Let $H\in\cL_{\geq}(\cH_0)$ have compact resolvent. Then
\[
\cA_H(\sigma,c)
:=
\Bigl\{
\pi\in\cS(\cH_0\otimes\cH):
\Tr_{\cH_0}\pi=\sigma,\ 
\Tr\!\bigl(H\,\Tr_{\cH}\pi\bigr)\le c
\Bigr\}
\]
is compact in the trace-norm topology.
\end{enumerate}
\end{prop}
\begin{proof}
For~(i), let $(\pi_n)_n\subset\cA(\sigma,c)$. Since $C$ has compact
resolvent and $\sup_n\Tr(C\pi_n)\le c$, Lemma~\ref{lemma:boundedHcompact} yields a subsequence
$(\pi_{n_k})_k$ and a state
$\widetilde\pi\in\cS(\cH_0\otimes\cH)$ such that $\pi_{n_k}\rightarrow\widetilde\pi$ in trace norm and $\Tr(C\widetilde\pi)\le c$. By the trace-norm continuity of the partial trace,
\[
\Tr_{\cH_0}\widetilde\pi
=
\lim_{k\to\infty}\Tr_{\cH_0}\pi_{n_k}
=
\sigma.
\]
Hence $\widetilde\pi\in\cA(\sigma,c)$, proving that
$\cA(\sigma,c)$ is compact.

For~(ii), let $(\pi_n)_n\subset\cA_H(\sigma,c)$. The bound $\Tr(H\Tr_{\cH}\pi_n)\le c$ and Lemma~\ref{lemma:boundedHcompact} yield a subsequence
$(\Tr_{\cH}\pi_{n_k})_k$ and a state $\widetilde\rho\in\cS(\cH_0)$ such that $\Tr_{\cH}\pi_{n_k}\rightarrow\widetilde\rho$ in trace norm. Since $\Tr_{\cH_0}\pi_{n_k}=\sigma$ for every $k$, Theorem~\ref{theo:compact} yields a further subsequence,
which we do not relabel, and a state
$\widetilde\pi\in\cS(\cH_0\otimes\cH)$ such that $\pi_{n_k}\rightarrow\widetilde\pi$ in trace norm, with
\[
\Tr_{\cH_0}\widetilde\pi=\sigma,
\qquad
\Tr_{\cH}\widetilde\pi=\widetilde\rho.
\]
Finally, by the lower semicontinuity of
$\rho\mapsto\Tr(H\rho)$,
\[
\Tr(H\widetilde\rho)
\le
\liminf_{k\to\infty}\Tr(H\Tr_{\cH}\pi_{n_k})
\le c.
\]
Hence $\widetilde\pi\in\cA_H(\sigma,c)$, proving that
$\cA_H(\sigma,c)$ is compact.
\end{proof}
\begin{coro}\label{coro:compactc}
Let $S$ be a finite index set.
Assume that $S$ is partitioned as $S=S_{\mathrm{I}}\sqcup S_{\mathrm{F}}$. Let $\cH_0$ and $\cH_s$ for all $s\in S$ be separable Hilbert spaces. For each $s\in S_{\mathrm{I}}$, let $C_s\in\cL_\geq(\cH_0\otimes\cH_s)$ be a cost operator with compact resolvent.
For each $s\in S_{\mathrm{F}}$, let $H_s\in\cL_\geq(\cH_0)$ be a nonnegative operator with compact resolvent.

For each $s\in S$, let $\sigma_s\in \cS(\cH_s)$ and $c_s\in\R_+$.
Define the feasible set
\[
\cA_s :=
\begin{cases}
\bigl\{\pi\in\cS(\cH_0\otimes\cH_s):\ \Tr_{\cH_0}\pi=\sigma_s,\ \Tr(C_s\pi)\le c_s\bigr\},
& s\in S_{\mathrm{I}},\\[1mm]
\bigl\{\pi\in\cS(\cH_0\otimes\cH_s):\ \Tr_{\cH_0}\pi=\sigma_s,\ \Tr\!\bigl(H_s\,\Tr_{\cH_s}\pi\bigr)\le c_s\bigr\},
& s\in S_{\mathrm{F}}.
\end{cases}
\]
Finally, set
\[
\cA :=
\Bigl\{(\pi_s)_{s\in S}\in\prod_{s\in S}\cA_s:\ \Tr_{\cH_s}\pi_s=\Tr_{\cH_{s'}}\pi_{s'}\ \text{for all } s,s'\in S\Bigr\}.
\]
Then $\cA$ is compact in the norm $\|(\,T_s\,)_{s\in S}\|:=\sum_{s\in S}\|T_s\|_1$.
\end{coro}

\begin{proof}
For each $s\in S_{\mathrm I}$, the set $\cA_s$ is compact by
Proposition~\ref{prop:compact1}(i), whereas for each
$s\in S_{\mathrm F}$ it is compact by
Proposition~\ref{prop:compact1}(ii). Since $S$ is finite, the product $\prod_{s\in S}\cA_s$ is compact for the product topology,
and the norm $\sum_{s\in S}\|\cdot\|_1$ induces an equivalent metrization. The constraint $\Tr_{\cH_s}\pi_s=\Tr_{\cH_{s'}}\pi_{s'}$ for all $s,s'\in S$ defines a closed subset of the product,
because of the continuity of the partial trace $\Tr_{\cH_s}$ in trace norm. Hence $\cA$ is compact as a closed subset of a compact set.
\end{proof}
\begin{coro}[Compactness under additional marginal energy constraints]
\label{coro:compact-additional-marginal-energy}
Assume the hypotheses of Corollary~\ref{coro:compactc}. Let
$k\in\mathbb N$. For each $\ell\in[k]$, let
$A_\ell\in\cL_\geq(\cH_0)$ be a self-adjoint operator on $\cH_0$, and let
$E_\ell\in[0,+\infty)$. Define
\[
    \cA_{\mathbf A,\mathbf E}
    :=
    \left\{
    (\pi_s)_{s\in S}\in\cA:
    \Tr\!\left(
        A_\ell\,\Tr_{\cH_s}\pi_s
    \right)
    \leq E_\ell
    \quad
    \text{for all }s\in S,\ \ell\in[k]
    \right\}.
\]
Then $\cA_{\mathbf A,\mathbf E}$ is compact in the norm $
    \|(\,T_s\,)_{s\in S}\|
    :=
    \sum_{s\in S}\|T_s\|_1 $.
\end{coro}
\begin{proof}
By Corollary~\ref{coro:compactc}, the set $\cA$ is compact in the
product trace norm. Fix $s\in S$ and $\ell\in[k]$. By Proposition~\ref{prop:marginal-energy-closed},
the set
\[
    \left\{
    \pi\in\cS(\cH_0\otimes\cH_s):
    \Tr\!\left(
        A_\ell\,\Tr_{\cH_s}\pi
    \right)
    \leq E_\ell
    \right\}
\]
is closed in the trace-norm topology. Hence
$\cA_{\mathbf A,\mathbf E}$ is the intersection of $\cA$ with
finitely many closed sets. Therefore, $\cA_{\mathbf A,\mathbf E}$ is a closed subset of the
compact set $\cA$, and hence is compact.
\end{proof}
\paragraph{Bounds for quadrature-generated cost operators.}

The next proposition provides upper and lower bounds on quadrature-generated cost operators. 

\begin{prop}[Two-sided bounds]
\label{prop:quadrature_bounds}
Let $m\in\N$, $\cH_0,\cH$ be separable Hilbert spaces and \(C\) be a quadrature-generated cost on
\(\cH_0\otimes\cH\), associated with quadrature generators
\(\{R_j\}_{j\in[2m]}\) on \(\cH_0\) and
\(\{S_j\}_{j\in[2m]}\) on \(\cH\). Let $H_0:=\sum_{j=1}^{2m} R_j^2$ and $H:=\sum_{j=1}^{2m} S_j^2$ be the corresponding one-body square-sums as in \eqref{intro:eq:one-body-H}. Then, 
\[
    C
    \le
    2\bigl(H_0\otimes \I_{\cH}
    +
    \I_{\cH_0}\otimes H\bigr),\text{\quad and}\quad
    C+3\,\I_{\cH_0}\otimes H
    \ge
    \frac34\,H_0\otimes \I_{\cH}.
\]
\end{prop}

\begin{proof}
On the algebraic tensor core
$\mathcal H_R^\infty\otimes_{\mathrm{alg}}\mathcal H_S^\infty$,
\[
\begin{aligned}
    &2\sum_{j=1}^{2m}
    \bigl(\|R_j\otimes \I_{\cH}\psi\|^2
    +\|\I_{\cH_0}\otimes S_j\psi\|^2\bigr)\\
    &\quad-
    \sum_{j=1}^{2m}
    \|(R_j\otimes \I_{\cH}-\I_{\cH_0}\otimes S_j)\psi\|^2\\
    &\qquad=
    \sum_{j=1}^{2m}
    \|(R_j\otimes \I_{\cH}+\I_{\cH_0}\otimes S_j)\psi\|^2
    \ge0,
\end{aligned}
\]
and
\begin{equation}
\begin{split}
    &\sum_{j=1}^{2m}\|(R_j\otimes \I_{\cH}-\I_{\cH_0}\otimes S_j)\psi\|^2
    +
    3\sum_{j=1}^{2m}\|\I_{\cH_0}\otimes S_j\psi\|^2
    -
    \frac34\sum_{j=1}^{2m}\|R_j\otimes \I_{\cH}\psi\|^2\\
    =&
    \sum_{j=1}^{2m}
    \left\|
        \frac12R_j\otimes \I_{\cH}\psi-2\I_{\cH_0}\otimes S_j\psi
    \right\|^2
    \ge0.
\end{split}
\end{equation}
These quadratic-form inequalities hold on a common form core for
$C$, $H_0\otimes\I_{\cH}$, and $\I_{\cH_0}\otimes H$; passing to the
closures proves the claim.
\end{proof}
\subsection{Existence of QOT barycenters}\label{subsec:existence_proof}
In this subsection, we prove Theorem~\ref{thm:existenceQOTbary_intro}
by the direct method of the calculus of variations.
\begin{proof}

By assumption~\eqref{eq:exas},
$ \cB_{\mathbf A,\mathbf E}^{\mathbf C,\boldsymbol\alpha}
(\boldsymbol\sigma)<+\infty$.
Choose a minimizing sequence
\[
\left(
(\pi_s^{(n)})_{s\in[N]}
\right)_{n\in\N}
\subset
\Gamma_{\mathrm{joint}}^{\mathbf A,\mathbf E}
(\boldsymbol\sigma)
\]
for the joint formulation~\eqref{eq:joint-primal-setting}, so that
\[
\sum_{s=1}^N
\alpha_s\Tr\!\bigl(C_s\pi_s^{(n)}\bigr)
\rightarrow
\cB_{\mathbf A,\mathbf E}^{\mathbf C,\boldsymbol\alpha}
(\boldsymbol\sigma).
\]
After discarding finitely many terms if necessary, we may assume that
\[
\sum_{s=1}^N
\alpha_s\Tr\!\bigl(C_s\pi_s^{(n)}\bigr)
\le
 \cB_{\mathbf A,\mathbf E}^{\mathbf C,\boldsymbol\alpha}
(\boldsymbol\sigma)+1,
\qquad n\in\N.
\]

Fix $s_0\in[N]$ as in Theorem~\ref{thm:existenceQOTbary_intro}. We consider the following cases.
\medskip

\noindent

\emph{(i) Cost operators with compact resolvent.}
If $C_{s_0}$ is a cost operator with compact resolvent, define the set
\[
\cA_{s_0}:=\Bigl\{\pi\in\cS(\cH_0\otimes\cH_{s_0}):\ 
\begin{aligned}[t]
&\Tr_{\cH_0}\pi=\sigma_{s_0},\\
&\Tr(C_{s_0}\pi)\le
\bigl(\cB_{\mathbf A,\mathbf E}^{\mathbf C,\boldsymbol\alpha}
(\boldsymbol\sigma)+1\bigr)/\alpha_{s_0}
\end{aligned}
\Bigr\}.
\]
By Proposition~\ref{prop:compact1}(i), $\cA_{s_0}$ is compact in the
trace-norm topology, and by construction
$\pi_{s_0}^{(n)}\in \cA_{s_0}$ for every $n\in\N$.

\medskip

\noindent
\emph{(ii) Quadrature-generated cost.}
Assume $C_{s_0}$ is a quadrature-generated cost with associated one-body square-sums
$H^{(s_0)}_0$ on $\cH_0$ and $H^{(s_0)}$ on $\cH_{s_0}$. Define
\[
\cA_{s_0}:=\Bigl\{\pi\in\cS(\cH_0\otimes\cH_{s_0}):\ 
\begin{aligned}[t]
&\Tr_{\cH_0}\pi=\sigma_{s_0},\\
&\Tr\!\big(H^{(s_0)}_0\,\Tr_{\cH_{s_0}}\pi\big)\\
&\qquad\le \tfrac{4}{3}\Bigl(
\tfrac{\cB_{\mathbf A,\mathbf E}^{\mathbf C,\boldsymbol\alpha}
(\boldsymbol\sigma)+1}{\alpha_{s_0}}
+3\Tr(H^{(s_0)}\sigma_{s_0})\Bigr)
\end{aligned}
\Bigr\}.
\]
By Proposition~\ref{prop:quadrature_bounds}, we have
\[
\Tr\!\big(C_{s_0}\pi_{s_0}^{(n)}\big)
\ge \tfrac{3}{4}\Tr\!\big(H^{(s_0)}_0\,\Tr_{\cH_{s_0}}\pi_{s_0}^{(n)}\big)-3\Tr(H^{(s_0)}\sigma_{s_0}),
\]
since $\Tr_{\cH_0}\pi_{s_0}^{(n)}=\sigma_{s_0}$. Hence, by Proposition~\ref{prop:compact1}(ii), $\cA_{s_0}$ is compact in the trace-norm topology, and $\pi_{s_0}^{(n)}\in \cA_{s_0}$ for every $n\in\N$.

\medskip
In either case, after passing to a subsequence, we may assume that $\pi_{s_0}^{(n)}\to\pi_{s_0}^\ast$ in trace norm. Let
$\rho_n:=\Tr_{\cH_{s_0}}\pi_{s_0}^{(n)}$ and $\rho^\ast:=\Tr_{\cH_{s_0}}\pi_{s_0}^\ast$. By trace-norm continuity of the partial trace, $\rho_n\to\rho^\ast$ in trace norm.

For every $s\ne s_0$, one has $\pi_s^{(n)}\in\Gamma(\rho_n,\sigma_s)$. Theorem~\ref{theo:compact} gives a further trace-norm convergent subsequence of $(\pi_s^{(n)})_n$. Since $N$ is finite, successive subsequence extractions yield
$(\pi_s^{(n)})_{s\in[N]}\to(\pi_s^\ast)_{s\in[N]}$ in product trace norm, where $\pi_s^\ast\in\Gamma(\rho^\ast,\sigma_s)$ for every $s\in[N]$.

Moreover, $\rho_n\in\cS_{\mathbf A,\mathbf E}(\cH_0)$ for every $n\in\N$. By Proposition~\ref{prop:Hclosed},
$\Tr(A_\ell\rho^\ast)\leq E_\ell$ for every $\ell\in[k]$. Hence
$(\pi_s^\ast)_{s\in[N]}\in\Gamma_{\mathrm{joint}}^{\mathbf A,\mathbf E}(\boldsymbol\sigma)$.

Finally, for each $s\in[N]$, the map $\pi\in\cS(\cH_0\otimes\cH_s)
\longmapsto
\Tr(C_s\pi)\in[0,+\infty]$ is lower semicontinuous in trace norm. Indeed, by Proposition~\ref{prop:spectral-trunc-positive},
\[
\Tr(C_s\pi)
=
\sup\Bigl\{
\Tr(B\pi):
B\in\cB_S(\cH_0\otimes\cH_s),\
0\le B\le C_s
\Bigr\},
\]
and each map $\pi\mapsto\Tr(B\pi)$ is trace-norm continuous. Therefore,
\[
\sum_{s\in [N]}\alpha_s\,\Tr(C_s\pi_s^\ast)
\le \liminf_{n\to\infty}\sum_{s\in [N]}\alpha_s\,\Tr(C_s\pi_s^{(n)})
=\cB_{\mathbf A,\mathbf E}^{\mathbf C,\boldsymbol\alpha}(\boldsymbol{\sigma})\rev{,}
\]
so $(\pi_s^\ast)_{s\in [N]}$ is optimal. Let $\overline{\rho}:=\Tr_{\cH_s}\pi_s^\ast$ for any $s\in[N]$;
this is independent of $s$ by the common-marginal constraint. Then $\overline{\rho}\in\mathrm{Bar}^{\mathbf C,\boldsymbol{\alpha}}_{\mathbf A,\mathbf E}(\boldsymbol{\sigma})$ is a barycenter and
\[
 \cB_{\mathbf A,\mathbf E}^{\mathbf C,\boldsymbol\alpha}(\boldsymbol{\sigma})
=\min_{\rho\in\cS_{\mathbf A,\mathbf E}(\cH_0)}\ \sum_{s=1}^{N}\alpha_s\,\operatorname{QOT}\!\big(\rho,\sigma_s\mid C_s\big).
\]
\end{proof}

\subsection{Strong duality}\label{subsec:duality_proof}

We next establish the strong dual formulation~\eqref{eq:energy-constrained-duality-barycenter}. 
The proof is based on Fenchel--Rockafellar duality. Bunth--Pitrik--Titkos--Virosztek~\cite{BunthPitrikTitkosVirosztek2026} prove strong Kantorovich duality for a linearized channel-based quantum transport problem with generic costs and analyze optimal couplings on quantum bits for distinguished costs and restricted state classes. Here the primal variable is instead an $N$-tuple of couplings with a common barycentric marginal, the Hilbert spaces may be infinite-dimensional, the costs may be unbounded, and the barycentric marginal may satisfy the energy constraints~\eqref{eq:constrain}.

\begin{proof}[Proof of Theorem~\ref{thm:energy-constrained-barycenter-duality}]
Equip the product space $\prod_{s=1}^N
\cB_S(\cH_0\otimes\cH_s)$ with the maximum product norm $(B_s)_{s\in[N]}\mapsto\max_{s\in[N]}\|B_s\|_{\mathrm{op}}$ and define $f:\prod_{s=1}^N\cB_S(\cH_0\otimes\cH_s)
\rightarrow
(-\infty,+\infty]$ by
\begin{equation}\label{eq:def:f}
f((B_s)_{s\in[N]})
:=
\iota_{\cB_\le(\mathbf C)}((B_s)_{s\in[N]})
=
\begin{cases}
0,
&
B_s\le C_s,\quad s\in[N],
\\[1mm]
+\infty,
&
\text{otherwise},
\end{cases}
\end{equation}
with $\cB_\le(\mathbf C)
    :=
    \left\{
    (B_s)_{s=1}^N\in \prod_{s=1}^N \cB_S(\cH_0\otimes \cH_s):
    B_s\le C_s,\quad s\in[N]
    \right\}$. Define also $g_{\mathbf A,\mathbf E}:
\prod_{s=1}^N
\cB_S(\cH_0\otimes\cH_s)
\rightarrow
(-\infty,+\infty]$ by
\begin{equation}\label{eq:defg}
g_{\mathbf A,\mathbf E}((B_s)_{s\in[N]})
:=
\inf_{\mathcal R_{\mathbf A,\mathbf E}((B_s)_{s\in[N]})}
\left[
-\sum_{s=1}^N\alpha_s\Tr(V_s\sigma_s)
+
\sum_{\ell=1}^k\eta_\ell E_\ell
\right],
\end{equation}
with the convention $\inf\varnothing=+\infty$, where
\[
\begin{aligned}
\mathcal R_{\mathbf A,\mathbf E}((B_s)_{s\in[N]})
:=
\Bigl\{
&
\bigl(
(U_s^{(\ell)})_{s\in[N],\,\ell\in[k]},
(V_s)_{s\in[N]},
(\eta_\ell)_{\ell\in[k]}
\bigr):
\\
&
U_s^{(\ell)}\in\cB_S(\cH_0),\quad
V_s\in\cB_S(\cH_s),\quad
\eta_\ell\ge0,
\\
&
B_s=
\left(
\sum_{\ell=1}^k U_s^{(\ell)}
\right)\oplus V_s,
\qquad s\in[N],
\\
&
-\sum_{s=1}^N
\alpha_s U_s^{(\ell)}
\le
\eta_\ell A_\ell,
\qquad \ell\in[k]
\Bigr\}.
\end{aligned}
\]
The function $f$ is proper and convex, and
Proposition~\ref{prop:gAE-convex-proper} shows that
$g_{\mathbf A,\mathbf E}$ is proper and convex. By Proposition~\ref{prop:FR-left-identification-constrained},
\begin{equation}
\begin{aligned}
&
\inf_{(B_s)_{s\in[N]}
\in
\prod_{s=1}^N
\cB_S(\cH_0\otimes\cH_s)}
\left\{
f((B_s)_{s\in[N]})
+
g_{\mathbf A,\mathbf E}((B_s)_{s\in[N]})
\right\}
\\
&\qquad
=
-
\sup
\Biggl\{
\sum_{s=1}^N
\alpha_s\Tr(V_s\sigma_s)
-
\sum_{\ell=1}^k
\eta_\ell E_\ell
:
\bigl(
(U_s,V_s)_{s\in[N]},
(\eta_\ell)_{\ell\in[k]}
\bigr)
\in\mathsf{Pot}_{\mathbf A,\mathbf E}^{\mathbf C,\boldsymbol\alpha}(\boldsymbol\sigma)\Biggr\}.
\end{aligned}
\label{eq:FR-left-main}
\end{equation}
Thus the left-hand side of the Fenchel--Rockafellar formula is the
negative of the claimed dual value. On the other hand, the conjugate computations in
Propositions~\ref{prop:conjugate-f-cost-constraint}
and~\ref{prop:conjugate-gAE}, together with
Proposition~\ref{prop:FR-right-identification-constrained}, give
\begin{equation}
\begin{aligned}
&
\sup_{(\Lambda_s)_{s\in[N]}
\in
\prod_{s=1}^N
\cB_S(\cH_0\otimes\cH_s)^*}
\left\{
-f^*((-\Lambda_s)_{s\in[N]})
-
g_{\mathbf A,\mathbf E}^*((\Lambda_s)_{s\in[N]})
\right\}
\\
&\qquad
=
-
\inf_{(\pi_s)_{s\in[N]}
\in
\Gamma_{\mathrm{joint}}^{\mathbf A,\mathbf E}
(\boldsymbol\sigma)}
\sum_{s=1}^N
\alpha_s\Tr(C_s\pi_s)
\\
&\qquad
=
-
\cB_{\mathbf A,\mathbf E}^{\mathbf C,\boldsymbol\alpha}
(\boldsymbol\sigma).
\end{aligned}
\label{eq:FR-right-main}
\end{equation}
Thus the right-hand side is the negative of the joint primal value.

It remains to verify the Fenchel--Rockafellar qualification condition. Fix $\delta>0$ and consider $B_s^0:=-\delta\I_{\cH_0\otimes\cH_s}$ for $s\in[N]$. Taking $U_s^{(\ell)}=0$, $V_s=-\delta\I_{\cH_s}$ and $\eta_\ell=0$ shows that $(B_s^0)_{s\in[N]}\in\Dom g_{\mathbf A,\mathbf E}$. Moreover, since $C_s\ge0$, if $\|B_s-B_s^0\|_{\mathrm{op}}<\frac{\delta}{2}$ for all $s\in[N]$, then $B_s
\le
-\frac{\delta}{2}
\I_{\cH_0\otimes\cH_s}
\le
C_s$ for all $s\in[N]$. Thus $(B_s^0)_{s\in[N]}\in
\Dom f\cap\Dom g_{\mathbf A,\mathbf E}$,
and $f$ is continuous at this point.

The qualification condition of
Theorem~\ref{app:FRTheorem} is thus satisfied. Applying the
Fenchel--Rockafellar theorem and combining
\eqref{eq:FR-left-main} and \eqref{eq:FR-right-main} yields
\[
\cB_{\mathbf A,\mathbf E}^{\mathbf C,\boldsymbol\alpha}
(\boldsymbol\sigma)
=
\sup
\Biggl\{
\sum_{s=1}^N
\alpha_s\Tr(V_s\sigma_s)
-
\sum_{\ell=1}^k
\eta_\ell E_\ell
:
\bigl(
(U_s,V_s)_{s\in[N]},
(\eta_\ell)_{\ell\in[k]}
\bigr)
\in
\mathsf{Pot}_{\mathbf A,\mathbf E}^{\mathbf C,\boldsymbol\alpha}(\boldsymbol\sigma)
\Biggr\},
\]
which is precisely
\eqref{eq:energy-constrained-duality-barycenter}.
\end{proof}

\begin{proof}[Proof of Corollary~\ref{cor:unconstrained-barycenter-duality}]
Apply Theorem~\ref{thm:energy-constrained-barycenter-duality}
with $k=1$, $A_1=0$, and $E_1=0$. The primal problem is then
unconstrained, while the dual compatibility condition becomes
$\sum_{s=1}^N\alpha_sU_s\geq0$.

For any admissible family, set
$U:=\sum_{s=1}^N\alpha_sU_s$ and replace $U_s$ by $U_s-U$
for every $s\in[N]$. Since $U\geq0$, the cost inequalities remain
valid, the weighted sum of the first-marginal potentials becomes
zero, and the dual objective is unchanged. Conversely, every
family with weighted sum zero is admissible for the preceding
inequality constraint. Hence the two dual suprema coincide,
which proves \eqref{eq:unconstrained-barycenter-duality}.
\end{proof}
\subsection{\texorpdfstring{A posteriori finite-energy bounds for\\ quadrature-generated barycenters}{A posteriori finite-energy bounds for quadrature-generated barycenters}}
\label{subsec:finite_energy}

We next prove an a posteriori finite-energy estimate for barycenters
associated with quadrature-generated costs. The estimate shows that the
barycenter automatically inherits a finite one-body energy bound from the
input marginals, without imposing this bound as an additional constraint.
This is analogous to the finite-moment inheritance property of classical
Wasserstein barycenters.

We denote by $\Gamma_{\mathrm{joint}}(\boldsymbol{\sigma})$ the joint feasible set encoding the
common barycentric marginal constraint\rev{:}
\begin{equation}\label{eq:joint-feasible}
\Gamma_{\mathrm{joint}}(\boldsymbol{\sigma})
:=
\left\{
(\pi_s)_{s\in[N]}
\;\middle|\;
\begin{aligned}
&\pi_s\in\cS(\cH_0\otimes\cH_s), \, \Tr_{\cH_0}(\pi_s)=\sigma_s,\\
&\text{and } \Tr_{\cH_s}(\pi_s)
 =\Tr_{\cH_{s'}}(\pi_{s'}),\forall s,s'\in[N]
\end{aligned}
\right\}.
\end{equation}

The unconstrained QOT barycenter problem \eqref{eq:unconstrained-barycenter-primal} can be written as
\begin{equation}\label{eq:joint-primal-setting-un}
 \cB_0^{\mathbf C,\boldsymbol{\alpha}}(\boldsymbol{\sigma})
:=
\inf_{\rho\in\cS(\cH_0)}
\sum_{s=1}^N
\alpha_s\,
\operatorname{QOT}\!\big(\rho,\sigma_s\mid C_s\big)=\inf_{(\pi_s)_{s\in[N]}\in\Gamma_{\mathrm{joint}}(\boldsymbol{\sigma})}
\ \sum_{s=1}^N \alpha_s\,\Tr\!\big(\pi_s\,C_s\big).
\end{equation}
\begin{proof}[Proof of Proposition~\ref{prop:finitesecondemomentrho}]
By Proposition~\ref{prop:quadrature_bounds}, for any
$(\pi_s)_{s\in[N]}\in
\Gamma_{\mathrm{joint}}(\boldsymbol\sigma)$
with common barycentric marginal $\rho$,
\[
\Tr(C_s\pi_s)
\le
2\Tr\!\left(
H_0^{(s)}\,\Tr_{\cH_s}\pi_s
\right)
+
2\Tr\!\left(
H^{(s)}\,\Tr_{\cH_0}\pi_s
\right)=
2\Tr(H_0^{(s)}\rho)
+
2\Tr(H^{(s)}\sigma_s).
\]
The last equality holds since $\Tr_{\cH_s}\pi_s=\rho$ and $\Tr_{\cH_0}\pi_s=\sigma_s$. Therefore,
\[
\cB_0^{\mathbf C,\boldsymbol\alpha}
(\boldsymbol\sigma)
\le
2\left(
M+
\sum_{s=1}^N
\alpha_s\Tr(H^{(s)}\sigma_s)
\right),
\]
where $M:=\inf_{\rho\in\cS(\cH_0)}
\sum_{s=1}^N
\alpha_s\Tr(H_0^{(s)}\rho)$.

On the other hand, let $\overline\rho$ be a barycenter and let
$(\pi_s^\ast)_{s\in[N]}$ be optimal couplings with common
barycentric marginal $\overline\rho$. By the lower bound in
Proposition~\ref{prop:quadrature_bounds},
\[
\Tr(C_s\pi_s^\ast)
\ge
\frac34\Tr(H_0^{(s)}\overline\rho)
-
3\Tr(H^{(s)}\sigma_s).
\]
Summing over $s$ gives
\[
\cB_0^{\mathbf C,\boldsymbol\alpha}
(\boldsymbol\sigma)
\ge
\frac34
\sum_{s=1}^N
\alpha_s
\Tr
\left( H_0^{(s)}\overline\rho\right)
-
3
\sum_{s=1}^N
\alpha_s\Tr(H^{(s)}\sigma_s).
\]
Combining the two bounds yields
\begin{equation}
\sum_{s=1}^N
\alpha_s\Tr
\left( H_0^{(s)}\overline\rho
\right)
\le
\frac83 M
+
\frac{20}{3}
\sum_{s=1}^N
\alpha_s\Tr(H^{(s)}\sigma_s)
<+\infty.
\end{equation}
\end{proof}
\subsection{\texorpdfstring{Examples: finite-dimensional barycenters and\\ $2$-quantum Wasserstein-type barycenters}{Examples: finite-dimensional barycenters and 2-quantum Wasserstein-type barycenters}}
\label{subsec:examples-qw2}

The existence and duality results established above require one cost operator from either of two structural classes: cost operators with compact resolvent and quadrature-generated costs.
These assumptions cover a broad range of situations, including both finite-dimensional
quantum systems and the main continuous-variable quadratic transport models considered
in this paper.

We begin with the finite-dimensional case. If $\cH_0$ and $(\cH_s)_{s\in[N]}$ are all
finite-dimensional, then every cost operator
\[
C_s\in \cL_{\geq}(\cH_0\otimes\cH_s)
\]
is automatically bounded and has finite spectrum. In particular, each $C_s$ is a cost operator with compact resolvent.
Hence, from Theorem \ref{thm:existenceQOTbary_intro} and Corollary~\ref{cor:unconstrained-barycenter-duality}, \rev{we obtain the following:}
\begin{proposition}
Let $N\in\mathbb N$, $\boldsymbol{\alpha}:=(\alpha_s)_{s\in[N]}$ be positive weights satisfying $\sum_{s=1}^N \alpha_s = 1$. Let $\cH_0$ and $(\cH_s)_{s\in[N]}$ be finite-dimensional Hilbert spaces, $\boldsymbol{\sigma}:=(\sigma_s)_{s\in[N]}$ with $\sigma_s\in\cS(\cH_s)$,
and $\mathbf C:=(C_s)_{s\in[N]}$ be a family of cost operators.

Then the unconstrained QOT barycenter problem
\eqref{eq:unconstrained-barycenter-primal}
admits a minimizer and satisfies the duality formula
\eqref{eq:unconstrained-barycenter-duality}.
\end{proposition}

Our second main class of examples is given by the continuous-variable quadratic
transport costs arising in quantum optimal transport. These include the
state-state quadratic transport cost for quantum densities introduced by Golse--Mouhot--Paul and further developed by Caglioti--Golse--Paul~\cite{caglioti2021towards,Golse-Mouhot-Paul:2016}
and the channel-based quadratic transport cost introduced by
De Palma and Trevisan~\cite{DPaTre19}, which can be treated in a unified $\tau$-formalism.

Let $m\in\mathbb N$ and $\tau\in\{\pm1\}$. Recall the space $L_2^{(\tau)}(\R^m)$ and $\hat{\mathbf{R}}^{(\tau)}$ defined in \eqref{eq:tau}. We use hats to denote operators acting on these
continuous-variable Hilbert spaces. Consider the bipartite space $L_2(\mathbb{R}^m)\otimes L_2^{(\tau)}(\mathbb{R}^m)$. The unified canonical vector $\hat{\mathbf{R}}_\tau$ is defined in \eqref{eq:bitau}. We define the cost operator
\begin{equation}\label{eq:tau-cost}
\hat C^{(\tau)}
:=
\sum_{j=1}^{2m}
\bigl(\hat R_j\otimes \I_{L_2^{(\tau)}(\mathbb{R}^m)}-\I_{L_2(\mathbb{R}^m)}\otimes \hat R_j^{(\tau)}\bigr)^2.
\end{equation}
For $\tau=+1$ this recovers the Golse--Mouhot--Caglioti--Paul cost, whereas
for $\tau=-1$ it recovers the De Palma--Trevisan cost. They are typically unbounded and fail to have compact resolvent; nevertheless, they belong to the class of
quadrature-generated costs introduced in Section~\ref{intro:statementofresults}. In the quantum-state formulation, the associated
one-body square-sums are both equal to the harmonic oscillator Hamiltonian
\begin{equation}\label{eq:oscillator}
\hat{\mathsf{H}}_{\mathrm{osc}}
=
-\hbar^2\nabla_x^2+\sum_{j=1}^m x_j^2,
\end{equation}
whereas in the quantum-channel formulation they are given by
$(\hat{\mathsf{H}}_{\mathrm{osc}},\hat{\mathsf{H}}_{\mathrm{osc}}^\top)$.

In accordance with the unified $\tau$-formalism, we define
\[
\Gamma_{+1}(\hat\rho,\hat\sigma):=\Gamma(\hat\rho,\hat\sigma),
\qquad
\Gamma_{-1}(\hat\rho,\hat\sigma):=\Gamma(\hat\rho,\hat\sigma^\top),
\]
where $\Gamma(\hat\rho,\hat\sigma)$ is
given by Definition~\ref{def:quantumcouplings}. Accordingly, for $\tau\in\{\pm1\}$ we define
\begin{equation}\label{eq:qw}
\QW_{2,\tau}^2(\hat\rho,\hat\sigma)
:=
\inf_{\hat\pi\in\Gamma_\tau(\hat\rho,\hat\sigma)}
\Tr\!\bigl(\hat C^{(\tau)}\hat\pi\bigr).
\end{equation}
The corresponding barycenter problems are defined in
\eqref{eq:tau-barycenter-intro}.

Applying Theorem~\ref{thm:existenceQOTbary_intro},
Corollary~\ref{cor:unconstrained-barycenter-duality}, and
Proposition~\ref{prop:finitesecondemomentrho}, we obtain the following.
\begin{proposition}
Let $m,N\in\mathbb N$, let $\boldsymbol\alpha=(\alpha_s)_{s\in[N]}$ be
positive weights satisfying $\sum_{s=1}^N\alpha_s=1$, and let
$\hat{\boldsymbol\sigma}=(\hat\sigma_s)_{s\in[N]}$ with
$\hat\sigma_s\in\cS(L_2(\mathbb R^m))$ such that
\[
\Tr(\hat{\mathsf{H}}_{\mathrm{osc}}\hat\sigma_s)<+\infty,
\qquad s\in[N].
\]
Then the $2$-quantum Wasserstein barycenter problems
\eqref{eq:tau-barycenter-intro} admit minimizers and satisfy the duality
formula~\eqref{eq:unconstrained-barycenter-duality}. Moreover, every barycenter
$\hat{\overline\rho}_\tau$ has finite energy with respect to the harmonic
oscillator Hamiltonian:
\[
\Tr(\hat{\mathsf{H}}_{\mathrm{osc}}\hat{\overline\rho}_\tau)<\infty.
\]
\end{proposition}
 \section{Gaussian 2-QW barycenters: reduction and uniqueness}\label{sec:gaussian-states}

\subsection{Gaussian states and Gaussian couplings}
\label{subsec:gaussian-states-couplings}

We first recall Gaussian states in the abstract setting of an irreducible
regular Weyl system; see \cite[Sec.~III.A]{HolevoWerner2001BosonicGaussian}
and \cite[Chap.~12]{Holevo2019QSCI}. We then specialize to the standard
CCR system, its transpose representation introduced in
Section~\ref{subsec:dual-transpose-CCR}, and the bipartite
$\tau$-systems used in the $2$-QW problems.

Let $\Theta\in\cM_{2m}(\R)$ be a nondegenerate skew-symmetric matrix,
let $W:\R^{2m}\to\mathcal U(\cH)$ be an irreducible regular
$\Theta$-Weyl system, and let
$\mathbf R=(R_1,\ldots,R_{2m})$ be its canonical generators. We write $\mathsf H_{\mathbf R}:=\sum_{j=1}^{2m}R_j^2$, understood in the closure sense of
Remark~\ref{rem:quadratic-hamiltonian-closure-convention}.

We first specify the class of states for which the canonical first and second moments are well defined.
\begin{deff}[States with finite canonical second moments;
see {\cite[Secs.~2.3.3 and~12.3.1]{Holevo2019QSCI}}]
\label{def:finite-canonical-second-moments}
We define
\[
\cS_2^R(\cH)
:=
\left\{
\rho\in\cS(\cH):
\Tr(\mathsf H_{\mathbf R}\rho)<+\infty
\right\}.
\]
For $\rho\in\cS_2^R(\cH)$, its canonical mean vector
$\mathbf m_\rho=(m_{\rho,1},\ldots,m_{\rho,2m})^\top\in\R^{2m}$
and covariance matrix
$\Sigma_\rho=[(\Sigma_\rho)_{jk}]_{j,k\in[2m]}
\in\cM_{2m,\mathrm{sym}}(\R)$ are defined by
$m_{\rho,j}:=\Tr(\rho R_j)$ and
\[
(\Sigma_\rho)_{jk}
:=
\frac12
\Tr\!\left(
\rho\left[
(R_j-m_{\rho,j}\I_\cH)(R_k-m_{\rho,k}\I_\cH)
+
(R_k-m_{\rho,k}\I_\cH)(R_j-m_{\rho,j}\I_\cH)
\right]
\right).
\]
\end{deff}

For every state, the Weyl characteristic function is well defined; on
the finite-moment class, it describes Gaussianity as follows.
\begin{deff}[Gaussian states]
\label{def:weyl-characteristic-gaussian-state}
For every $\rho\in\cS(\cH)$, its Weyl characteristic function is
\[
\chi_\rho(\boldsymbol\xi)
:=
\Tr\bigl(\rho W(\boldsymbol\xi)\bigr),
\qquad
\boldsymbol\xi\in\R^{2m}.
\]
For $\rho\in\cS_2^R(\cH)$, following
\cite[Sec.~III.A]{HolevoWerner2001BosonicGaussian} and
\cite[Sec.~12.3.2]{Holevo2019QSCI}, the state $\rho$ is called Gaussian if
\[
\chi_\rho(\boldsymbol\xi)
=
\exp\!\left(
i\mathbf m_\rho^\top\boldsymbol\xi
-\frac12\boldsymbol\xi^\top\Sigma_\rho\boldsymbol\xi
\right),
\qquad
\boldsymbol\xi\in\R^{2m}.
\]
We denote the set of Gaussian states on $\cH$ by $\cS^G(\cH)$.
\end{deff}
The possible covariance matrices of states with finite canonical second moments are characterized by the Robertson--Schr\"odinger uncertainty relation.
\begin{theorem}[Robertson--Schr\"odinger uncertainty relation;
see {\cite[Thm.~12.17]{Holevo2019QSCI}}]
\label{thm:rs-uncertainty-canonical}
For every $\rho\in\cS_2^R(\cH)$,
\[
\Sigma_\rho+\frac{i}{2}\Theta\succeq0.
\]
Conversely, if
$\Sigma\in\cM_{2m,\mathrm{sym}}(\R)$ satisfies
$\Sigma+\frac{i}{2}\Theta\succeq0$, then for every
$\mathbf m\in\R^{2m}$ there exists a Gaussian state with mean
$\mathbf m$ and covariance $\Sigma$.
\end{theorem}
Since the Weyl system is irreducible, the characteristic function
determines the state uniquely. Thus every admissible pair
$(\mathbf m,\Sigma)$ determines a unique Gaussian state, which allows us to introduce the following projection.
\begin{deff}[Gaussian projection]
\label{def:gaussian-projection}
For $\rho\in\cS_2^R(\cH)$, the \emph{Gaussian projection}
$\mathcal G(\rho)$ is the unique Gaussian state with the same mean and
covariance as $\rho$:
\[
\mathbf m_{\mathcal G(\rho)}=\mathbf m_\rho,
\qquad
\Sigma_{\mathcal G(\rho)}=\Sigma_\rho.
\]
\end{deff}

Recall from
Section~\ref{subsec:dual-transpose-CCR} that, for $\tau\in\{\pm1\}$, $L_2^{(\tau)}(\R^m)$ denotes either
$L_2(\R^m)$ or its dual $L_2(\R^m)^*$, with canonical generators
$\hat{\mathbf R}^{(\tau)}=\hat{\mathbf R}$ for $\tau=+1$ and
$\hat{\mathbf R}^{(\tau)}=\hat{\mathbf R}^{\top}$ for $\tau=-1$.
For either choice of $\tau$, the possible covariance matrices of states on $L_2^{(\tau)}(\R^m)$ are precisely the matrices in
$\mathrm{Cov}_m^{\mathrm Q}$ defined in \eqref{eq:covqm}. Indeed, the Robertson--Schr\"odinger condition is
$\Sigma+\frac{i\tau\hbar}{2}\Omega_m\succeq0$, and the two choices of $\tau$ are equivalent by complex conjugation since $\Sigma$ is real symmetric. In particular, the notions of Gaussian state and Gaussian projection apply to both CCR representations. Moreover, transposition preserves Gaussianity and, whenever the Gaussian projection is defined,
$\mathcal G(\hat\sigma^\top)=\mathcal G(\hat\sigma)^\top$.

\subsubsection{Symplectic structure and Gaussian unitaries}
\label{subsubsec:symplectic-coherent}

A matrix $S\in\cM_{2m}(\R)$ is called symplectic if $S\Omega_mS^\top=\Omega_m$. The group of such matrices is denoted by $\Sp(2m,\R)$. The basic normal form for positive covariance matrices is provided by
Williamson's theorem; see \cite{williamson1936algebraic,serafini2023quantum,weedbrook2012gaussian}.

\begin{theorem}[Williamson's theorem]
\label{thm:williamson}
For every
$\Sigma\in\cM_{2m,\mathrm{sym}}(\R)$ with $\Sigma\succ0$, there exist
$S\in\Sp(2m,\R)$ and positive numbers
$\nu_1,\ldots,\nu_m$, unique up to permutation, such that
\[
\Sigma
=
S\left(
\bigoplus_{j=1}^m{}_D\nu_j\I_2
\right)S^\top.
\]
The numbers $\nu_1,\ldots,\nu_m$ are called the symplectic eigenvalues
of $\Sigma$.
\end{theorem}

For $\Sigma\in\mathrm{Cov}_m^{\mathrm Q}$, all symplectic eigenvalues satisfy $\nu_j\geq\hbar/2$; see, e.g., \cite[Sec.~II]{weedbrook2012gaussian}. Accordingly, a Gaussian state with covariance matrix $\Sigma$ is pure if and only if $\nu_j=\hbar/2$ for every $j\in[m]$, and faithful if and
only if $\nu_j>\hbar/2$ for every $j\in[m]$. We call a mixed Gaussian state \emph{partially pure} if at least one, but not all, of its symplectic eigenvalues equals $\hbar/2$.

For $\mathbf z\in\R^{2m}$, the Weyl operator $W_{\mathrm{CCR}}\!\left(\frac{\Omega_m\mathbf z}{\hbar}\right)$ implements the phase-space displacement $W_{\mathrm{CCR}}\!\left(\frac{\Omega_m\mathbf z}{\hbar}\right)^\dagger
\hat{\mathbf R}
W_{\mathrm{CCR}}\!\left(\frac{\Omega_m\mathbf z}{\hbar}\right)
=
\hat{\mathbf R}+\mathbf z$. Thus Weyl displacements shift the mean by $\mathbf z$ while leaving the covariance matrix unchanged.

Symplectic transformations are unitarily implemented on the standard CCR system. More precisely, for every
$S\in\Sp(2m,\R)$ there exists a unitary $\hat U_S$,
unique up to phase, such that $\hat U_S^\dagger\hat{\mathbf R}\hat U_S=S\hat{\mathbf R}$. Together with Weyl displacements, these unitary implementations give the Gaussian unitaries; see
\cite[Sec.~II.A.2]{weedbrook2012gaussian} and \cite[Secs.~3.2.1--3.2.2]{serafini2023quantum}.
\begin{deff}[Gaussian unitary]
\label{def:gaussian-unitary}
A unitary $\hat U\in\mathcal U(L_2(\R^m))$ is called Gaussian if its action on the canonical variables is affine symplectic, namely $\hat U^\dagger\hat{\mathbf R}\hat U=S\hat{\mathbf R}+\mathbf z$ for some $S\in\Sp(2m,\R)$ and
$\mathbf z\in\R^{2m}$. Consequently,
$$
\mathbf m_{\hat U\rho\hat U^\dagger}
=
S\mathbf m_\rho+\mathbf z,
\qquad
\Sigma_{\hat U\rho\hat U^\dagger}
=
S\Sigma_\rho S^\top.
$$
\end{deff}

\paragraph{Coherent states.}
Let $\ket{\boldsymbol 0}:=|0\rangle^{\otimes m}$ denote the $m$-mode Fock ground state, which is also called the vacuum state. It has zero mean and covariance
$(\hbar/2)\I_{2m}$; coherent states are obtained by displacing the Fock ground state and have the same covariance \cite[Secs.~II.B.1--II.B.2]{weedbrook2012gaussian}. For $\mathbf m\in\R^{2m}$, the coherent state with mean $\mathbf m$ is defined by
\[
\proj{\mathbf m}:=
W_{\mathrm{CCR}}\!\left(\frac{\Omega_m\mathbf m}{\hbar}\right)
\proj{\boldsymbol 0}
W_{\mathrm{CCR}}\!\left(\frac{\Omega_m\mathbf m}{\hbar}\right)^\dagger.
\]
A state $\hat\rho\in\cS_2^R(L_2(\R^m))$ satisfies $\mathbf m_{\hat\rho}=\mathbf m$ and $\Sigma_{\hat\rho}=\frac{\hbar}{2}\I_{2m}$ if and only if
\[
\hat\rho=\proj{\mathbf m}.
\]

\subsubsection{Bipartite Gaussian states and Gaussian couplings}
\label{subsubsec:bipartite-gaussian-couplings}

For $\tau\in\{\pm1\}$, consider the bipartite space
$L_2(\R^m)\otimes L_2^{(\tau)}(\R^m)$ and the tensor-product Weyl
system. We write
\[
W_{\mathrm{CCR}}^{(+1)}:=W_{\mathrm{CCR}},
\qquad
W_{\mathrm{CCR}}^{(-1)}:=W_{\mathrm{CCR}}^\top.
\]
The tensor-product Weyl system is
\[
W_\tau(\boldsymbol\xi,\boldsymbol\eta)
:=
W_{\mathrm{CCR}}(\boldsymbol\xi)
\otimes
W_{\mathrm{CCR}}^{(\tau)}(\boldsymbol\eta),
\qquad
\boldsymbol\xi,\boldsymbol\eta\in\R^{2m}.
\]
Its canonical generators are the components of $\hat{\mathbf R}_\tau$
defined in \eqref{eq:bitau}, and its commutation matrix is
$\hbar\bigl(\Omega_m\oplus_D(\tau\Omega_m)\bigr)$. Thus a state on the
bipartite system is Gaussian precisely in the sense of
Definition~\ref{def:weyl-characteristic-gaussian-state} for this
tensor-product Weyl representation.

Let
$\hat\pi\in
\cS_2^R\bigl(L_2(\R^m)\otimes L_2^{(\tau)}(\R^m)\bigr)$
have mean and covariance
\[
\mathbf m_{\hat\pi}
=
\begin{pmatrix}
\mathbf m_1\\
\mathbf m_2
\end{pmatrix},
\qquad
\Sigma_{\hat\pi}
=
\begin{pmatrix}
\Sigma_{11}&A\\
A^\top&\Sigma_{22}
\end{pmatrix}.
\]

\begin{prop}[Reduction of moments and Gaussianity]
\label{prop:reduction-moments}
The reduced states
$\hat\pi_1:=\Tr_2\hat\pi$ and
$\hat\pi_2:=\Tr_1\hat\pi$ have moments
\[
\mathbf m_{\hat\pi_1}=\mathbf m_1,
\qquad
\Sigma_{\hat\pi_1}=\Sigma_{11},
\qquad
\mathbf m_{\hat\pi_2}=\mathbf m_2,
\qquad
\Sigma_{\hat\pi_2}=\Sigma_{22}.
\]
If $\hat\pi$ is Gaussian, then both reduced states are Gaussian.
\end{prop}

\begin{proof}
The moment identities follow from
Proposition~\ref{prop:partialtrace-unbounded}. For Gaussianity, the
reduced characteristic functions are obtained by restriction,
$\chi_{\hat\pi_1}(\boldsymbol\xi)
=
\chi_{\hat\pi}(\boldsymbol\xi,0)$ and $\chi_{\hat\pi_2}(\boldsymbol\eta)
=
\chi_{\hat\pi}(0,\boldsymbol\eta)$, and are therefore Gaussian.
\end{proof}

In particular, Gaussian projection commutes with reduction:
\[
\Tr_2\mathcal G(\hat\pi)
=
\mathcal G(\Tr_2\hat\pi),
\qquad
\Tr_1\mathcal G(\hat\pi)
=
\mathcal G(\Tr_1\hat\pi).
\]

For Gaussian states
$\hat\rho^G,\hat\sigma^G\in\cS^G(L_2(\R^m))$, we define the Gaussian
coupling set by
\begin{equation}
\Gamma_\tau^G(\hat\rho^G,\hat\sigma^G)
:=
\Gamma_\tau(\hat\rho^G,\hat\sigma^G)
\cap
\cS^G\bigl(
L_2(\R^m)\otimes L_2^{(\tau)}(\R^m)
\bigr).
\label{eq:gaussian-coupling-set}
\end{equation}
If a Gaussian coupling has marginal covariances
$\Sigma_1,\Sigma_2$ and cross-covariance $A$, its covariance is
$\smat{\Sigma_1&A\\A^\top&\Sigma_2}$, and admissibility is equivalent to $A\in\mathrm{Corr}_\tau(\Sigma_1,\Sigma_2)$ defined in \eqref{eq:corr}.
\paragraph{Quadratic costs and Gaussian projection.}
For any coupling whose marginals have finite canonical second moments, the bipartite state also has finite canonical second moments, since the expectations of the two one-body square-sums are determined by its marginals. The quadratic transport cost can therefore be expressed entirely in terms of these moments.

\begin{prop}[Quadratic costs depend only on moments]
\label{prop:stateTrC}
Let $\tau\in\{\pm1\}$ and let
$\hat\pi\in
\cS_2^R\bigl(
L_2(\R^m)\otimes L_2^{(\tau)}(\R^m)
\bigr)$
have mean and covariance as above. Then
\[
\Tr\bigl(\hat C^{(\tau)}\hat\pi\bigr)
=
\Tr(\Sigma_{11})
+
\Tr(\Sigma_{22})
-
2\Tr(A)
+
\|\mathbf m_1-\mathbf m_2\|^2.
\]
\end{prop}

\begin{proof}
Expanding the quadratic cost \eqref{eq:tau-cost} gives, for each
$j\in[2m]$,
\[
\Tr\!\left[
\bigl(
\hat R_j\otimes\I
-
\I\otimes\hat R_j^{(\tau)}
\bigr)^2
\hat\pi
\right]
=
(\Sigma_{11})_{jj}
+
(\Sigma_{22})_{jj}
-
2A_{jj}
+
(m_{1,j}-m_{2,j})^2.
\]
Summing over $j$ proves the formula.
\end{proof}

\begin{prop}[Gaussian projection decreases the $2$-QW cost]
\label{prop:gaussian-projection-lower}
Let
$\hat\rho,\hat\sigma\in\cS_2^R(L_2(\R^m))$. Then, for every
$\tau\in\{\pm1\}$,
\[
\QW_{2,\tau}^2
\bigl(
\mathcal G(\hat\rho),
\mathcal G(\hat\sigma)
\bigr)
\leq
\QW_{2,\tau}^2(\hat\rho,\hat\sigma).
\]
\end{prop}

\begin{proof}
For any
$\hat\pi\in\Gamma_\tau(\hat\rho,\hat\sigma)$,
Proposition~\ref{prop:reduction-moments} implies that
$\mathcal G(\hat\pi)$ is a coupling of
$\mathcal G(\hat\rho)$ and $\mathcal G(\hat\sigma)$ in the
$\tau$-formulation. Since Gaussian projection preserves the mean and
covariance,
Proposition~\ref{prop:stateTrC} gives
\[
\Tr\bigl(
\hat C^{(\tau)}\mathcal G(\hat\pi)
\bigr)
=
\Tr\bigl(
\hat C^{(\tau)}\hat\pi
\bigr).
\]
Taking the infimum over $\hat\pi$ proves the claim.
\end{proof}

\begin{theorem}[Gaussian optimal couplings]
\label{thm:Gaussian-opt-plan}
Let
$\hat\rho^G,\hat\sigma^G\in\cS^G(L_2(\R^m))$. Then, for every
$\tau\in\{\pm1\}$, the corresponding $2$-QW transport problem admits a
Gaussian optimal coupling.
\end{theorem}

\begin{proof}
Let $\hat\pi_{\mathrm{opt}}$ be an optimal coupling, whose existence
follows from the general existence results for the two $2$-QW
transport problems
\cite{caglioti2021towards,DPaTre19}. Since its Gaussian marginals have finite canonical second moments,
the preceding observation implies that
$\hat\pi_{\mathrm{opt}}$ also has finite canonical second moments.
Hence its Gaussian projection
$\mathcal G(\hat\pi_{\mathrm{opt}})$ is well defined and has the same
Gaussian marginals by
Proposition~\ref{prop:reduction-moments}, and
Proposition~\ref{prop:stateTrC} gives
\[
\Tr\bigl(
\hat C^{(\tau)}
\mathcal G(\hat\pi_{\mathrm{opt}})
\bigr)
=
\Tr\bigl(
\hat C^{(\tau)}
\hat\pi_{\mathrm{opt}}
\bigr).
\]
Hence $\mathcal G(\hat\pi_{\mathrm{opt}})$ is a Gaussian optimal
coupling. 
\end{proof}
 
\subsection{Gaussian reduction of the \texorpdfstring{$2$}{2}-QW barycenter problem}
\label{subsec:gaussian-barycenter-reduction}

Throughout this subsection, let $m,N\in\N$, let
$\alpha_s>0$, $s\in[N]$, satisfy $\sum_{s=1}^N\alpha_s=1$, and let $\hat{\boldsymbol\sigma}^G=(\hat\sigma_s^G)_{s\in[N]}$ be Gaussian states on $L_2(\R^m)$ with mean vectors $\mathbf m_s\in\R^{2m}$ and covariance matrices $\Sigma_s\in\mathrm{Cov}_m^{\mathrm Q}$. By Proposition~\ref{prop:finitesecondemomentrho}, every barycenter has finite canonical second moments. More precisely, the lower bound in Proposition~\ref{prop:quadrature_bounds} shows that every state with finite barycenter objective belongs to $\cS_2^R(L_2(\R^m))$. Consequently, the outer infimum in
\eqref{eq:tau-barycenter-intro} is unchanged if its search space is restricted to $\cS_2^R(L_2(\R^m))$. In particular, the Gaussian projection of every barycenter is well defined.

\paragraph{Gaussian projection of barycenters.}
The projection inequality from
Proposition~\ref{prop:gaussian-projection-lower} immediately yields a
useful property of arbitrary, possibly non-Gaussian, barycenters.

\begin{prop}[Gaussian projection of barycenters and optimal couplings]
\label{prop:gaussian-projection-barycenter}
Let $\tau\in\{\pm1\}$ and let $\hat\rho$ be a $2$-QW barycenter of
$\hat{\boldsymbol\sigma}^G$. Then $\mathcal G(\hat\rho)$ is also a
barycenter and
\[
\QW_{2,\tau}^2
\bigl(\mathcal G(\hat\rho),\hat\sigma_s^G\bigr)
=
\QW_{2,\tau}^2
\bigl(\hat\rho,\hat\sigma_s^G\bigr),
\qquad s\in[N].
\]
Moreover, if
$\hat\pi_s\in\Gamma_\tau(\hat\rho,\hat\sigma_s^G)$ is any optimal
coupling, then $\mathcal G(\hat\pi_s)$ is an optimal Gaussian coupling
between $\mathcal G(\hat\rho)$ and $\hat\sigma_s^G$.
\end{prop}

\begin{proof}
Since $\mathcal G(\hat\sigma_s^G)=\hat\sigma_s^G$,
Proposition~\ref{prop:gaussian-projection-lower} gives
\[
\sum_{s=1}^N \alpha_s\, \QW_{2,\tau}^2(\mathcal G(\hat\rho),\hat\sigma_s^G)
\leq
\sum_{s=1}^N \alpha_s\, \QW_{2,\tau}^2(\hat\rho,\hat\sigma_s^G).
\]
As $\hat\rho$ is a barycenter, equality must hold and
$\mathcal G(\hat\rho)$ is also a barycenter. Since every
$\alpha_s>0$ and the corresponding pairwise inequalities hold
termwise, equality of the weighted sums implies
\[
\QW_{2,\tau}^2
\bigl(\mathcal G(\hat\rho),\hat\sigma_s^G\bigr)
=
\QW_{2,\tau}^2
\bigl(\hat\rho,\hat\sigma_s^G\bigr),
\qquad s\in[N].
\]

Let $\hat\pi_s$ be optimal between $\hat\rho$ and
$\hat\sigma_s^G$. By Proposition~\ref{prop:reduction-moments},
$\mathcal G(\hat\pi_s)$ is a coupling between
$\mathcal G(\hat\rho)$ and $\hat\sigma_s^G$, and
Proposition~\ref{prop:stateTrC} gives
\[
\Tr\!\left(
\hat C^{(\tau)}\mathcal G(\hat\pi_s)
\right)
=
\Tr\!\left(
\hat C^{(\tau)}\hat\pi_s
\right).
\]
The preceding pairwise equality therefore shows that
$\mathcal G(\hat\pi_s)$ is optimal.
\end{proof}

\paragraph{Covariance reduction.}
We next express the transport cost between Gaussian states entirely in
terms of their first and second moments.

\begin{prop}[Gaussian $2$-QW covariance reduction]
\label{prop:gaussian-pairwise-covariance}
Let $\hat\rho^G,\hat\sigma^G\in\cS^G(L_2(\R^m))$ have means
$\mathbf m,\mathbf m_\sigma$ and covariances
$\Sigma,\Sigma_\sigma$, respectively. Then, for
$\tau\in\{\pm1\}$,
\[
\QW_{2,\tau}^2(\hat\rho^G,\hat\sigma^G)
=
\|\mathbf m-\mathbf m_\sigma\|^2
+
\min_{A\in\mathrm{Corr}_\tau(\Sigma,\Sigma_\sigma)}
\left\{
\Tr(\Sigma)+\Tr(\Sigma_\sigma)-2\Tr(A)
\right\}.
\]
\end{prop}

\begin{proof}
By Theorem~\ref{thm:Gaussian-opt-plan}, an optimal coupling may be
chosen Gaussian. By Proposition~\ref{prop:reduction-moments}, its mean
is $(\mathbf m,\mathbf m_\sigma)$ and its covariance has block form
\[
\begin{pmatrix}
\Sigma&A\\
A^\top&\Sigma_\sigma
\end{pmatrix},
\qquad
A\in\mathrm{Corr}_\tau(\Sigma,\Sigma_\sigma).
\]
Proposition~\ref{prop:stateTrC} then gives the stated objective.
Conversely, every
$A\in\mathrm{Corr}_\tau(\Sigma,\Sigma_\sigma)$ determines an admissible
Gaussian bipartite state with these moments, whose marginals are
$\hat\rho^G$ and $\hat\sigma^G$. Hence the optimization over Gaussian
couplings is exactly the stated finite-dimensional problem.
\end{proof}

We can now prove the Gaussian reduction stated in
Theorem~\ref{thm:gaussian-barycenter-intro}.

\begin{proof}[Proof of Theorem~\ref{thm:gaussian-barycenter-intro}]
By the existence result of Section~\ref{sec:qot-bary}, the barycenter
problem admits a minimizer. Proposition~\ref{prop:gaussian-projection-barycenter}
shows that the Gaussian projection of any barycenter is again a
barycenter. Thus a Gaussian barycenter exists.

Let $\hat\rho^G$ be a Gaussian state with mean
$\mathbf m\in\R^{2m}$ and covariance
$\Sigma\in\mathrm{Cov}_m^{\mathrm Q}$. Applying
Proposition~\ref{prop:gaussian-pairwise-covariance} to each input gives
\[
\sum_{s=1}^N
\alpha_s\,
\QW_{2,\tau}^2(\hat\rho^G,\hat\sigma_s^G)
=
\min_{\substack{
A_s\in\mathrm{Corr}_\tau(\Sigma,\Sigma_s),\\
s\in[N]}}
\sum_{s=1}^N
\alpha_s
\left[
\Tr(\Sigma)+\Tr(\Sigma_s)-2\Tr(A_s)
+
\|\mathbf m-\mathbf m_s\|^2
\right].
\]
Therefore the Gaussian barycenter problem is equivalent to
\begin{equation}
\min_{\substack{
\mathbf m\in\R^{2m},\,
\Sigma\in\mathrm{Cov}_m^{\mathrm Q},\\
A_s\in\mathrm{Corr}_\tau(\Sigma,\Sigma_s),\ s\in[N]}}
\left\{
\Tr(\Sigma)
+
\sum_{s=1}^N
\alpha_s
\bigl(\Tr(\Sigma_s)-2\Tr(A_s)\bigr)
+
\sum_{s=1}^N
\alpha_s\|\mathbf m-\mathbf m_s\|^2
\right\}.
\end{equation}

Set $\overline{\mathbf m}:=\sum_{s=1}^N\alpha_s\mathbf m_s$. The identity
\[
\sum_{s=1}^N\alpha_s\|\mathbf m-\mathbf m_s\|^2
=
\|\mathbf m-\overline{\mathbf m}\|^2
+
\sum_{s=1}^N
\alpha_s\|\mathbf m_s-\overline{\mathbf m}\|^2
\]
shows that the mean is uniquely minimized at
$\mathbf m=\overline{\mathbf m}$. After minimizing over the mean variable, the Gaussian barycenter problem reduces to \eqref{eq:qw2cov}. Thus every Gaussian barycenter has mean
$\overline{\mathbf m}$ and covariance equal to the first component of
an optimizer of \eqref{eq:qw2cov}. Conversely,
every optimizer of \eqref{eq:qw2cov}, together with
the mean $\overline{\mathbf m}$, determines a Gaussian barycenter.
This proves the theorem.
\end{proof}
The same argument also gives a useful statement for arbitrary
barycenters, which will be used in the uniqueness analysis below.
\begin{corollary}[Covariance optimality for arbitrary barycenters]
\label{cor:arbitrary-barycenter-covariance}
Fix $\tau\in\{\pm1\}$ and let $\hat\rho$ be any, possibly
non-Gaussian, barycenter of $\hat{\boldsymbol\sigma}^G$. Then $\mathbf m_{\hat\rho}=\overline{\mathbf m}$. Moreover, if for each $s\in[N]$, $\hat\pi_s\in\Gamma_\tau(\hat\rho,\hat\sigma_s^G)$ is any optimal
coupling and $A_s$ denotes its cross-covariance, then
\[
\bigl(
\Sigma_{\hat\rho},A_1,\ldots,A_N
\bigr)
\]
is an optimizer of the covariance problem
\eqref{eq:qw2cov}.
\end{corollary}

\begin{proof}
By Proposition~\ref{prop:gaussian-projection-barycenter},
$\mathcal G(\hat\rho)$ is a Gaussian barycenter. Hence
Theorem~\ref{thm:gaussian-barycenter-intro} gives
$\mathbf m_{\hat\rho}
=
\mathbf m_{\mathcal G(\hat\rho)}
=
\overline{\mathbf m}$.

For each $s$, Proposition~\ref{prop:gaussian-projection-barycenter}
also shows that $\mathcal G(\hat\pi_s)$ is an optimal Gaussian coupling
between $\mathcal G(\hat\rho)$ and $\hat\sigma_s^G$. Gaussian
projection preserves the covariance of $\hat\pi_s$, so
$\mathcal G(\hat\pi_s)$ has block covariance
\[
\begin{pmatrix}
\Sigma_{\hat\rho}&A_s\\
A_s^\top&\Sigma_s
\end{pmatrix}.
\]
Since $\mathcal G(\hat\rho)$ is a Gaussian barycenter, the covariance
reduction above implies that
$(\Sigma_{\hat\rho},A_1,\ldots,A_N)$ attains
\eqref{eq:qw2cov}.
\end{proof}
 
\subsection{Uniqueness with one faithful Gaussian input}
\label{subsec:faithful-gaussian-global-uniqueness}
We now strengthen the Gaussian reduction obtained in the previous
subsection to uniqueness among all quantum states. Assume that at least
one prescribed Gaussian state is faithful. Let $k\in[N]$ and, after
relabeling the inputs if necessary, assume that the first $k$ inputs
are faithful, so that
\begin{equation}
\Sigma_r+\frac{i\hbar}{2}\Omega_m\succ0,
\qquad r\in[k].
\label{eq:faithful-inputs}
\end{equation}
Since each $\Sigma_r$ is real symmetric, complex conjugation gives
$\Sigma_r+\frac{i\tau\hbar}{2}\Omega_m\succ0$ for
$\tau\in\{\pm1\}$.

The proof has two layers. We first use covariance duality and
complementary slackness for the faithful inputs to obtain uniqueness of
the barycenter covariance and of the cross-covariances with the
faithful inputs.

We then establish a reconstruction principle that is
independent of faithfulness and of the barycenter problem. At fixed
first moments, every state $\hat\pi$ satisfying a covariance slackness relation with
a positive semidefinite matrix $Z$ can be written as $\hat\pi=\Phi_G(\hat\eta)$, where the fixed Gaussian channel $\Phi_G$ satisfies
$\rev{\Ran(\mathbb X_\Phi)}=\ker(\Real Z)$. Any Gaussian channel whose phase-space transformation matrix is
injective on $\ker(\Real Z)$ therefore reconstructs the state from its
output; in particular, a Gaussian output forces the original state to
be Gaussian. In the faithful-input setting, dual attainment for the
covariance problem and complementary slackness supply a dual optimizer
satisfying this injectivity condition after reduction to the standard
CCR. This proves uniqueness among all quantum states.
\subsubsection{Covariance duality and covariance uniqueness}
\label{subsubsec:faithful-covariance-dual}
We now study the finite-dimensional covariance problem
\eqref{eq:qw2cov} through its dual formulation. Under the faithful-input
assumption~\eqref{eq:faithful-inputs}, the dual problem satisfies strong
duality and its optimum is attained. The additional point, compared with
the general state-level duality established in
Theorem~\ref{thm:energy-constrained-barycenter-duality} and
Corollary~\ref{cor:unconstrained-barycenter-duality}, is precisely this
dual attainment: the earlier results identify the optimal value, but do
not assert the existence of a dual optimizer. Here, one may therefore
fix a single dual optimizer and pair it with every primal optimizer,
yielding a common complementary-slackness relation.

This complementary-slackness relation is the \rev{key ingredient in}
the uniqueness argument. It yields uniqueness of the optimal
barycenter covariance and of the cross-covariances associated with the
faithful inputs.

We begin with the covariance dual problem. Fix $\tau\in\{\pm1\}$ and assume \eqref{eq:faithful-inputs}. Define the set
\[
\ccov^\tau
:=
\left\{
\mathbf V=(V_r)_{r\in[k]}
\in
\prod_{r=1}^k\cM_{4m,\mathrm{sa}}(\C):
V_r+
\begin{pmatrix}
\frac{i\hbar}{2}\Omega_m&0\\
0&\Sigma_r+\frac{i\tau\hbar}{2}\Omega_m
\end{pmatrix}
\succeq0,\quad r\in[k]
\right\},
\]
and the linear operator $\lcc:
\cM_{2m,\mathrm{sym}}(\R)
\times
\prod_{r=1}^k\cM_{2m}(\R)
\rightarrow
\prod_{r=1}^k\cM_{4m,\mathrm{sa}}(\C)$ by
\[
\lcc(\Sigma,A_1,\ldots,A_k)
:=
\left(
\begin{pmatrix}
\Sigma&A_r\\
A_r^\top&0
\end{pmatrix}
\right)_{r\in[k]}.
\]

Define
$\fc^\tau,\gc^\tau:
\prod_{r=1}^k\cM_{4m,\mathrm{sa}}(\C)
\to\R\cup\{+\infty\}$
by
\[
\fc^\tau(\mathbf V):=\iota_{\ccov^\tau}(\mathbf V)
=
\begin{cases}
0,&\mathbf V\in\ccov^\tau,\\
+\infty,&\text{otherwise},
\end{cases}
\]
\[
\begin{aligned}
\gc^\tau(\mathbf V)
:=
\inf_{\substack{
\Sigma\in\mathrm{Cov}_m^{\mathrm Q},\\
A_1,\ldots,A_N\in\cM_{2m}(\R)}}
\quad&
\Tr(\Sigma)-2\sum_{s=1}^N\alpha_s\Tr(A_s)
\\
\text{subject to}\quad&
\lcc(\Sigma,A_1,\ldots,A_k)=\mathbf V,
\\
&
A_s\in\mathrm{Corr}_\tau(\Sigma,\Sigma_s),
\qquad s=k+1,\ldots,N.
\end{aligned}
\]
If $k=N$, we instead set
\[
\gc^\tau(\mathbf V)
:=
\inf_{\substack{
\Sigma\in\cM_{2m,\mathrm{sym}}(\R),\\
A_1,\ldots,A_N\in\cM_{2m}(\R)}}
\left\{
\Tr(\Sigma)-2\sum_{s=1}^N\alpha_s\Tr(A_s):
\lcc(\Sigma,A_1,\ldots,A_N)=\mathbf V
\right\}.
\]
The functions
$\fc^\tau$ and $\gc^\tau$ are obtained by separating the covariance
constraints in \eqref{eq:qw2cov} corresponding to the faithful inputs: $\fc^\tau$ is the indicator of these constraints, while $\gc^\tau$ is the partial-value function containing the objective and the remaining covariance constraints. 
\begin{prop}[Dual formulation of the covariance matrix optimization problem]
\label{prop:covdual}
Fix $\tau\in\{\pm1\}$ and assume \eqref{eq:faithful-inputs}. For $Z_r\in\cM_{4m,\mathrm{sa}}(\C)$, we denote its
$2m\times2m$ blocks by $(Z_r)_{ij}$, $i,j\in\{1,2\}$. Recall that
$\mathsf P_{\tau,\mathrm{cov}}$ denotes the minimum value of
\eqref{eq:qw2cov}. If $k<N$, define
\begin{equation}
\begin{aligned}
\mathsf D_{\tau,\mathrm{cov}}
:=&
\sum_{s=1}^N\alpha_s\Tr(\Sigma_s)
+
\max_{\substack{
Z_r\in\cM_{4m,\geq}(\C),\ r\in[k],\\
\Real (Z_r)_{12}=-\alpha_r\I_{2m},\ r\in[k]}}
\Bigg\{
-\sum_{r=1}^k
\Real\Tr\left[
Z_r^\dagger
\begin{pmatrix}
\frac{i\hbar}{2}\Omega_m&0\\
0&\Sigma_r+\frac{i\tau\hbar}{2}\Omega_m
\end{pmatrix}
\right]
\\
&
-
\sup_{\substack{
\Sigma\in\mathrm{Cov}_m^{\mathrm Q},\\
A_s\in\mathrm{Corr}_\tau(\Sigma,\Sigma_s),\
s=k+1,\ldots,N}}
\left\{
\Tr\left[
\left(
\sum_{r=1}^k\Real (Z_r)_{11}-\I_{2m}
\right)\Sigma
\right]
+
2\sum_{s=k+1}^N\alpha_s\Tr(A_s)
\right\}
\Bigg\}.
\end{aligned}
\label{eq:covariance-dual-faithful}
\end{equation}
If $k=N$, define instead 
\begin{equation} 
\begin{aligned} 
\mathsf D_{\tau,\mathrm{cov}} := \sum_{s=1}^{N}\alpha_s\Tr(\Sigma_s) + \max_{\substack{ Z_s\in\cM_{4m,\geq}(\C),\ s\in[N],\\ \Real (Z_s)_{12}=-\alpha_s\I_{2m},\ s\in[N],\\ \sum_{s=1}^N\Real (Z_s)_{11}=\I_{2m}}} \left\{ -\sum_{s=1}^N \Real\Tr\left[ Z_s^\dagger \begin{pmatrix} \frac{i\hbar}{2}\Omega_m&0\\ 0&\Sigma_s+\frac{i\tau\hbar}{2}\Omega_m 
\end{pmatrix} \right] \right\}. 
\end{aligned} \label{eq:covariance-dual-fully-faithful} \end{equation} 

Strong duality holds:
$\mathsf P_{\tau,\mathrm{cov}}=
\mathsf D_{\tau,\mathrm{cov}}$. Moreover, the corresponding dual maximum is attained.
\end{prop}
\begin{proof}
By the definitions of 
$\fc^\tau$ and $\gc^\tau$,
\[
\mathsf P_{\tau,\mathrm{cov}}
=
\sum_{s=1}^N\alpha_s\Tr(\Sigma_s)
+
\inf_{\mathbf V}
\left\{
\fc^\tau(\mathbf V)+\gc^\tau(\mathbf V)
\right\}.
\]

Choose $t>\hbar/2$ and set $\mathbf V_0:=\lcc(t\I_{2m},0,\ldots,0)$. By \eqref{eq:faithful-inputs},
$(\mathbf V_0)_r+
\smat{
\frac{i\hbar}{2}\Omega_m&0\\
0&\Sigma_r+\frac{i\tau\hbar}{2}\Omega_m}
\succ0$ for all $r\in[k]$, so $\mathbf V_0\in\operatorname{int}(\Dom\fc^\tau)$.
Moreover, $\mathbf V_0\in\Dom\gc^\tau$: if $k<N$,
this follows from $0\in\mathrm{Corr}_\tau(t\I_{2m},\Sigma_s)$ for $k+1\leq s\leq N$, while for $k=N$ it follows directly from the definition of $\gc^\tau$. Since $\fc^\tau$ vanishes on a neighborhood of $\mathbf V_0$, it is continuous at $\mathbf V_0$.
Fenchel--Rockafellar Theorem~\ref{app:FRTheorem} therefore yields
\[
\mathsf P_{\tau,\mathrm{cov}}
=
\sum_{s=1}^N\alpha_s\Tr(\Sigma_s)
+
\max_{\mathbf Z}
\left\{
-(\fc^\tau)^*(-\mathbf Z)
-
(\gc^\tau)^*(\mathbf Z)
\right\},
\]
and the maximum is attained. By
Proposition~\ref{prop:dualfunccov},
$(\fc^\tau)^*(-\mathbf Z)<+\infty$ requires
$Z_r\succeq0$ for every $r\in[k]$, while
$(\gc^\tau)^*(\mathbf Z)<+\infty$ requires $\Real (Z_r)_{12}
=-\alpha_r\I_{2m}$ for all $r\in[k]$. 
Substituting the formulas for $(\fc^\tau)^*$ and $(\gc^\tau)^*$
from Proposition~\ref{prop:dualfunccov} gives
\eqref{eq:covariance-dual-faithful} when $k<N$. When $k=N$, the
remaining supremum over
$\Sigma\in\cM_{2m,\mathrm{sym}}(\R)$ is finite if and only if $\sum_{s=1}^N\Real (Z_s)_{11}=\I_{2m}$, in which case it equals zero. This gives
\eqref{eq:covariance-dual-fully-faithful}. Hence, in either case,
$\mathsf P_{\tau,\mathrm{cov}}=\mathsf D_{\tau,\mathrm{cov}}$, and
the corresponding dual maximum is attained.
\end{proof}
Strong duality and dual attainment yield the following complementary-slackness relation.
\begin{corollary}[Complementary slackness for the covariance dual]
\label{cor:faithful-complementary-slackness}
Fix $\tau\in\{\pm1\}$ and assume \eqref{eq:faithful-inputs}. Let $\mathbf Z^\star=(Z_r^\star)_{r\in[k]}$ be any dual optimizer from Proposition~\ref{prop:covdual}. Then, for every optimizer
$(\Sigma^\star,A_1^\star,\ldots,A_N^\star)$ of
\eqref{eq:qw2cov},
\begin{equation}
\begin{pmatrix}
\Sigma^\star+\frac{i\hbar}{2}\Omega_m&A_r^\star\\
(A_r^\star)^\top&
\Sigma_r+\frac{i\tau\hbar}{2}\Omega_m
\end{pmatrix}
Z_r^\star
=
0,
\qquad r\in[k].
\label{eq:faithful-complementary-slackness}
\end{equation}
\end{corollary}
\begin{proof}
Let $\mathbf V^\star:=\lcc(\Sigma^\star,A_1^\star,\ldots,A_k^\star)$. By strong duality and optimality, equality holds in the Fenchel
inequality for $\fc^\tau$, hence
\[
\fc^\tau(\mathbf V^\star)
+
(\fc^\tau)^*(-\mathbf Z^\star)
=
-\Real\sum_{r=1}^k
\Tr\bigl[(Z_r^\star)^\dagger V_r^\star\bigr].
\]
Since $\fc^\tau(\mathbf V^\star)=0$, the formula for
$(\fc^\tau)^*$ gives
\[
\sum_{r=1}^k
\Real\Tr\left[
(Z_r^\star)^\dagger
\begin{pmatrix}
\Sigma^\star+\frac{i\hbar}{2}\Omega_m&A_r^\star\\
(A_r^\star)^\top&
\Sigma_r+\frac{i\tau\hbar}{2}\Omega_m
\end{pmatrix}
\right]
=
0.
\]
Each summand is nonnegative. Hence every summand vanishes, and
\eqref{eq:faithful-complementary-slackness} follows from the fact that
$P,Q\succeq0$ and $\Tr(PQ)=0$ imply $PQ=0$.
\end{proof}

The complementary-slackness relation above yields the desired
covariance uniqueness.
\begin{prop}[Uniqueness of the barycenter covariance and faithful-input cross-covariances]
\label{prop:faithful-covariance-uniqueness}
Fix $\tau\in\{\pm1\}$ and assume \eqref{eq:faithful-inputs}. Then there
exists a unique $\overline\Sigma_\tau\in\mathrm{Cov}_m^{\mathrm Q}$ such that every optimizer
$(\Sigma^\star,A_1^\star,\ldots,A_N^\star)$ of
\eqref{eq:qw2cov} satisfies $\Sigma^\star=\overline\Sigma_\tau$.
Moreover, for every $r\in[k]$, there exists a unique $\overline A_{r,\tau}\in\cM_{2m}(\R)$ such that every optimizer satisfies $A_r^\star=\overline A_{r,\tau}$.
\end{prop}

\begin{proof}
Let
$(\Sigma^{(j)},A_1^{(j)},\ldots,A_N^{(j)})$, $j=1,2$, be two
optimizers of \eqref{eq:qw2cov}, and let
$\mathbf Z^\star=(Z_r^\star)_{r\in[k]}$ be any dual optimizer from
Proposition~\ref{prop:covdual}. Fix $r\in[k]$. By
Corollary~\ref{cor:faithful-complementary-slackness},
\[
\begin{pmatrix}
\Sigma^{(j)}+\frac{i\hbar}{2}\Omega_m&A_r^{(j)}\\
(A_r^{(j)})^\top&
\Sigma_r+\frac{i\tau\hbar}{2}\Omega_m
\end{pmatrix}
Z_r^\star
=
0,
\qquad
j=1,2.
\]
Set $\Delta\Sigma:=\Sigma^{(1)}-\Sigma^{(2)}$ and $\Delta A_r:=A_r^{(1)}-A_r^{(2)}$. Subtracting the two identities gives
\[
\begin{pmatrix}
\Delta\Sigma&\Delta A_r\\
\Delta A_r^\top&0
\end{pmatrix}
Z_r^\star
=
0.
\]
The $(2,2)$-block yields $\Delta A_r^\top(Z_r^\star)_{12}=0$. Since $\Delta A_r$ is real and $\Real(Z_r^\star)_{12}=-\alpha_r\I_{2m}$, we obtain $0
=
\Delta A_r^\top\Real(Z_r^\star)_{12}
=
-\alpha_r\Delta A_r^\top$. Hence $\Delta A_r=0$. The $(1,2)$-block then gives $\Delta\Sigma(Z_r^\star)_{12}=0$, and therefore
$0=\Delta\Sigma\Real(Z_r^\star)_{12}=-\alpha_r\Delta\Sigma$. Thus $\Delta\Sigma=0$.

Since $r\in[k]$ was arbitrary, the same argument gives uniqueness of
the cross-covariance associated with every faithful input.
\end{proof}
Therefore, we obtain uniqueness within the Gaussian class.
\begin{corollary}[Uniqueness among Gaussian barycenters]
\label{cor:faithful-unique-gaussian-minimizer}
Assume \eqref{eq:faithful-inputs}. For each $\tau\in\{\pm1\}$, there exists exactly one Gaussian $2$-QW barycenter of the inputs
$\hat\sigma_1^G,\ldots,\hat\sigma_N^G$. We denote it by
$\hat{\overline\rho}_\tau^G$. Its mean is $\overline{\mathbf m}=\sum_{s=1}^N\alpha_s\mathbf m_s$, and its covariance is the unique matrix $\overline\Sigma_\tau$ from Proposition~\ref{prop:faithful-covariance-uniqueness}.
\end{corollary}
\begin{proof}
By Theorem~\ref{thm:gaussian-barycenter-intro}, \rev{every Gaussian barycenter} has mean
$\overline{\mathbf m}=\sum_{s=1}^N\alpha_s\mathbf m_s$, and its
covariance is the first component of an optimizer of
\eqref{eq:qw2cov}. Conversely, every optimizer of
\eqref{eq:qw2cov}, together with the mean
$\overline{\mathbf m}$, determines a Gaussian barycenter.

Proposition~\ref{prop:faithful-covariance-uniqueness} shows that the
covariance component of every optimizer is the same matrix
$\overline\Sigma_\tau$. Hence every Gaussian barycenter has mean
$\overline{\mathbf m}$ and covariance $\overline\Sigma_\tau$.
Since a Gaussian state is uniquely determined by its mean and
covariance, the Gaussian barycenter is unique.
\end{proof}
\subsubsection{Reconstruction under a covariance slackness relation}
\label{subsubsec:saturated-face-reconstruction}
The covariance uniqueness in Proposition~\ref{prop:faithful-covariance-uniqueness} is not, by itself, sufficient
to determine the barycenter state, since first and second moments do not in general determine a quantum state uniquely. This motivates a more general reconstruction question: under what additional covariance-level condition can a state be uniquely recovered once its first and second moments are fixed?

We consider a state with covariance matrix
$\Sigma\in\mathrm{Cov}_{2m}^{\mathrm Q}$ and a matrix
$Z\in\cM_{4m,\geq}(\C)$ satisfying
\[
\left(
\Sigma+\frac{i\hbar}{2}\Omega_{2m}
\right)Z=0.
\]
We refer to this as a covariance slackness relation. At fixed first moments, all states satisfying this relation admit a representation through a common Gaussian channel. Moreover, if another Gaussian channel \rev{$\Psi_G$} has phase-space transformation matrix \rev{$\mathbb X_\Psi$} satisfying $\ker(\Real Z)\cap\ker(\rev{\mathbb X_\Psi})=\{0\}$, Proposition~\ref{prop:saturated-face-reconstruction} shows that its output determines the state uniquely.
\paragraph{Gaussian channels.}
We adopt the characteristic-function description compatible with the
Weyl convention of
Section~\ref{subsec:standard-ccr-representation}; see
\cite{HolevoWerner2001BosonicGaussian,CarusoEtAl2008}.

\begin{deff}[Gaussian channel]
\label{def:gaussian-channel}
Recall that a quantum channel between Hilbert spaces $\cH_1$ and
$\cH_2$ is a completely positive trace-preserving linear map
$\Phi:\mathcal T^{(1)}(\cH_1)\rightarrow
\mathcal T^{(1)}(\cH_2)$.
Let $m_2\in\N$ and $m_1\in\N\cup\{0\}$. A quantum channel
$\Phi_G:\mathcal T^{(1)}(L_2(\R^{m_1}))\rightarrow
\mathcal T^{(1)}(L_2(\R^{m_2}))$
is called Gaussian if there exist
$\mathbb{X}_G\in\R^{2m_2\times2m_1}$,
$\mathbb Y_G\in\cM_{2m_2,\mathrm{sym}}(\R)$, and
$\mathbf m_G\in\R^{2m_2}$ such that, for every
$\hat\rho\in\cS(L_2(\R^{m_1}))$ and every
$\mathbf z\in\R^{2m_2}$,
\begin{equation}
\chi_{\Phi_G(\hat\rho)}(\mathbf z)
=
\chi_{\hat\rho}(\mathbb{X}_G^\top\mathbf z)
\exp\left(
i\mathbf m_G^\top\mathbf z
-\frac12\mathbf z^\top\mathbb Y_G\mathbf z
\right).
\label{eq:gaussian-channel-characteristic-function}
\end{equation}
For parameters of this form, complete positivity is equivalent to
$\mathbb Y_G+\frac{i\hbar}{2}
(\Omega_{m_2}-\mathbb{X}_G\Omega_{m_1}\mathbb{X}_G^\top)\succeq0$.
We call $\mathbb{X}_G$ the phase-space transformation matrix of
$\Phi_G$.
\end{deff}

We allow a zero-mode input system, with $L_2(\R^0)=\C$. Its unique
state is regarded as Gaussian, with characteristic function
identically equal to one; $\Omega_0$ is the empty symplectic matrix,
and the corresponding phase-space matrices have zero input columns.
We also set $\cS_2^R(L_2(\R^0))
:=
\cS(L_2(\R^0))$, which consists of its unique state, with zero-dimensional mean and
covariance.
It follows directly from
\eqref{eq:gaussian-channel-characteristic-function} that, for every
$\hat\rho\in\cS_2^R(L_2(\R^{m_1}))$,
$\mathbf m_{\Phi_G(\hat\rho)}
=\mathbb{X}_G\mathbf m_{\hat\rho}+\mathbf m_G$ and
$\Sigma_{\Phi_G(\hat\rho)}
=\mathbb{X}_G\Sigma_{\hat\rho}\mathbb{X}_G^\top+\mathbb Y_G$.
In particular, Gaussian channels map Gaussian states to Gaussian
states.

\begin{lemma}[State reconstruction via a Gaussian channel with full-column-rank \rev{\mbox{phase-space}} transformation matrix]
\label{lem:faithful-gaussian-deconvolution}
Let $\Phi_G:\mathcal T^{(1)}(L_2(\R^{m_1}))\rightarrow
\mathcal T^{(1)}(L_2(\R^{m_2}))$
be a Gaussian channel with phase-space transformation matrix
$\mathbb{X}_G$. If $\mathbb{X}_G$ has full column rank, then $\Phi_G$
is injective on quantum states. Moreover, if $\Phi_G(\hat\rho)$ is
Gaussian, then $\hat\rho$ is Gaussian.
\end{lemma}
\begin{proof}
The zero-mode case is immediate. Otherwise,
$\mathbb{X}_G^\top$ is surjective, so there exists a linear map
$\mathbb J:\R^{2m_1}\rightarrow\R^{2m_2}$ such that
$\mathbb{X}_G^\top\mathbb J=\I_{2m_1}$.
By \eqref{eq:gaussian-channel-characteristic-function}, for every
$\hat\rho$ and $\mathbf u\in\R^{2m_1}$,
\[
\chi_{\hat\rho}(\mathbf u)
=
\chi_{\Phi_G(\hat\rho)}(\mathbb J\mathbf u)
\exp\left(
-i\mathbf m_G^\top\mathbb J\mathbf u
+
\frac12
\mathbf u^\top
\mathbb J^\top\mathbb Y_G\mathbb J
\mathbf u
\right).
\]
Thus the output characteristic function uniquely determines the input
characteristic function, proving injectivity. If
$\Phi_G(\hat\rho)$ is Gaussian, the right-hand side is Gaussian in
$\mathbf u$. In particular, its restrictions
$t\mapsto\chi_{\hat\rho}(t e_j)$ are one-dimensional Gaussian
characteristic functions, so all canonical second moments of
$\hat\rho$ are finite. Hence
$\hat\rho\in\cS_2^R(L_2(\R^{m_1}))$, and
Definition~\ref{def:weyl-characteristic-gaussian-state} implies that
$\hat\rho$ is Gaussian.
\end{proof}

The next result shows that a covariance slackness relation determines
a common Gaussian channel through which every state in the corresponding
slackness class can be represented.
\begin{prop}[Uniform Gaussian-channel representation under covariance slackness]
\label{prop:dual-support-gaussian-source}
Let $\Sigma\in\mathrm{Cov}_{2m}^{\mathrm Q}$ and
$Z\in\cM_{4m,\geq}(\C)$ satisfy the covariance slackness relation
\[
\left(
\Sigma+\frac{i\hbar}{2}\Omega_{2m}
\right)Z=0.
\]
Then $\rank(\Real Z)$ is even. Set $q:=\frac12\rank(\Real Z)$ and fix $\mathbf m\in\R^{4m}$. There exists a fixed Gaussian channel $\Phi_G:
\mathcal T^{(1)}(L_2(\R^{2m-q}))
\rightarrow
\mathcal T^{(1)}(L_2(\R^{2m}))$ with phase-space transformation matrix \rev{$\mathbb X_\Phi$} satisfying $\Ran(\rev{\mathbb X_\Phi})=\ker(\Real Z)$. Moreover,
\[
\begin{aligned}
&
\left\{
\hat\pi\in\cS_2^R(L_2(\R^{2m})) :
\begin{gathered}
\mathbf m_{\hat\pi}=\mathbf m,\\[-2pt]
\left(
\Sigma_{\hat\pi}+\frac{i\hbar}{2}\Omega_{2m}
\right)Z=0
\end{gathered}
\right\}
\\[-2pt]
&\qquad=
\left\{
\Phi_G(\hat\eta) :
\hat\eta\in\cS_2^R(L_2(\R^{2m-q})),
\ \mathbf m_{\hat\eta}=0
\right\}.
\end{aligned}
\]

\end{prop}
\begin{proof}
Since $\Sigma\in\mathrm{Cov}_{2m}^{\mathrm Q}$, Theorem~\ref{thm:rs-uncertainty-canonical} yields a Gaussian state
$\hat\pi_0\in\cS_2^R(L_2(\R^{2m}))$ with mean $\mathbf m$ and
covariance $\Sigma$. Apply Lemma~\ref{lem:faithful-dual-Williamson} to $(\hat\pi_0,\Sigma,Z)$. It follows that $\rank(\Real Z)$ is even. Therefore $q:=\frac12\rank(\Real Z)\in\mathbb{N}\cup\{0\}$. Lemma~\ref{lem:faithful-dual-Williamson} also provides a symplectic matrix $S=(S_0\ S_1)\in\Sp(4m,\R)$, where
$S_0\in\R^{4m\times2q}$ and
$S_1\in\R^{4m\times2(2m-q)}$, such that, for some $d_1,\ldots,d_q>0$,
\begin{equation}\label{eq:SZS}
S^\top ZS
=
\bigoplus_{j=1}^q{}_D
\frac{d_j}{2}
\begin{pmatrix}
1&-i\\
i&1
\end{pmatrix}
\oplus_D 0_{2(2m-q)},
\end{equation}
and $\Ran(S_1)=\ker(\Real Z)$. Fix a Gaussian unitary $\hat U_S$ implementing $S$. Define
\[
\Phi_G(\hat\rho)
:=
W_{\mathrm{CCR}}\!\left( \frac{\Omega_{2m}\mathbf m}{\hbar} \right)\hat U_S
\left(
|0\rangle\langle0|^{\otimes q}\otimes\hat\rho
\right)
\hat U_S^\dagger W_{\mathrm{CCR}}\!\left( \frac{\Omega_{2m}\mathbf m}{\hbar} \right)^\dagger.
\]
By Lemma~\ref{lem:fock-ancilla-Gaussian-channel},
$\Phi_G$ is a Gaussian channel whose phase-space transformation
matrix is \rev{$\mathbb X_\Phi=S_1$}. Consequently, $\Ran(\rev{\mathbb X_\Phi})=\ker(\Real Z)$.

We now prove the claimed equality of sets. Let
$\hat\pi\in\cS_2^R(L_2(\R^{2m}))$ satisfy $\mathbf m_{\hat\pi}=\mathbf m$ and $\left(
\Sigma_{\hat\pi}
+\frac{i\hbar}{2}\Omega_{2m}
\right)Z=0$. Define $\hat\pi'
:=
\hat U_S^\dagger
W_{\mathrm{CCR}}\!\left( \frac{\Omega_{2m}\mathbf m}{\hbar} \right)^\dagger
\hat\pi
W_{\mathrm{CCR}}\!\left( \frac{\Omega_{2m}\mathbf m}{\hbar} \right)
\hat U_S$. Then $\hat\pi'$ has zero mean, and its covariance is $V:=S^{-1}\Sigma_{\hat\pi}S^{-\top}$. By symplecticity of $S$, the covariance slackness relation becomes $\left(
V+\frac{i\hbar}{2}\Omega_{2m}
\right)S^\top ZS=0$.

For $j\in[q]$, let $v_j:=e_{2j-1}+ie_{2j}$. The range of the $j$-th nonzero block of $S^\top ZS$ in \eqref{eq:SZS} is spanned by $v_j$. Hence $\left(
V+\frac{i\hbar}{2}\Omega_{2m}
\right)v_j=0$. Since $\Omega_{2m}v_j=iv_j$, this gives $Vv_j=\frac{\hbar}{2}v_j$. Because $V$ is real and symmetric, it follows that $Ve_{2j-1}=\frac{\hbar}{2}e_{2j-1}$ and $Ve_{2j}=\frac{\hbar}{2}e_{2j}$ for $j\in[q]$, and therefore $V
=
\frac{\hbar}{2}\I_{2q}
\oplus_D V_{\mathrm{res}}$. Each of the first $q$ one-mode marginals of $\hat\pi'$ thus has
zero mean and covariance $(\hbar/2)\I_2$, and is consequently the
Fock ground state. Iterating the pure-marginal factorization, we
obtain $\hat\pi'=
|0\rangle\langle0|^{\otimes q}\otimes\hat\eta$ for some $\hat\eta\in\cS_2^R(L_2(\R^{2m-q}))$ with zero mean. It follows directly from the definition of $\Phi_G$ that $\hat\pi=\Phi_G(\hat\eta)$. This proves the inclusion from left to right.

Conversely, let $\hat\eta\in\cS_2^R(L_2(\R^{2m-q}))$ with $\mathbf m_{\hat\eta}=0$ and set $\hat\pi:=\Phi_G(\hat\eta)$. By construction of $\Phi_G$, $\mathbf m_{\hat\pi}=\mathbf m$ and $\Sigma_{\hat\pi}
=
S\left(
\frac{\hbar}{2}\I_{2q}
\oplus_D\Sigma_{\hat\eta}
\right)S^\top$. Moreover, the canonical form \eqref{eq:SZS} of $S^\top ZS$ gives $\left[
\frac{\hbar}{2}\I_{2q}
\oplus_D\Sigma_{\hat\eta}
+
\frac{i\hbar}{2}\Omega_{2m}
\right]
S^\top ZS
=0$. Conjugating this identity by $S$ and $S^{-1}$ yields $\left(
\Sigma_{\hat\pi}
+\frac{i\hbar}{2}\Omega_{2m}
\right)Z=0$. Thus $\hat\pi$ belongs to the set on the left, proving the reverse
inclusion.
\end{proof}
\begin{prop}[Injectivity and Gaussianity under covariance slackness]
\label{prop:saturated-face-reconstruction}
Let $Z\in\cM_{4m,\geq}(\C)$, and let $\Psi_G:
\mathcal T^{(1)}(L_2(\R^{2m}))
\rightarrow
\mathcal T^{(1)}(L_2(\R^\ell))$ be a Gaussian channel with phase-space transformation matrix \rev{$\mathbb X_\Psi\in\R^{2\ell\times4m}$}. Assume $\ker(\Real Z)\cap\ker(\rev{\mathbb X_\Psi})=\{0\}$. If
$\hat\pi_1,\hat\pi_2\in\cS_2^R(L_2(\R^{2m}))$ have the same mean and
satisfy $\left(
\Sigma_{\hat\pi_j}+\frac{i\hbar}{2}\Omega_{2m}
\right)Z=0$ for $j=1,2$ and $\Psi_G(\hat\pi_1)=\Psi_G(\hat\pi_2)$, then $\hat\pi_1=\hat\pi_2$. Moreover, if
$\hat\pi\in\cS_2^R(L_2(\R^{2m}))$ satisfies $\left(
\Sigma_{\hat\pi}+\frac{i\hbar}{2}\Omega_{2m}
\right)Z=0$ and $\Psi_G(\hat\pi)$ is Gaussian, then $\hat\pi$ is Gaussian.
\end{prop}
\begin{proof}
Set $q:=\frac12\rank(\Real Z)$. For the first statement, apply
Proposition~\ref{prop:dual-support-gaussian-source} with
$\Sigma=\Sigma_{\hat\pi_1}$ and with the common mean of
$\hat\pi_1$ and $\hat\pi_2$. \rev{There exists a common Gaussian channel}
$\Phi_G$, with phase-space transformation matrix $\mathbb X_\Phi$,
and zero-mean states $\hat\eta_1,\hat\eta_2$ such that $\hat\pi_j=\Phi_G(\hat\eta_j)$ and $\Ran(\mathbb X_\Phi)=\ker(\Real Z)$ for $j=1,2$.
Since $\dim\ker(\Real Z)=4m-2q=2(2m-q)$, $\mathbb X_\Phi$ has full column rank.

The composition $\Psi_G\circ\Phi_G$ is Gaussian with phase-space
transformation matrix $\rev{\mathbb X_\Psi}\mathbb X_\Phi$. If
$\rev{\mathbb X_\Psi}\mathbb X_\Phi\mathbf u=0$, then $\mathbb X_\Phi\mathbf u
\in
\ker(\Real Z)\cap\ker(\rev{\mathbb X_\Psi})$, and hence $\mathbb X_\Phi\mathbf u=0$ by the assumed intersection
condition. Since $\mathbb X_\Phi$ has full
column rank, $\mathbf u=0$. Thus
$\rev{\mathbb X_\Psi}\mathbb X_\Phi$ has full column rank. By
Lemma~\ref{lem:faithful-gaussian-deconvolution},
$\Psi_G\circ\Phi_G$ is injective. Therefore
$\Psi_G(\hat\pi_1)=\Psi_G(\hat\pi_2)$ implies
$\hat\eta_1=\hat\eta_2$, and hence $\hat\pi_1=\hat\pi_2$.

For the second statement, apply
Proposition~\ref{prop:dual-support-gaussian-source} with
$\Sigma=\Sigma_{\hat\pi}$ and $\mathbf m=\mathbf m_{\hat\pi}$. Then
$\hat\pi=\Phi_G(\hat\eta)$ for some zero-mean $\hat\eta$, with
$\Ran(\mathbb X_\Phi)=\ker(\Real Z)$. The same argument shows that
$\rev{\mathbb X_\Psi}\mathbb X_\Phi$ has full column rank. Since $\Psi_G(\hat\pi)
=
(\Psi_G\circ\Phi_G)(\hat\eta)$ is Gaussian, Lemma~\ref{lem:faithful-gaussian-deconvolution} implies
that $\hat\eta$ is Gaussian. Hence
$\hat\pi=\Phi_G(\hat\eta)$ is Gaussian.
\end{proof}
\begin{corollary}[Second-marginal injectivity and Gaussianity under covariance slackness]
\label{cor:saturated-face-marginal-reconstruction}
Let $Z\in\cM_{4m,\geq}(\C)$ satisfy
\begin{equation}
\ker(\Real Z)
\cap
\left(
\R^{2m}\oplus\{0\}^{2m}
\right)
=
\{0\}.
\label{eq:second-marginal-observable-support}
\end{equation}
If $\hat\pi_1,\hat\pi_2\in\cS_2^R(L_2(\R^{2m}))$ have the same mean,
satisfy
$(\Sigma_{\hat\pi_j}+\frac{i\hbar}{2}\Omega_{2m})Z=0$,
$j=1,2$,
and have the same second marginal, then
$\hat\pi_1=\hat\pi_2$.

Moreover, if $\hat\pi\in\cS_2^R(L_2(\R^{2m}))$ satisfies
$(\Sigma_{\hat\pi}+\frac{i\hbar}{2}\Omega_{2m})Z=0$
and its second marginal is Gaussian, then
$\hat\pi$ is Gaussian.
\end{corollary}
\begin{proof}
The partial trace $\Psi_G:=\Tr_1$ is a Gaussian channel with
phase-space transformation matrix
$\operatorname{pr}_2:=\smat{0_{2m}&\I_{2m}}$.
Since
$\ker(\operatorname{pr}_2)
=
\R^{2m}\oplus\{0\}^{2m}$,
condition \eqref{eq:second-marginal-observable-support} gives
$\ker(\Real Z)\cap\ker(\operatorname{pr}_2)=\{0\}$.
The result follows from
Proposition~\ref{prop:saturated-face-reconstruction}.
\end{proof}
\subsubsection{From faithful-input covariance uniqueness to global uniqueness} 
\label{subsubsec:faithful-rigidity-global-uniqueness} 
We now return to the barycenter problem. Covariance uniqueness first fixes the first and second moments of every barycenter and of every optimal coupling with a faithful input. This determines the Gaussian projection of every barycenter, but does not yet determine the barycenter itself. 
\begin{corollary}[Covariance uniqueness for arbitrary barycenters and optimal couplings with faithful inputs] 
\label{cor:faithful-barycenter-rigidity} 
Fix $\tau\in\{\pm1\}$ and assume \eqref{eq:faithful-inputs}. Let 
$\hat\rho_\tau$ be any $2$-QW barycenter of 
$\hat{\boldsymbol\sigma}^G$. Then 
\[ 
\mathbf m_{\hat\rho_\tau} 
= 
\overline{\mathbf m}, 
\qquad 
\Sigma_{\hat\rho_\tau} 
= 
\overline\Sigma_\tau, 
\qquad 
\mathcal G(\hat\rho_\tau) 
= 
\hat{\overline\rho}_\tau^G, 
\] 
where $\hat{\overline\rho}_\tau^G$ is the unique Gaussian barycenter 
from Corollary~\ref{cor:faithful-unique-gaussian-minimizer}. 
 
Moreover, fix $r\in[k]$, and let $\hat\pi_{\tau,r}$ be any optimal 
coupling between $\hat\rho_\tau$ and the faithful input 
$\hat\sigma_r^G$. Then the cross-covariance of 
$\hat\pi_{\tau,r}$ is $\overline A_{r,\tau}$ from 
Proposition~\ref{prop:faithful-covariance-uniqueness}, and hence 
\begin{equation} 
\Sigma_{\hat\pi_{\tau,r}} 
= 
\begin{pmatrix} 
\overline\Sigma_\tau&\overline A_{r,\tau}\\ 
\overline A_{r,\tau}^\top&\Sigma_r 
\end{pmatrix}. 
\label{eq:faithful-input-coupling-covariance} 
\end{equation} 
Furthermore, for every dual optimizer 
$\mathbf Z^\star=(Z_r^\star)_{r\in[k]}$ from 
Proposition~\ref{prop:covdual}, 
\begin{equation} 
\begin{pmatrix} 
\overline\Sigma_\tau+\frac{i\hbar}{2}\Omega_m& 
\overline A_{r,\tau}\\ 
\overline A_{r,\tau}^\top& 
\Sigma_r+\frac{i\tau\hbar}{2}\Omega_m 
\end{pmatrix} 
Z_r^\star 
= 
0. 
\label{eq:faithful-input-coupling-complementarity} 
\end{equation} 
\end{corollary} 
 
\begin{proof} 
By Corollary~\ref{cor:arbitrary-barycenter-covariance}, 
$\mathbf m_{\hat\rho_\tau}=\overline{\mathbf m}$ and 
$\Sigma_{\hat\rho_\tau}$ is the covariance component of an optimizer 
of \eqref{eq:qw2cov}. Hence 
Proposition~\ref{prop:faithful-covariance-uniqueness} gives 
$\Sigma_{\hat\rho_\tau}=\overline\Sigma_\tau$. Since Gaussian 
projection preserves first and second moments, 
$\mathcal G(\hat\rho_\tau)$ has mean $\overline{\mathbf m}$ and 
covariance $\overline\Sigma_\tau$, and therefore 
$\mathcal G(\hat\rho_\tau)=\hat{\overline\rho}_\tau^G$. 
 
Fix $r\in[k]$ and the prescribed optimal coupling 
$\hat\pi_{\tau,r}$, and denote its cross-covariance by $A_r$. 
Choose arbitrary optimal couplings with the remaining inputs, thereby 
completing $\hat\pi_{\tau,r}$ to a family of optimal couplings. 
Corollary~\ref{cor:arbitrary-barycenter-covariance} implies that the 
corresponding tuple 
$(\Sigma_{\hat\rho_\tau},A_1,\ldots,A_N)$ is an optimizer of 
\eqref{eq:qw2cov}. Hence 
Proposition~\ref{prop:faithful-covariance-uniqueness} gives 
$A_r=\overline A_{r,\tau}$, proving 
\eqref{eq:faithful-input-coupling-covariance}. 
 
Finally, 
Corollary~\ref{cor:faithful-complementary-slackness}, applied to the 
same covariance optimizer and any dual optimizer 
$\mathbf Z^\star$, gives 
\[ 
\begin{pmatrix} 
\Sigma_{\hat\rho_\tau}+\frac{i\hbar}{2}\Omega_m&A_r\\ 
A_r^\top& 
\Sigma_r+\frac{i\tau\hbar}{2}\Omega_m 
\end{pmatrix} 
Z_r^\star 
= 
0. 
\] 
Substituting 
$\Sigma_{\hat\rho_\tau}=\overline\Sigma_\tau$ and 
$A_r=\overline A_{r,\tau}$ yields 
\eqref{eq:faithful-input-coupling-complementarity}. 
\end{proof} 
 
Thus covariance uniqueness identifies the Gaussian projection of every 
barycenter $\mathcal G(\hat\rho_\tau) = \hat{\overline\rho}_\tau^G$. To obtain uniqueness of the barycenter as a quantum state, it remains to prove that $\hat\rho_\tau=\mathcal G(\hat\rho_\tau)$, or equivalently, that $\hat\rho_\tau$ is Gaussian. We obtain this conclusion by reconstructing an optimal 
coupling from its second marginal under the covariance slackness 
relation. 
 
To establish this equality, we compare an optimal coupling with its 
Gaussian projection. We first reduce the two choices of $\tau$ to the 
standard CCR and verify that the faithful-input covariance-dual optimizer 
satisfies the condition required for reconstruction from the second 
marginal. 
 
\begin{prop}[Standard-CCR reduction of a faithful-input covariance-dual optimizer] 
\label{prop:faithful-input-standard-CCR} 
Fix $\tau\in\{\pm1\}$ and $r\in[k]$. Let $\hat\rho_\tau$ be a $2$-QW 
barycenter of $\hat{\boldsymbol\sigma}^G$, let 
$\hat\pi_{\tau,r}$ be an optimal coupling between $\hat\rho_\tau$ and 
the faithful input $\hat\sigma_r^G$, and let 
$\mathbf Z^\star=(Z_r^\star)_{r\in[k]}$ be a dual optimizer from 
Proposition~\ref{prop:covdual}. Define 
\[ 
\widetilde\pi_{\tau,r} 
:= 
\bigl(\I_{L_2(\R^m)}\otimes J_{\tau,m}\bigr) 
\hat\pi_{\tau,r} 
\bigl(\I_{L_2(\R^m)}\otimes J_{\tau,m}\bigr)^\dagger, 
\qquad 
\widetilde Z_r^\star 
:= 
\bigl(\I_{2m}\oplus_DF_{\tau,m}\bigr) 
Z_r^\star 
\bigl(\I_{2m}\oplus_DF_{\tau,m}\bigr)^\top. 
\] 
Then $\widetilde\pi_{\tau,r}$ is a state on 
$L_2(\R^m)\otimes L_2(\R^m)$ whose first marginal is 
$\hat\rho_\tau$ and whose second marginal is Gaussian. Moreover, 
$\widetilde Z_r^\star\in\cM_{4m,\geq}(\C)$ and 
\begin{equation}
\begin{aligned}
\left(
\Sigma_{\widetilde\pi_{\tau,r}}
+\frac{i\hbar}{2}\Omega_{2m}
\right)
\widetilde Z_r^\star
&=0,\\
\bigl(\Real\widetilde Z_r^\star\bigr)_{12}
&=-\alpha_rF_{\tau,m},\\
\ker\bigl(\Real\widetilde Z_r^\star\bigr)
\cap
\left(
\R^{2m}\oplus\{0\}^{2m}
\right)
&=\{0\}.
\end{aligned}
\label{eq:faithful-standard-certificate} 
\end{equation} 
\end{prop} 
\begin{proof} 
By Proposition~\ref{prop:tau-standard-CCR-identification}, 
conjugation of the second subsystem by $J_{\tau,m}$ transforms the 
canonical variables according to 
$\I_{2m}\oplus_DF_{\tau,m}$. Hence 
\[ 
\Sigma_{\widetilde\pi_{\tau,r}} 
= 
\bigl(\I_{2m}\oplus_DF_{\tau,m}\bigr) 
\Sigma_{\hat\pi_{\tau,r}} 
\bigl(\I_{2m}\oplus_DF_{\tau,m}\bigr)^\top. 
\] 
Since the conjugation acts only on the second subsystem, the first 
marginal remains $\hat\rho_\tau$, while the second marginal is the 
standard-CCR image of $\hat\sigma_r^G$ and is therefore Gaussian. 
 
By Corollary~\ref{cor:faithful-barycenter-rigidity}, 
\[ 
\left[ 
\Sigma_{\hat\pi_{\tau,r}} 
+ 
\frac{i\hbar}{2} 
\bigl( 
\Omega_m\oplus_D(\tau\Omega_m) 
\bigr) 
\right] 
Z_r^\star 
= 
0. 
\] 
Since $Z_r^\star\in\cM_{4m,\geq}(\C)$, congruence by 
$\I_{2m}\oplus_DF_{\tau,m}$ gives 
$\widetilde Z_r^\star\in\cM_{4m,\geq}(\C)$. Using 
\begin{align*} 
\I_{4m}&=(\I_{2m}\oplus_DF_{\tau,m})^\top 
(\I_{2m}\oplus_DF_{\tau,m}) 
,\\ 
\Omega_{2m}&=\bigl(\I_{2m}\oplus_DF_{\tau,m}\bigr) 
\bigl( 
\Omega_m\oplus_D(\tau\Omega_m) 
\bigr) 
\bigl(\I_{2m}\oplus_DF_{\tau,m}\bigr)^\top 
, 
\end{align*} 
we obtain $\left( 
\Sigma_{\widetilde\pi_{\tau,r}} 
+ 
\frac{i\hbar}{2}\Omega_{2m} 
\right) 
\widetilde Z_r^\star = 0$. Moreover, the upper-right block of $\widetilde Z_r^\star$ is 
$(Z_r^\star)_{12}F_{\tau,m}^\top$. Since $F_{\tau,m}$ is real 
symmetric and 
$\Real(Z_r^\star)_{12}=-\alpha_r\I_{2m}$, $\bigl(\Real\widetilde Z_r^\star\bigr)_{12} 
= 
-\alpha_rF_{\tau,m}$. 
 
Finally, suppose 
$(\mathbf x,\mathbf0)\in 
\ker(\Real\widetilde Z_r^\star)$. Taking the second block row gives $-\alpha_rF_{\tau,m}^\top\mathbf x=0$. Since $\alpha_r>0$ and $F_{\tau,m}$ is invertible, $\mathbf x=0$. Hence  $\ker(\Real\widetilde Z_r^\star) 
\cap 
\left( 
\R^{2m}\oplus\{0\}^{2m} 
\right) 
= 
\{0\}$. 
\end{proof} 

The proof of 
Theorem~\ref{thm:faithful-global-uniqueness-intro} proceeds by showing that every optimal coupling with a faithful input is Gaussian. After reducing to the standard CCR in Proposition~\ref{prop:faithful-input-standard-CCR}, the coupling satisfies the covariance slackness relation and the condition required for reconstruction from the second marginal, while its prescribed second marginal is Gaussian. Corollary~\ref{cor:saturated-face-marginal-reconstruction} therefore implies that the coupling is Gaussian. Taking the first marginal shows that the barycenter is Gaussian, and covariance uniqueness then identifies it with the unique Gaussian barycenter.
\begin{proof}[Proof of 
Theorem~\ref{thm:faithful-global-uniqueness-intro}] 
As above, after relabeling the inputs if necessary, we assume that $\hat\sigma_1^G,\ldots,\hat\sigma_k^G$ are faithful and choose $s_0\in[k]$.

Fix $\tau\in\{\pm1\}$, and let 
$\hat\rho_\tau$ be any $2$-QW barycenter of 
$\hat{\boldsymbol\sigma}^G$. Let 
$\hat\pi_{\tau,s_0}$ be an optimal coupling between 
$\hat\rho_\tau$ and $\hat\sigma_{s_0}^G$, and fix a dual optimizer 
$\mathbf Z^\star=(Z_r^\star)_{r\in[k]}$ from 
Proposition~\ref{prop:covdual}. 
 
Apply Proposition~\ref{prop:faithful-input-standard-CCR} to 
$\hat\pi_{\tau,s_0}$. The resulting state 
$\widetilde\pi_{\tau,s_0}$ satisfies 
\[ 
\left( 
\Sigma_{\widetilde\pi_{\tau,s_0}} 
+ 
\frac{i\hbar}{2}\Omega_{2m} 
\right) 
\widetilde Z_{s_0}^\star 
= 
0 
\] 
and $\ker(\Real\widetilde Z_{s_0}^\star) \cap \left( 
\R^{2m}\oplus\{0\}^{2m} 
\right) 
= 
\{0\}$. 

By Proposition~\ref{prop:faithful-input-standard-CCR}, the second 
marginal of $\widetilde\pi_{\tau,s_0}$ is Gaussian. Therefore 
Corollary~\ref{cor:saturated-face-marginal-reconstruction} implies 
that $\widetilde\pi_{\tau,s_0}$ is Gaussian, and hence so is its first 
marginal $\hat\rho_\tau$. 
 
By Corollary~\ref{cor:faithful-barycenter-rigidity}, $\mathcal G(\hat\rho_\tau) 
= 
\hat{\overline\rho}_\tau^G$. Since $\hat\rho_\tau$ is Gaussian, $\hat\rho_\tau 
= 
\mathcal G(\hat\rho_\tau) 
= 
\hat{\overline\rho}_\tau^G$. Since $\hat\rho_\tau$ was arbitrary, the $2$-QW barycenter is unique 
and equals $\hat{\overline\rho}_\tau^G$. 
\end{proof}
 
\subsection{Semiclassical limit of Gaussian barycenters}
\label{subsec:gaussian-semiclassical}
We next compare Gaussian $2$-QW barycenters with classical quadratic
Wasserstein barycenters in the semiclassical limit. Fix
$\mathbf m_s\in\R^{2m}$ and
$\Sigma_s\in\cM_{2m,>}(\R)$, $s\in[N]$, independently of $\hbar$.

Let
$\mathcal N(\overline{\mathbf m},\overline\Sigma)$ denote the
classical $2$-Wasserstein barycenter of
$\{\mathcal N(\mathbf m_s,\Sigma_s)\}_{s\in[N]}$ with weights
$(\alpha_s)_{s\in[N]}$. The covariance part of the Gaussian
$2$-Wasserstein geometry is the Bures--Wasserstein geometry on positive
definite matrices; see, e.g.,~\cite{BhatiaJainLim2019}. By
\cite[Thm.~6.1]{AguCar} and
\cite[Eq.~(7)]{chen2018optimal}, $\overline{\mathbf m}
=\sum_{s=1}^N\alpha_s\mathbf m_s$, while
$\overline\Sigma$ is the unique optimal covariance of
\begin{equation}
\label{eq:classical-gaussian-barycenter-cov}
\begin{aligned}
\mathsf P_{0,\mathrm{cov}}
:=
&\min_{\substack{
\Sigma\in\cM_{2m,\mathrm{sym}}(\R),\
A_s\in\cM_{2m}(\R),\ s\in[N]}}
\quad
\Tr(\Sigma)
+
\sum_{s=1}^N\alpha_s
\bigl(
\Tr(\Sigma_s)-2\Tr(A_s)
\bigr)\\
\text{subject to}\quad&
\Sigma\succeq0,\qquad
\begin{pmatrix}
\Sigma&A_s\\
A_s^\top&\Sigma_s
\end{pmatrix}
\succeq0,
\quad s\in[N].
\end{aligned}
\end{equation}
Compared with the Bures--Wasserstein covariance problem~\eqref{eq:classical-gaussian-barycenter-cov}, the quantum covariance problem~\eqref{eq:qw2cov} additionally requires
\[
\Sigma+\frac{i\hbar}{2}\Omega_m\succeq0,
\qquad
\begin{pmatrix}\Sigma&A_s\\A_s^\top&\Sigma_s\end{pmatrix}
+\frac{i\hbar}{2}\bigl(\Omega_m\oplus_D(\tau\Omega_m)\bigr)\succeq0,
\quad s\in[N].
\]
The first inequality makes $\Sigma$ a quantum covariance matrix, and the second makes each block matrix the covariance of an admissible Gaussian quantum coupling. Thus these constraints impose quantum realizability, which is absent from the positive-semidefinite-matrix geometry of~\cite{BhatiaJainLim2019}. Their reduction to the block positive-semidefinite constraints in~\eqref{eq:classical-gaussian-barycenter-cov} as $\hbar\to0^+$ is the mechanism underlying the semiclassical convergence proved below.

For sufficiently small $\hbar>0$, let $\hat\sigma_s^\hbar$ be the Gaussian quantum state with mean $\mathbf m_s$ and covariance $\Sigma_s$. Since $\Sigma_s\succ0$, all these states are faithful for sufficiently small $\hbar$. Hence Theorem~\ref{thm:faithful-global-uniqueness-intro} gives, for each $\tau\in\{\pm1\}$, a unique $2$-QW barycenter $\hat{\overline\rho}_\tau^\hbar$. Denote its covariance by $\overline\Sigma_\tau(\hbar)$, and denote by
$\mathsf P_{\tau,\mathrm{cov}}^\hbar$ the minimum value of the
covariance problem \eqref{eq:qw2cov}.

For later use, we recall the Wigner function associated with the Weyl
convention of Section~\ref{subsec:standard-ccr-representation}; see
\cite[Sec.~II.A]{weedbrook2012gaussian}.
\begin{deff}[Wigner function of a Gaussian state]
\label{def:wigner-function}
For $\hat\rho\in\cS^G(L_2(\R^m))$, its Wigner function is
defined by
\[
\mathbb W_{\hat\rho}(\mathbf z)
:=
\frac{1}{(2\pi)^{2m}}
\int_{\R^{2m}}
e^{-i\boldsymbol\xi^\top\mathbf z}
\chi_{\hat\rho}(\boldsymbol\xi)\,
d\boldsymbol\xi.
\]
If $\hat\rho$ has mean $\mathbf m$ and covariance $\Sigma$, then
$\mathbb W_{\hat\rho}$ is the density of
$\mathcal N(\mathbf m,\Sigma)$.
\end{deff}

\begin{proof}[Proof of Theorem~\ref{thm:gaussian-semiclassical-intro}]
Fix $\tau\in\{\pm1\}$. We first compare the quantum and classical covariance problems. If
$(\Sigma,A_1,\ldots,A_N)$ is quantum feasible, then
\[
\begin{pmatrix}
\Sigma&A_s\\
A_s^\top&\Sigma_s
\end{pmatrix}
+
\frac{i\hbar}{2}
\bigl(\Omega_m\oplus_D(\tau\Omega_m)\bigr)
\succeq0,
\qquad s\in[N].
\]
Restricting the corresponding quadratic forms to real vectors gives
$\smat{\Sigma&A_s\\A_s^\top&\Sigma_s}\succeq0$ for every $s\in[N]$.
Thus every quantum-feasible tuple is classically feasible, and hence
$\mathsf P_{\tau,\mathrm{cov}}^\hbar
\geq\mathsf P_{0,\mathrm{cov}}$.

Conversely, let
$(\overline\Sigma,A_1^0,\ldots,A_N^0)$ be an optimizer of
\eqref{eq:classical-gaussian-barycenter-cov}. Since $\Sigma_s\succ0$,
the Schur complement gives
$\overline\Sigma-A_s^0\Sigma_s^{-1}(A_s^0)^\top\succeq0$ for every
$s\in[N]$.

For
$0<\hbar\leq\min_{s\in[N]}\lambda_{\min}(\Sigma_s)$, one has
$\Sigma_s+\frac{i\tau\hbar}{2}\Omega_m\succ0$ and $
\left\|
\left(
\Sigma_s+\frac{i\tau\hbar}{2}\Omega_m
\right)^{-1}
\right\|_{\mathrm{op}}
\leq
\frac{2}{\lambda_{\min}(\Sigma_s)}$. Using the resolvent identity $A^{-1}-B^{-1}=A^{-1}(B-A)B^{-1}$ and $\|\Omega_m\|_{\mathrm{op}}=1$, we therefore obtain
\[
\left\|
\left(
\Sigma_s+\frac{i\tau\hbar}{2}\Omega_m
\right)^{-1}
-
\Sigma_s^{-1}
\right\|_{\mathrm{op}}
\leq
\frac{\hbar}{\lambda_{\min}(\Sigma_s)^2},
\]
uniformly in $\tau\in\{\pm1\}$. Consequently,
\[
\begin{aligned}
&\left\|
\frac{i\hbar}{2}\Omega_m
-
A_s^0
\left[
\left(
\Sigma_s+\frac{i\tau\hbar}{2}\Omega_m
\right)^{-1}
-
\Sigma_s^{-1}
\right]
(A_s^0)^\top
\right\|_{\mathrm{op}}\leq
\hbar
\left(
\frac12+
\frac{\|A_s^0\|_{\mathrm{op}}^2}
{\lambda_{\min}(\Sigma_s)^2}
\right).
\end{aligned}
\]
Choose $c>
\max_{s\in[N]}
\left\{
\frac12+
\frac{\|A_s^0\|_{\mathrm{op}}^2}
{\lambda_{\min}(\Sigma_s)^2}
\right\}$ and set
$\Sigma_\hbar:=\overline\Sigma+c\hbar\I_{2m}$. Then
\[
\begin{aligned}
&\Sigma_\hbar+\frac{i\hbar}{2}\Omega_m
-
A_s^0
\left(
\Sigma_s+\frac{i\tau\hbar}{2}\Omega_m
\right)^{-1}
(A_s^0)^\top
\\
&=
\overline\Sigma
-
A_s^0\Sigma_s^{-1}(A_s^0)^\top
+
c\hbar\I_{2m}
+
\frac{i\hbar}{2}\Omega_m
-
A_s^0
\left[
\left(
\Sigma_s+\frac{i\tau\hbar}{2}\Omega_m
\right)^{-1}
-
\Sigma_s^{-1}
\right]
(A_s^0)^\top
\succeq0.
\end{aligned}
\]
Hence, by the Schur complement,
\[
\begin{pmatrix}
\Sigma_\hbar+\frac{i\hbar}{2}\Omega_m&A_s^0\\
(A_s^0)^\top&
\Sigma_s+\frac{i\tau\hbar}{2}\Omega_m
\end{pmatrix}
\succeq0,
\qquad s\in[N].
\]
Thus $(\Sigma_\hbar,A_1^0,\ldots,A_N^0)$ is quantum feasible for both
values of $\tau$. Evaluating the objective and combining with the
opposite inequality above gives
\[
0
\leq
\mathsf P_{\tau,\mathrm{cov}}^\hbar
-
\mathsf P_{0,\mathrm{cov}}
\leq
2mc\,\hbar,
\]
which proves the convergence of the optimal values.

We next prove convergence of the minimizing covariances. Let
$(\overline\Sigma_\tau(\hbar),
A_{1,\tau}^\hbar,\ldots,A_{N,\tau}^\hbar)$ be an optimal quantum
tuple. It is also classically feasible. The standard factorization of
positive block matrices \cite[Prop.~1.3.2]{Bhatia2007PDM} and the
Cauchy--Schwarz inequality give
\[
\Tr(A_{s,\tau}^\hbar)
\leq
\sqrt{\Tr\overline\Sigma_\tau(\hbar)}
\sqrt{\Tr\Sigma_s},
\qquad s\in[N].
\]
Consequently,
\[
\mathsf P_{\tau,\mathrm{cov}}^\hbar
\geq
\Tr\overline\Sigma_\tau(\hbar)
-
2\sqrt{\Tr\overline\Sigma_\tau(\hbar)}
\sum_{s=1}^N\alpha_s\sqrt{\Tr\Sigma_s}
+
\sum_{s=1}^N\alpha_s\Tr\Sigma_s.
\]
Since the optimal values are uniformly bounded for sufficiently small
$\hbar$, the family
$\{\overline\Sigma_\tau(\hbar)\}$ is bounded. Classical feasibility
and the same block-factorization argument imply that
$\{A_{s,\tau}^\hbar\}$ is bounded for every $s\in[N]$.

Let $\hbar_n\rightarrow0$. Passing to a subsequence if necessary, we may assume $\overline\Sigma_\tau(\hbar_n)\rightarrow\Sigma_*$ and $A_{s,\tau}^{\hbar_n}\rightarrow A_s^*$ for all $s\in[N]$. The limit tuple $(\Sigma_*,A_1^*,\ldots,A_N^*)$ is classically feasible. By continuity
of the objective and the convergence of the optimal values,
\[
\Tr(\Sigma_*)
+
\sum_{s=1}^N
\alpha_s
\bigl(\Tr(\Sigma_s)-2\Tr(A_s^*)\bigr)
=
\mathsf P_{0,\mathrm{cov}}.
\]
Thus $\Sigma_*$ minimizes
\eqref{eq:classical-gaussian-barycenter-cov}. By uniqueness,
$\Sigma_*=\overline\Sigma$. Since every convergent subsequence has the
same limit, $\overline\Sigma_\tau(\hbar)
\rightarrow
\overline\Sigma$.

Finally,
$\hat{\overline\rho}_\tau^\hbar$ is Gaussian with mean
$\overline{\mathbf m}$ and covariance
$\overline\Sigma_\tau(\hbar)$. Hence, by
Definition~\ref{def:wigner-function},
\[
\mathbb{W}_{\hat{\overline\rho}_\tau^\hbar}(\mathbf z)\,d\mathbf z
=
\mathcal N\bigl(
\overline{\mathbf m},
\overline\Sigma_\tau(\hbar)
\bigr)(d\mathbf z).
\]
Since the means are fixed and
$\overline\Sigma_\tau(\hbar)\rightarrow\overline\Sigma$, the
corresponding characteristic functions converge pointwise. Hence the
Wigner distributions converge weakly to
$\mathcal N(\overline{\mathbf m},\overline\Sigma)$.
\end{proof}
\subsection{A case study: thermal states}
\label{subsec:thermal}
We call a Gaussian state on $L_2(\R^m)$ \emph{isotropic} if its
covariance matrix is $w\I_{2m}$ for some $w\geq\hbar/2$.
Equivalently, the uncertainty principle reduces to
$w\geq\hbar/2$ in this isotropic setting. Following standard
terminology in continuous-variable quantum optics
\cite[Sec.~II.B]{weedbrook2012gaussian}, we also refer to a
zero-mean isotropic Gaussian state as a \emph{thermal state}. A
\emph{displaced thermal state} is a Gaussian state with mean
$\mathbf m\in\R^{2m}$ and covariance $w\I_{2m}$.

We first record a general boundary case that will settle uniqueness
when all thermal inputs are pure.

\begin{lemma}[Barycenter of pure inputs]
\label{lem:pure-input-coherent-barycenter}
Let
$\hat\sigma_1,\ldots,\hat\sigma_N\in\cS_2^R(L_2(\R^m))$
be pure states with means
$\mathbf m_1,\ldots,\mathbf m_N$, and set
$\overline{\mathbf m}:=\sum_{s=1}^N\alpha_s\mathbf m_s$.
Then, for every $\tau\in\{\pm1\}$, the corresponding $2$-QW
barycenter is unique and is the coherent state
$\proj{\overline{\mathbf m}}$.
\end{lemma}

\begin{proof}
By Proposition~\ref{prop:finitesecondemomentrho}, every barycenter has
finite canonical second moments. Let
$\hat\rho\in\cS_2^R(L_2(\R^m))$. Since the prescribed marginal
$\hat\sigma_s$ is pure, every coupling between $\hat\rho$ and
$\hat\sigma_s$ factorizes. Hence its cross-covariance vanishes, and
Proposition~\ref{prop:stateTrC} gives
\[
\begin{aligned}
\sum_{s=1}^N
\alpha_s\QW_{2,\tau}^2(\hat\rho,\hat\sigma_s)
={}&
\Tr(\Sigma_{\hat\rho})
+
\|\mathbf m_{\hat\rho}-\overline{\mathbf m}\|^2+
\sum_{s=1}^N\alpha_s
\left(
\Tr(\Sigma_{\hat\sigma_s})
+
\|\mathbf m_s-\overline{\mathbf m}\|^2
\right).
\end{aligned}
\]
The last sum is independent of $\hat\rho$, while
\[
\Tr(\Sigma_{\hat\rho})
+
\|\mathbf m_{\hat\rho}-\overline{\mathbf m}\|^2
=
\sum_{j=1}^{2m}
\Tr\left[
\hat\rho
\left(
\hat R_j-\overline m_j\I
\right)^2
\right].
\]
The operator on the right is a Weyl displacement of the harmonic
oscillator Hamiltonian and has
$\proj{\overline{\mathbf m}}$ as its unique ground state. Hence the
barycenter is uniquely equal to
$\proj{\overline{\mathbf m}}$.
\end{proof}

\begin{prop}[Uniqueness for displaced thermal inputs]
\label{prop:thermal-barycenter-uniqueness}
Let $\hat\sigma_s^{\mathrm{th}}$, $s\in[N]$, be displaced thermal
Gaussian states with means $\mathbf m_s$ and covariances
$w_s\I_{2m}$, where $w_s\geq\hbar/2$. Then, for every
$\tau\in\{\pm1\}$, the corresponding $2$-QW barycenter is unique
among all quantum states.
\end{prop}

\begin{proof}
If $w_r>\hbar/2$ for some $r\in[N]$, then
$\hat\sigma_r^{\mathrm{th}}$ is faithful, and uniqueness follows from
Theorem~\ref{thm:faithful-global-uniqueness-intro}. Otherwise
$w_s=\hbar/2$ for every $s\in[N]$, so all prescribed states are
coherent and therefore pure. In this case uniqueness follows from
Lemma~\ref{lem:pure-input-coherent-barycenter}.
\end{proof}

For pairs of displaced thermal Gaussian states, the two $2$-QW
problems admit the following explicit characterization.
The density formula is proved in Appendix~\ref{app:wassersteindistance},
while the channel formula follows from \cite[Thm.~3]{DPaTre19}, after
restoring $\hbar$ and using the cost normalization in
\eqref{eq:tau-cost}.

\Needspace{8\baselineskip}
\begin{prop}[2-QW distances between displaced thermal states]\leavevmode\par
\label{prop:qotcovid1}
Let
$\hat\rho^{\mathrm{th}},\hat\sigma^{\mathrm{th}}
\in\cS(L_2(\R^m))$
be displaced thermal states with means
$\mathbf m_{\hat\rho},\mathbf m_{\hat\sigma}$ and covariances
$w_{\hat\rho}\I_{2m},w_{\hat\sigma}\I_{2m}$, respectively, where
$w_{\hat\rho},w_{\hat\sigma}\geq\hbar/2$. Without loss of generality,
assume $w_{\hat\rho}\geq w_{\hat\sigma}$. Then
\begin{equation}
\QW_{2,+1}^2
(\hat\rho^{\mathrm{th}},\hat\sigma^{\mathrm{th}})
=
2m
\left[
\sqrt{w_{\hat\rho}-\frac{\hbar}{2}}
-
\sqrt{w_{\hat\sigma}-\frac{\hbar}{2}}
\right]^2
+
2m\hbar
+
\|\mathbf m_{\hat\rho}-\mathbf m_{\hat\sigma}\|^2,
\end{equation}
and
\begin{equation}
\QW_{2,-1}^2
(\hat\rho^{\mathrm{th}},\hat\sigma^{\mathrm{th}})
=
2m
\left[
\sqrt{w_{\hat\rho}+\frac{\hbar}{2}}
-
\sqrt{w_{\hat\sigma}-\frac{\hbar}{2}}
\right]^2
+
\|\mathbf m_{\hat\rho}-\mathbf m_{\hat\sigma}\|^2.
\end{equation}
\end{prop}

\begin{proof}[Proof of Corollary~\ref{cor:thermal-barycenter-intro}]
By Theorem~\ref{thm:gaussian-barycenter-intro}, there exists a
Gaussian barycenter with mean
$\overline{\mathbf m}:=\sum_{s=1}^N\alpha_s\mathbf m_s$, whose
covariance is the covariance component of an optimizer of
\eqref{eq:qw2cov}.

Fix $\tau\in\{\pm1\}$ and let
$(\Sigma,A_1,\ldots,A_N)$ be such an optimal covariance tuple. For
each $s\in[N]$, the covariance of the corresponding Gaussian optimal
coupling has the form
\[
\Sigma_{\hat\pi_s^G}
=
\begin{pmatrix}
\Sigma&A_s\\
A_s^\top&w_s\I_{2m}
\end{pmatrix}.
\]
By Corollary~\ref{coro:multimode-isotropic-reduction}, feasibility
implies feasibility of the isotropized block
\[
\begin{pmatrix}
\frac{\Tr(\Sigma)}{2m}\I_{2m}
&
\frac{\Tr(A_s)}{2m}\I_{2m}
\\
\frac{\Tr(A_s)}{2m}\I_{2m}
&
w_s\I_{2m}
\end{pmatrix}
+
\frac{i\hbar}{2}
\bigl(
\Omega_m\oplus_D(\tau\Omega_m)
\bigr)
\succeq0.
\]
Since isotropization preserves
$\Tr(\Sigma)$ and $\Tr(A_s)$,
Proposition~\ref{prop:stateTrC} shows that the corresponding displaced
thermal Gaussian state with mean $\overline{\mathbf m}$ and covariance
$\frac{\Tr(\Sigma)}{2m}\I_{2m}$ is also a barycenter. By
Proposition~\ref{prop:thermal-barycenter-uniqueness}, the barycenter is
unique. Hence it is necessarily a displaced thermal Gaussian state
with mean $\overline{\mathbf m}$ and covariance
$\overline w_\tau\I_{2m}$ for some
$\overline w_\tau\geq\hbar/2$.

For $\tau=+1$, Proposition~\ref{prop:qotcovid1} gives
\[
\begin{aligned}
\cB_{+1}(\hat{\boldsymbol\sigma}^{\mathrm{th}}\mid\boldsymbol{\alpha})
={}&
2m
\min_{w\geq\hbar/2}
\sum_{s=1}^N
\alpha_s
\left[
\sqrt{w-\frac{\hbar}{2}}
-
\sqrt{w_s-\frac{\hbar}{2}}
\right]^2+
2m\hbar
+
\sum_{s=1}^N
\alpha_s
\|\overline{\mathbf m}-\mathbf m_s\|^2.
\end{aligned}
\]
The unique minimizer is $
\overline w_{+1}
=
\left[
\sum_{s=1}^N
\alpha_s
\sqrt{w_s-\frac{\hbar}{2}}
\right]^2
+
\frac{\hbar}{2}$.

For $\tau=-1$, Proposition~\ref{prop:qotcovid1} gives
\begin{equation}
\label{eq:wctr-D2}
\begin{aligned}
\cB_{-1}(\hat{\boldsymbol\sigma}^{\mathrm{th}}\mid\boldsymbol{\alpha})
={}&
2m
\min_{w\geq\hbar/2}
\Bigg\{
\begin{aligned}
&\sum_{w_s<w}
\alpha_s
\left[
\sqrt{w+\frac{\hbar}{2}}
-
\sqrt{w_s-\frac{\hbar}{2}}
\right]^2\\
&\quad+
\sum_{w_s\geq w}
\alpha_s
\left[
\sqrt{w_s+\frac{\hbar}{2}}
-
\sqrt{w-\frac{\hbar}{2}}
\right]^2
\end{aligned}
\Bigg\}
\\
&+
\sum_{s=1}^N
\alpha_s
\|\overline{\mathbf m}-\mathbf m_s\|^2.
\end{aligned}
\end{equation}
By uniqueness of the barycenter,
$\overline w_{-1}$ is the unique minimizer of
\eqref{eq:wctr-D2}.
\end{proof}

\begin{coro}
\label{coro:densitycentervalue}
Let $m,N\in\mathbb N$,
$\boldsymbol\alpha=(\alpha_s)_{s\in[N]}$ be positive weights
satisfying $\sum_{s=1}^N\alpha_s=1$, and let
$\hat{\boldsymbol\sigma}^{\mathrm{th}}:=(\hat\sigma_s^{\mathrm{th}})_{s\in[N]}$ be
displaced thermal Gaussian states with means $\mathbf m_s$ and covariances
$w_s\I_{2m}$. 
Then
\begin{equation}
\begin{aligned}
\cB_{+1}
(\hat{\boldsymbol\sigma}^{\mathrm{th}}\mid\boldsymbol{\alpha})
={}&
\sum_{s=1}^N
\left\{
2m\alpha_s
\left[
\sqrt{\overline w_{+1}-\frac{\hbar}{2}}
-
\sqrt{w_s-\frac{\hbar}{2}}
\right]^2
+
\alpha_s
\|\overline{\mathbf m}-\mathbf m_s\|^2
\right\}
+
2m\hbar
\\
={}&
2m
\left(
\sum_{s=1}^N
\alpha_s
\left(
w_s-\frac{\hbar}{2}
\right)
-
\left[
\sum_{s=1}^N
\alpha_s
\sqrt{w_s-\frac{\hbar}{2}}
\right]^2
\right)
\\
&+
\left(
\sum_{s=1}^N
\alpha_s\|\mathbf m_s\|^2
-
\left\|
\sum_{s=1}^N
\alpha_s\mathbf m_s
\right\|^2
\right)
+
2m\hbar.
\end{aligned}
\end{equation}
\end{coro}
 \subsection{Tensorization and nonuniqueness on the partially pure boundary}
\label{subsec:nonfaithful-boundary}

We first record a tensorization property of $2$-QW barycenters.

\begin{lemma}[Tensorization of barycenters]
\label{lem:tensorization-barycenter}
Fix $\tau\in\{\pm1\}$. Let $r,N\in\N$ and, for each $j\in[r]$,
let $m_j\in\N$. For $j\in[r]$ and $s\in[N]$, let $\hat\sigma_s^{(j)}\in\cS\bigl(L_2(\R^{m_j})\bigr)$ be states with common weights $\alpha_s>0$ satisfying $\sum_{s=1}^N\alpha_s=1$. Suppose that
$\hat\rho_j\in\cS(L_2(\R^{m_j}))$ is a $2$-QW barycenter of
$\{\hat\sigma_s^{(j)}\}_{s\in[N]}$ with these weights. Then $\bigotimes_{j=1}^r\hat\rho_j$ is a $2$-QW barycenter of $\left\{
\bigotimes_{j=1}^r\hat\sigma_s^{(j)}
\right\}_{s\in[N]}$ on $\bigotimes_{j=1}^rL_2(\R^{m_j})
\simeq L_2(\R^{m_1+\cdots+m_r})$.
\end{lemma}

\begin{proof}
Let $\hat\rho$ be an arbitrary state on the product system and let
$\hat\rho^{(j)}$ denote its $j$th marginal. Since the quadratic cost
is the sum of the costs on the individual factors, taking the
corresponding factor marginals of any admissible coupling gives
\[
\QW_{2,\tau}^2
\left(
\hat\rho,
\bigotimes_{j=1}^r\hat\sigma_s^{(j)}
\right)
\geq
\sum_{j=1}^r
\QW_{2,\tau}^2
\left(
\hat\rho^{(j)},\hat\sigma_s^{(j)}
\right).
\]
Therefore
\[
\begin{aligned}
\sum_{s=1}^N\alpha_s
\QW_{2,\tau}^2
\left(
\hat\rho,
\bigotimes_{j=1}^r\hat\sigma_s^{(j)}
\right)
&\geq
\sum_{j=1}^r
\sum_{s=1}^N\alpha_s
\QW_{2,\tau}^2
\left(
\hat\rho^{(j)},\hat\sigma_s^{(j)}
\right)
\\
&\geq
\sum_{j=1}^r
\sum_{s=1}^N\alpha_s
\QW_{2,\tau}^2
\left(
\hat\rho_j,\hat\sigma_s^{(j)}
\right),
\end{aligned}
\]
where the second inequality follows from the barycenter property of
each $\hat\rho_j$.

If one of the factor barycenter values on the right-hand side is
$+\infty$, this lower bound makes the conclusion immediate. We may
therefore assume that all factor barycenter values are finite.

Conversely, fix $\varepsilon>0$. For every $s\in[N]$ and $j\in[r]$,
choose an $\varepsilon$-optimal coupling
$\hat\pi_{s,j}^{\varepsilon}$ between
$\hat\rho_j$ and $\hat\sigma_s^{(j)}$. Then
$\bigotimes_{j=1}^r\hat\pi_{s,j}^{\varepsilon}$ is an admissible
coupling between
$\bigotimes_{j=1}^r\hat\rho_j$ and
$\bigotimes_{j=1}^r\hat\sigma_s^{(j)}$, and additivity of the
quadratic cost gives
\[
\QW_{2,\tau}^2
\left(
\bigotimes_{j=1}^r\hat\rho_j,
\bigotimes_{j=1}^r\hat\sigma_s^{(j)}
\right)
\leq
\sum_{j=1}^r
\QW_{2,\tau}^2
\left(
\hat\rho_j,\hat\sigma_s^{(j)}
\right)
+
r\varepsilon.
\]
After multiplying by $\alpha_s$, summing over $s\in[N]$, and letting
$\varepsilon\to0$, we obtain the reverse inequality. Hence
$\bigotimes_{j=1}^r\hat\rho_j$ is a barycenter.
\end{proof}

Lemma~\ref{lem:pure-input-coherent-barycenter} shows that
nonfaithfulness alone does not imply nonuniqueness: when all inputs are pure, the barycenter is still unique. The faithful-input assumption in Theorem~\ref{thm:faithful-global-uniqueness-intro} cannot, however, be removed from a general uniqueness statement. The following example shows that, on the partially pure boundary, there may exist a continuum of Gaussian barycenters as well as non-Gaussian barycenters. In particular, distinct barycenters may have identical first and second canonical moments.

\begin{example}[Nonuniqueness and non-Gaussianity for partially pure inputs]
\label{ex:nonfaithful-boundary}
For $w\geq\hbar/2$, let $\hat\sigma_w^{\mathrm{th}}\in\cS\bigl(L_2(\R)\bigr)$
denote the centered one-mode thermal state with covariance $w\I_2$.
By Corollary~\ref{cor:thermal-barycenter-intro}, the equal-weight
one-mode $2$-QW barycenter of
$\hat\sigma_{3\hbar/2}^{\mathrm{th}}$ and
$\hat\sigma_{\hbar/2}^{\mathrm{th}}$ is
$\hat\sigma_{3\hbar/4}^{\mathrm{th}}$ for $\tau=+1$ and
$\hat\sigma_{\hbar}^{\mathrm{th}}$ for $\tau=-1$.

Consider the $2$-QW barycenter problem with inputs
\[
\hat\sigma_1
=
\hat\sigma_{3\hbar/2}^{\mathrm{th}}
\otimes
\hat\sigma_{\hbar/2}^{\mathrm{th}},
\qquad
\hat\sigma_2
=
\hat\sigma_{\hbar/2}^{\mathrm{th}}
\otimes
\hat\sigma_{3\hbar/2}^{\mathrm{th}},
\qquad
\alpha_1=\alpha_2=\frac12.
\]
These are centered Gaussian states with covariances
$\Sigma_1=\frac{\hbar}{2}\smat{3\I_2&0\\0&\I_2}$ and
$\Sigma_2=\frac{\hbar}{2}\smat{\I_2&0\\0&3\I_2}$.
Thus both inputs are partially pure and nonfaithful.

\emph{Gaussian nonuniqueness.}
For $t\in[-1,1]$, let $\hat\rho_t^\tau$ be the centered Gaussian state
with covariance
\[
\overline\Sigma_t^{+}
=
\frac{\hbar}{4}
\begin{pmatrix}
3\I_2&t\I_2\\
t\I_2&3\I_2
\end{pmatrix},
\qquad
\overline\Sigma_t^{-}
=
\frac{\hbar}{2}
\begin{pmatrix}
2\I_2&t\I_2\\
t\I_2&2\I_2
\end{pmatrix}.
\]
For $|t|\leq1$, these matrices belong to $\mathrm{Cov}_2^{\mathrm Q}$,
and the one-mode marginals of $\hat\rho_t^\tau$ are
$\hat\sigma_{3\hbar/4}^{\mathrm{th}}$ for $\tau=+1$ and
$\hat\sigma_{\hbar}^{\mathrm{th}}$ for $\tau=-1$.

By Lemma~\ref{lem:nonfaithful-boundary-couplings}, for each
$\tau\in\{\pm1\}$, $t\in[-1,1]$, and $s\in\{1,2\}$, there exists an
optimal coupling between $\hat\rho_t^\tau$ and $\hat\sigma_s$
attaining the corresponding tensorized one-mode cost. Hence
$\hat\rho_t^\tau$ is an equal-weight $2$-QW barycenter of
$\hat\sigma_1$ and $\hat\sigma_2$. Since
$\overline\Sigma_t^\tau$ varies with $t$, this gives a continuum of
distinct Gaussian barycenters. At $t=0$,
$\hat\rho_0^\tau$ is
$(\hat\sigma_{3\hbar/4}^{\mathrm{th}})^{\otimes2}$ for $\tau=+1$
and $(\hat\sigma_{\hbar}^{\mathrm{th}})^{\otimes2}$ for $\tau=-1$,
namely the product barycenter from
Lemma~\ref{lem:tensorization-barycenter}.

\emph{Non-Gaussian barycenters.}
The $2$-QW barycenter objective is convex in the barycentric state:
if $\hat\rho_1$ and $\hat\rho_2$ are $2$-QW barycenters, convex
combinations of their optimal couplings show that
$\lambda\hat\rho_1+(1-\lambda)\hat\rho_2$ is again a $2$-QW
barycenter for every $\lambda\in[0,1]$. Thus, for $t\in(0,1]$,
\[
\hat\omega_t^\tau
:=
\frac12\hat\rho_t^\tau+\frac12\hat\rho_{-t}^\tau
\]
is a barycenter. Since $\overline\Sigma_t^\tau$ is affine in $t$,
$\mathbf m_{\hat\omega_t^\tau}=0$ and
$\Sigma_{\hat\omega_t^\tau}=\overline\Sigma_0^\tau$.
However, $\hat\omega_t^\tau$ is non-Gaussian, being a nontrivial
mixture of two centered Gaussian states with distinct covariance
matrices. Thus $\hat\omega_t^\tau$ and the Gaussian barycenter
$\hat\rho_0^\tau$ are distinct barycenters with identical first and
second canonical moments.
\end{example}

Hence this example exhibits two distinct sources of nonuniqueness: the covariance problem $\mathsf P_{\tau,\mathrm{cov}}$ from \eqref{eq:qw2cov} itself has multiple optimizers, and even after fixing the optimal covariance \(\overline\Sigma_0^\tau\), the first and second moments do not determine the barycenter as a quantum state.

\begin{revision}
\section*{Acknowledgments}

A.G. gratefully acknowledges the National Council for Scientific and
Technological Development (CNPq) for a research productivity fellowship.
Z.L. was supported by the Natural Sciences and Engineering Research Council
of Canada (NSERC) through the Alliance Grants program, under the QIMMIQ
project (ALLRP 592521-23). Part of this research was performed while the
authors were visiting the Institute for Pure and Applied Mathematics (IPAM),
which is supported by the National Science Foundation
(Grant Nos. DMS-1925919 and DMS-2422832).
\end{revision}

\appendix
\section[Compactness and lower semicontinuity of quantum states]{Compactness and lower semicontinuity of\\ quantum states}\label{append:compact}
In this appendix we collect compactness, tightness, and lower
semicontinuity results for quantum states on separable Hilbert spaces. These results are used in the direct method for the barycenter problem.

We begin with a spectral characterization of compact resolvent operators.
\begin{prop}[Purely discrete spectrum from compact resolvent; see {\cite[Prop.~5.12]{Schmudgen2012}}]\label{prop:purely}
Let $\cH$ be a separable Hilbert space, and let $A\in\cL_S(\cH)$ be a self-adjoint operator. Then $A$ has compact resolvent if and only if it has a purely discrete spectrum; that is, $\mathrm{spec}(A)$ consists only of eigenvalues of finite multiplicity with no accumulation point in \(\mathbb R\).
\end{prop}
\begin{cor}[Nonnegative operators with compact resolvent]
\label{cor:cost-compact-resolvent}
Let \(\cH\) be a separable Hilbert space, 
\(C\in\cL_{\geq}(\cH)\) be a nonnegative operator with compact resolvent and $\mathrm{spec}(C)$ be the spectrum of $C$ counting multiplicity. Then 
\begin{equation}
\label{eq:discrete-confining}
    \#\{\lambda\in\mathrm{spec}(C):\ \lambda<r\}<\infty,
    \qquad \text{for all } r>0 .
\end{equation}
Equivalently, if \(\cH\) is infinite-dimensional, the eigenvalues of \(C\),
counted with multiplicity, can be ordered as $0\le \lambda_1\le \lambda_2\le\cdots$ with $\lambda_j\to+\infty$.
\end{cor}
\begin{proof}
By Proposition~\ref{prop:purely}, \(C\) has purely discrete spectrum. Therefore, the conclusion follows from $C\geq0$.
\end{proof}

Combining \cite[Lem.~11.5]{Holevo2019QSCI} with
Corollary~\ref{cor:cost-compact-resolvent}, we obtain:
\begin{lemma}[See {\cite[Lem.~11.5]{Holevo2019QSCI}}]\label{lemma:boundedHcompact} Let $\cH$ be a separable Hilbert space and $(\rho_j)_{j\in\N}\subset\cS(\cH)$ be a sequence of states. Consider $H\in\cL_\geq(\cH)$ with compact resolvent. If there exists $A\in\R$ such that $0\le \Tr(H\rho_j)\le A$ for all $j\in\N$, then there exists a subsequence $(\rho_{j_k})_{k\in\N}$ converging in trace norm to some $\rho\in\cS(\cH)$, and \[ 0\le \Tr(H\rho)\le A. \] 

Equivalently, the sublevel set $\{\rho\in\cS(\cH):\Tr(H\rho)\le A\}$ is trace-norm compact.
\end{lemma}

In the general case, if we only assume that $H\in\cL_{\geq}(\cH)$
is a nonnegative self-adjoint operator, then compactness need not hold.
However, one still has trace-norm closedness of the energy-constrained set.
This follows from the lower semicontinuity of the energy functional
$\rho\mapsto \Tr(\rho H)$.

\begin{prop}[Lower semicontinuity; see {\cite[Lem.~11.4]{Holevo2019QSCI}} and \cite{WeisShirokov2021}]
\label{prop:Hclosed}
Let \(\cH\) be a separable Hilbert space and \(H\in\cL_{\geq}(\cH)\) be a nonnegative self-adjoint operator. Then the extended energy functional
\[
    \mathcal E_H:\cS(\cH)\to[0,+\infty],
    \qquad
    \mathcal E_H(\rho):=\Tr(H\rho),
\]
is affine and lower semicontinuous with respect to the trace-norm topology. Equivalently, if \(\rho_j\to\rho\) in trace norm, then
\[
    \Tr(H\rho)
    \le
    \liminf_{j\to\infty}\Tr(H\rho_j).
\]
Consequently, for every \(A\in[0,+\infty)\), the energy sublevel set $\left\{
    \rho\in\cS(\cH):\Tr(H\rho)\le A
    \right\}$ is closed in the trace-norm topology.
\end{prop}

\noindent\textbf{Compactness for couplings under converging marginals.}

\begin{theo}[Existence of a convergent subsequence in \(\Gamma(\rho, \sigma)\); see {\cite[Prop.~8]{DPaTre19}}]\label{theo:compact} Let \(\cH_1\) and \(\cH_2\) be separable Hilbert spaces, and let \(\rho \in \cS(\cH_1)\) and \(\sigma \in \cS(\cH_2)\) be density operators. Consider a sequence \((\pi_n)_{n\geq1} \subset \cS(\cH_1 \otimes \cH_2)\) of quantum couplings such that \(\pi_n \in \Gamma(\rho_n, \sigma_n)\), where \(\rho_n \to \rho\) and \(\sigma_n \to \sigma\) in trace norm. Then there exists a subsequence $(\pi_{n_k})_{k\ge1}$ and a coupling $\pi\in\Gamma(\rho,\sigma)$ such that \[ \pi_{n_k}\ \to\ \pi\quad\text{in trace norm}. \] 
\end{theo}

The following proposition shows that imposing a marginal energy bound defines
a trace-norm closed set.

\begin{prop}[Closedness of marginal energy sublevel sets]
\label{prop:marginal-energy-closed}
Let $\cH_1$ and $\cH_2$ be separable Hilbert spaces, and let
$H\in\cL_\geq(\cH_1)$ be a nonnegative self-adjoint operator on $\cH_1$. For $c\in[0,+\infty)$,
define
\[
    \mathcal C_c:=\left\{
    \pi\in\cS(\cH_1\otimes\cH_2):
    \Tr\!\left[
        H\,(\Tr_{\cH_2}\pi)
    \right]
    \leq c
    \right\}.
\]
Then $\mathcal C_c$ is closed with respect to the trace norm.
\end{prop}
\begin{proof}
Recall from Proposition~\ref{prop:Hclosed} that, for every nonnegative
self-adjoint operator $H$ on $\cH_1$, the energy sublevel set $\mathfrak S_{H,c}(\cH_1)
    :=
    \left\{
    \rho\in\cS(\cH_1):
    \Tr(H\rho)\leq c
    \right\}$ is closed in the trace-norm topology.

The partial trace $\Tr_{\cH_2}:
    \cT^{(1)}_S(\cH_1\otimes\cH_2)
    \to
    \cT^{(1)}_S(\cH_1)$ is trace-norm continuous. Moreover,
\[
    \mathcal C_c
    =
    \left(\Tr_{\cH_2}\right)^{-1}
    \left(\mathfrak S_{H,c}(\cH_1)\right)
    \cap
    \cS(\cH_1\otimes\cH_2).
\]
Since $\mathfrak S_{H,c}(\cH_1)$ is trace-norm closed and
$\Tr_{\cH_2}$ is trace-norm continuous, the inverse image
is trace-norm closed. The state space
$\cS(\cH_1\otimes\cH_2)$ is also trace-norm closed in
$\cT^{(1)}_S(\cH_1\otimes\cH_2)$. Therefore $\mathcal C_c$ is
trace-norm closed.
\end{proof}

\section{Essential self-adjointness of quadratic Hamiltonians}
\label{app:quadratic-hamiltonians}

This appendix proves Proposition~\ref{prop:quadratic-hamiltonian-esa}. The only external
input is the essential self-adjointness criterion of Fulsche--van Luijk~\cite[Thm.~1]{FulscheVanLuijk2025} for Weyl
pseudodifferential operators, used here only in the elementary case of real quadratic
polynomials.

\begin{lemma}[Quadratic Schr\"odinger operators with multiplicity]
\label{lem:quadratic-schrodinger-esa}
Let $\hat{\mathbf{R}}$ be the standard canonical vector on
\(L_2(\mathbb R^m)\) introduced in
Section~\ref{subsec:standard-ccr-representation}. Let
\(M\in\cM_{2m,\mathrm{sym}}(\mathbb R)\), \(\mathbf{b}\in\mathbb R^{2m}\), and
\(c\in\mathbb R\). Define, on \(\mathscr S(\mathbb R^m)\),
\[
\hat h^{(0)}_{M,\mathbf{b},c}
:=
\frac12\,\hat{\mathbf{R}}^\top M\hat{\mathbf{R}}
+
\mathbf{b}^\top\hat{\mathbf{R}}
+
c.
\]
Then \(\hat h^{(0)}_{M,\mathbf{b},c}\) is essentially self-adjoint on
\(\mathscr S(\mathbb R^m)\). Moreover, if \(\cK\) is a Hilbert space, then $\hat H^{(0)}_{M,\mathbf{b},c}
:=
\hat h^{(0)}_{M,\mathbf{b},c}\otimes \I_\cK$
is essentially self-adjoint on
\(\mathscr S(\mathbb R^m)\otimes_{\mathrm{alg}}\cK
\subset L_2(\mathbb R^m)\otimes \cK\), where $\otimes_{\mathrm{alg}}$ is the algebraic tensor product.
\end{lemma}

\begin{proof}
Since \(M\) is real symmetric and \(b,c\) are real, the operator
\(\hat h^{(0)}_{M,\mathbf{b},c}\) is the Weyl quantization of a real quadratic polynomial on
phase space. Hence it is essentially self-adjoint on \(\mathscr S(\mathbb R^m)\) by the
criterion of Fulsche--van Luijk~\cite[Thm.~1]{FulscheVanLuijk2025}. Indeed, a quadratic
polynomial has bounded derivatives of all orders at least two, so the hypotheses of that
criterion are satisfied.

Since \(\hat h^{(0)}_{M,\mathbf{b},c}\) is essentially self-adjoint
on \(\mathscr S(\mathbb R^m)\), its closure \(\hat h_{M,\mathbf{b},c}\) is self-adjoint. Therefore
\(\hat h_{M,\mathbf{b},c}\otimes \I_\cK\) is self-adjoint on \(L_2(\mathbb R^m)\otimes \cK\), and
\(\mathscr S(\mathbb R^m)\otimes_{\mathrm{alg}}\cK\) is a core for this tensor-product
operator. Thus
\(\hat H^{(0)}_{M,\mathbf{b},c}\) is essentially self-adjoint on
\(\mathscr S(\mathbb R^m)\otimes_{\mathrm{alg}}\cK\).
\end{proof}

\begin{proof}[Proof of Proposition~\ref{prop:quadratic-hamiltonian-esa}]
By the Stone--von Neumann theorem with
multiplicity, there exist a Hilbert space \(\cK\), a matrix
\(L\in\cM_{2m}(\mathbb R)\) satisfying $L^\top(\hbar\Omega_m)L=\Theta$, and a unitary $U:L_2(\mathbb R^m)\otimes \cK\to \cH$
such that $U^\dagger\mathbf{R}U
=
(L^\top\hat{\mathbf{R}})\otimes \I_\cK$. Hence, on the domain $\mathscr S(\mathbb R^m)\otimes_{\mathrm{alg}}\cK$, the transformed operator is
\[
U^\dagger H^{(0)}_{G,\mathbf{a},c} U
=
\left(
\frac12\,\hat{\mathbf{R}}^\top LGL^\top\hat{\mathbf{R}}
+
(L\mathbf a)^\top\hat{\mathbf{R}}
+
c\I_{L_2(\mathbb R^m)}
\right)\otimes\I_\cK.
\]
This is the operator of Lemma~\ref{lem:quadratic-schrodinger-esa} with $M=LGL^\top$ and $\mathbf{b}=L\mathbf{a}$.
Since \(G\) is real symmetric, \(M=LGL^\top\) is also real symmetric. Therefore
Lemma~\ref{lem:quadratic-schrodinger-esa} implies that the restriction of
\( U^\dagger H^{(0)}_{G,\mathbf{a},c}U\) to
\(\mathscr S(\mathbb R^m)\otimes_{\mathrm{alg}}\cK\) is essentially self-adjoint.

Finally, under the standard CCR representation with multiplicity,
\[
\mathscr S(\mathbb R^m)\otimes_{\mathrm{alg}}\cK
\subset U^\dagger\mathcal H_R^\infty.
\]
The restriction of \(U^\dagger H^{(0)}_{G,\mathbf{a},c}U\) to
\(U^\dagger\mathcal H_R^\infty\) is therefore a symmetric extension of its
essentially self-adjoint restriction to
\(\mathscr S(\mathbb R^m)\otimes_{\mathrm{alg}}\cK\). Since a symmetric
extension of an essentially self-adjoint operator has the same self-adjoint
closure, this extension is also essentially self-adjoint. By unitary
equivalence, \(H^{(0)}_{G,\mathbf{a},c}\) is essentially self-adjoint on
\(\mathcal H_R^\infty\).
\end{proof}
 \section{\texorpdfstring{Fenchel--Rockafellar}{Fenchel-Rockafellar} theorem}\label{app:FRtheorem}
We recall basic notions and state a standard form of the Fenchel--Rockafellar duality theorem.

\begin{deff}[Effective domain]
Let $E$ be a normed vector space with norm $\|\cdot\|$ and let $\psi:E\to(-\infty,+\infty]$ be an extended real-valued function. Its effective domain is $\Dom \psi:=\{x\in E:\psi(x)<+\infty\}$. We say $\psi$ is \emph{proper} if $\Dom \psi\neq\emptyset$ and $\psi\not\equiv+\infty$.
\end{deff}

\begin{deff}[Legendre--Fenchel conjugate]
Let $E$ be a normed vector space and $E^*$ denote the dual of $E$, with duality pairing $\langle f,x\rangle$ for $(f,x)\in E^*\times E$. The Legendre--Fenchel conjugate of a proper function $F:E\to(-\infty,+\infty]$ is the function $F^*:E^*\to(-\infty,+\infty]$ defined by
\[
F^*(f)\ :=\ \sup_{x\in E}\,\{\langle f,x\rangle - F(x)\}.
\]
\end{deff}

\begin{theo}[Fenchel--Rockafellar duality; see {\cite[Thm.~1.12]{brezis2011functional}}]\label{app:FRTheorem}
Let $E$ be a normed vector space and let $\psi,\phi:E\to(-\infty,+\infty]$ be proper convex functions. Assume there exists $x_0\in \Dom \psi\cap \Dom \phi$ at which $\psi$ is continuous. Then
\[
\inf_{x\in E}\{\psi(x)+\phi(x)\}
\;=\;
\sup_{f\in E^*}\{-\psi^*(-f)-\phi^*(f)\}
\;=\;
-\min_{f\in E^*}\{\psi^*(-f)+\phi^*(f)\}.
\]
Moreover, the supremum on the dual side is attained, i.e., there exists $f^\star\in E^*$ achieving the maximum.
\end{theo}
\subsection{\texorpdfstring{Operators in the dual of the bounded self-adjoint space $\cB_S(\cH)^*$}{Operators in the dual of the bounded self-adjoint operator space}}
\label{append:operator}

Let $\cH$ be a separable Hilbert space. The next lemmas adapt \cite{caglioti2021towards} from the complex spaces $\cB(\cH)$ and $\cB(\cH)^*$ to the real self-adjoint subspaces $\cB_S(\cH)$ and $\cB_S(\cH)^*$. The arguments rely only on the dualities $\cK_S(\cH)^*=\cT^{(1)}_S(\cH)$ and $\big(\cT^{(1)}_S(\cH)\big)^*=\cB_S(\cH)$, with pairing $\langle T,B\rangle=\Tr(TB)$.

\begin{lemma}[See {\cite[Lem.~3.1]{caglioti2021towards}}]\label{lemma:idmnorm}
Let $\Lambda\in\cB_S(\cH)^*$ be positive semidefinite, i.e.\ $\Lambda(B)\ge0$ for all $B\in\cB_S(\cH)$ with $B\ge0$. Then $\|\Lambda\|_*=\Lambda(\I_{\cH})$.
\end{lemma}

\begin{lemma}[See {\cite[Lem.~3.2]{caglioti2021towards}}]\label{lemma:b*decomposition}
Let $\Lambda\in\cB_S(\cH)^*$. There exist a unique $T\in\cT^{(1)}_S(\cH)$ and a unique $L\in\cB_S(\cH)^*$ such that
\[
\Lambda(B)=\Tr(TB)+L(B)\quad(\forall B\in\cB_S(\cH)),\qquad L|_{\cK_S(\cH)}=0.
\]
If moreover $\Lambda\ge0$, then $T\ge0$, $L\ge0$, $\|T\|_1\le \|\Lambda\|_*$, and $\|L\|_*\le \|\Lambda\|_*$.
In finite dimension one has $L\equiv0$.
\end{lemma}

\begin{lemma}[See {\cite[Lem.~3.3]{caglioti2021towards}}]\label{lemma:equalB*traceclass}
Let $\cH_1,\cH_2$ be separable Hilbert spaces and $\Lambda\in\cB_S(\cH_1\otimes\cH_2)^*$ be positive semidefinite. Assume there exist $\rho\in\cS(\cH_1)$ and $\sigma\in\cS(\cH_2)$ such that
\[
\Lambda(B_1\oplus B_2)=\Tr(B_1\rho)+\Tr(B_2\sigma)\quad\text{for all }B_1\in\cB_S(\cH_1),\ B_2\in\cB_S(\cH_2).
\]
Then there exists a coupling $\pi\in\Gamma(\rho,\sigma)$ such that
\[
\Lambda(B)=\Tr(B\,\pi)\quad\text{for all }B\in\cB_S(\cH_1\otimes\cH_2).
\]
\end{lemma}
For a positive functional
$\Lambda\in\cB_S(\cH)^*$ and
$A\in\cL_{\geq}(\cH)$, we define the extended evaluation of $\Lambda$ at $A$ by
\begin{equation}
\label{eq:ext-eval}
\Lambda(A)
:=
\sup
\left\{
\Lambda(B):
B\in\cB_S(\cH),\ B\leq A
\right\}
\in[0,+\infty].
\end{equation}
This extension is monotone and positively homogeneous. In particular, evaluation at the operators $A_\ell$ is well defined even when they are unbounded. If $\Lambda(B)=\Tr(B\rho)$ for all $B\in\cB_S(\cH)$, then Proposition~\ref{prop:spectral-trunc-positive} gives $\Lambda(A)=\Tr(A\rho)$.

\begin{prop}[Confinement for the costs with compact resolvent]
\label{prop:B*confine}
Let $C$ be a nonnegative self-adjoint operator on $\cH$ with compact
resolvent, and let $\Lambda\in\cB_S(\cH)^*$ be positive. By
Lemma~\ref{lemma:b*decomposition}, write
\[
\Lambda(B)=\Tr(\pi\,B)+L(B),\quad \pi\in\cT^{(1)}_{S}(\cH),\ L\in\cB_S(\cH)^*,\ L\ge0,\ L|_{\cK_S(\cH)}=0.
\]
Then
\[
\Lambda(C):=\sup\Big\{\Lambda(B):\ B\in\cB_S(\cH),\ B\le C\Big\}
=
\begin{cases}
\Tr(C\,\pi), & \text{if }L\equiv0,\\[2mm]
+\infty, & \text{if }L\not\equiv0.
\end{cases}
\]
\end{prop}

\begin{proof}
If $L\equiv0$ the supremum equals $\sup_{B\le C}\Tr(\pi B)=\Tr(\pi C)$ by monotone convergence along the spectral truncations of $C$.

Assume $L\not\equiv0$. Then $\cH$ must be infinite-dimensional.
Because $C$ has compact resolvent, it has purely discrete spectrum
$\{\lambda_j\}_{j\in\N}$ with $\lambda_j\to+\infty$ by
Corollary~\ref{cor:cost-compact-resolvent}; let
$\{e_j\}_{j\in\N}$ be corresponding eigenvectors. By
Lemma~\ref{lemma:idmnorm}, $L(\I_{\cH})=\|L\|_*>0$. Let $P_n:=\sum_{i=1}^n |e_i\rangle\langle e_i|$; then $P_n\in\cK_S(\cH)$ and $L(P_n)=0$, so
$L(\I_{\cH}-P_n)=L(\I_{\cH})=\|L\|_*$ for all $n$.
Since $\lambda_i\uparrow\infty$, for any $m>0$ there exists $n(m)$ with
$C\ge m(\I_{\cH}-P_{n(m)})$. Hence
\[
\sup_{B\le C}\Lambda(B)\ \ge\ \Lambda\big(m(\I_{\cH}-P_{n(m)})\big)
= m\,\big[\Tr\big(\pi(\I_{\cH}-P_{n(m)})\big)+\|L\|_*\big]\ \geq m\|L\|_*\xrightarrow[m\to\infty]{}\ +\infty.
\]
\end{proof}

\begin{prop}[Dichotomy for quadrature-generated costs]
\label{prop:type2costoperator}
Let $\cH_1,\cH_2$ be separable Hilbert spaces and let
$C\in\cL_{\geq}(\cH_1\otimes\cH_2)$ be a quadrature-generated
cost, with canonical generators $R_j$ on $\cH_1$ and $S_j$ on
$\cH_2$, $j\in[2m]$. Assume that the one-body square-sum
$H_1:=\sum_{j=1}^{2m}R_j^2$ has compact resolvent.
Let $\sigma\in\cS(\cH_2)$, and let
$\Lambda\in\cB_S(\cH_1\otimes\cH_2)^*$ be positive and satisfy
$\Lambda(\I_{\cH_1}\otimes B)=\Tr(B\sigma)$ for every
$B\in\cB_S(\cH_2)$. Then
\[
\Lambda(C)=
\begin{cases}
\Tr(C\pi),&
\text{if }\Lambda(\cdot)=\Tr(\pi\,\cdot)
\text{ for some }\pi\in\cS(\cH_1\otimes\cH_2),\\
+\infty,&\text{otherwise}.
\end{cases}
\]
\end{prop}

\begin{proof}
The marginal condition gives $\Lambda(\I_{\cH_1\otimes\cH_2})=1$.
By Lemma~\ref{lemma:b*decomposition}, write
$\Lambda(B)=\Tr(TB)+L(B)$, where $T\geq0$ is trace-class,
$L\geq0$, and $L$ vanishes on compact operators.
If $L=0$, then $\Tr(T)=1$ and the claim follows from the
normal-state evaluation of $\Lambda(C)$.

Suppose that $\delta:=L(\I_{\cH_1\otimes\cH_2})>0$.
Since the common smooth-vector space of the $S_j$ is dense,
there exists a finite-rank orthogonal projection $P$ on $\cH_2$
whose range is contained in $\bigcap_{j=1}^{2m}\Dom(S_j)$ and such
that $\Tr(P\sigma)>1-\delta/2$.
Set $\kappa_P:=\sum_{j=1}^{2m}\|S_jP\|_{\mathrm{op}}^2<\infty$.
On the common smooth core,
\[
(R_j\otimes P)\psi
=
(\I_{\cH_1}\otimes P)
(R_j\otimes\I_{\cH_2}-\I_{\cH_1}\otimes S_j)\psi
+
(\I_{\cH_1}\otimes PS_j)\psi.
\]
Here $PS_j$ extends to the bounded adjoint of $S_jP$.
Squaring, summing, and passing to the closed forms gives
\[
H_1\otimes P
\leq
2C+2\kappa_P\I_{\cH_1\otimes\cH_2}.
\]
For $a>0$, let $E_a:=\mathbf1_{[0,a)}(H_1)$.
Compact resolvent makes $E_a$ finite rank, and
$H_1\geq a(\I_{\cH_1}-E_a)$. Hence the bounded self-adjoint operator
\[
B_a:=\frac a2(\I_{\cH_1}-E_a)\otimes P
-\kappa_P\I_{\cH_1\otimes\cH_2}
\]
satisfies $B_a\leq C$.
Since $E_a\otimes P$ is finite rank and
$\Tr(T)=1-\delta$, we obtain
\[
\begin{aligned}
\Lambda((\I_{\cH_1}-E_a)\otimes P)
&=\Tr(P\sigma)-\Tr(T(E_a\otimes P))\\
&\geq\Tr(P\sigma)-\Tr(T)>\delta/2.
\end{aligned}
\]
The definition of the extended evaluation therefore yields
$\Lambda(C)\geq\Lambda(B_a)>a\delta/4-\kappa_P$.
Letting $a\rightarrow+\infty$ proves $\Lambda(C)=+\infty$.
\end{proof}

\subsection{\texorpdfstring{Dual relations via the Fenchel--Rockafellar theorem}{Dual relations via the Fenchel-Rockafellar theorem}}\label{subsec:fg}
\begin{prop}[Convexity and properness of the constrained additive functional]
\label{prop:gAE-convex-proper}
Assume that the constrained barycentric state set defined in \eqref{eq:constrain} is nonempty, i.e.
\[
    \cS_{\mathbf A,\mathbf E}(\cH_0)
    :=
    \left\{
    \rho\in\cS(\cH_0):
    \Tr(A_\ell\rho)\le E_\ell,\quad \ell\in[k]
    \right\}
    \neq\varnothing.
\]
Then the functional \(g_{\mathbf A,\mathbf E}:\prod_{s=1}^N \cB_S(\cH_0\otimes\cH_s)\to(-\infty,+\infty]\) defined in \eqref{eq:defg} is convex and proper. Moreover, $g_{\mathbf A,\mathbf E}(\boldsymbol 0)=0$.
\end{prop}
\begin{proof}
Convexity follows from the fact that \(g_{\mathbf A,\mathbf E}\) is an
infimal projection of an affine functional over a convex set of
representations. Explicitly, if \((B_s^{(0)})_{s\in[N]}\) and \((B_s^{(1)})_{s\in[N]}\) admit
representations $(B_s^{(j)})_{s\in[N]}
=
\left((\sum_{\ell=1}^k U_{s,j}^{(\ell)})\oplus V_{s,j}\right)_{s\in[N]}$ for $j=0,1$ with $-\sum_{s=1}^N\alpha_s U_{s,j}^{(\ell)}
\le
\eta_{\ell,j}A_\ell$, then, for \(t\in[0,1]\), the convex combinations
\[
\widetilde U_s^{(\ell)}
=
tU_{s,0}^{(\ell)}+(1-t)U_{s,1}^{(\ell)},\quad
\widetilde V_s
=
tV_{s,0}+(1-t)V_{s,1},\quad
\widetilde\eta_\ell
=
t\eta_{\ell,0}+(1-t)\eta_{\ell,1}
\]
give an admissible representation of
\((tB^{(0)}_s+(1-t)B^{(1)}_s)_{s\in[N]}\). The corresponding objective value is the same
convex combination of the two objective values. 

It remains to show that \(g_{\mathbf A,\mathbf E}\) does not take the value
\(-\infty\). Choose $\rho_0\in\cS_{\mathbf A,\mathbf E}(\cH_0)$. For any admissible representation $\bigl((U_s^{(\ell)})_{s\in[N],\,\ell\in[k]},
(V_s)_{s\in[N]},(\eta_\ell)_{\ell\in[k]}\bigr)$ of \(B\), we have
\[
\begin{aligned}
\sum_{s=1}^N
\alpha_s
\Tr\!\left(B_s(\rho_0\otimes\sigma_s)\right)
&=
\sum_{\ell=1}^k
\Tr\!\left[
\left(
\sum_{s=1}^N\alpha_s U_s^{(\ell)}
\right)\rho_0
\right]
+
\sum_{s=1}^N
\alpha_s\Tr(V_s\sigma_s).
\end{aligned}
\]
Since $-\sum_{s=1}^N\alpha_s U_s^{(\ell)}
\le
\eta_\ell A_\ell$ and $\Tr(A_\ell\rho_0)\le E_\ell$, we obtain $-\Tr\!\left[
\left(
\sum_{s=1}^N\alpha_s U_s^{(\ell)}
\right)\rho_0
\right]
\le
\eta_\ell E_\ell$.

Therefore
\begin{equation}\label{eq:lb}
-\sum_{s=1}^N\alpha_s\Tr(V_s\sigma_s)
+
\sum_{\ell=1}^k\eta_\ell E_\ell
\ge
-
\sum_{s=1}^N
\alpha_s
\Tr\!\left(B_s(\rho_0\otimes\sigma_s)\right).
\end{equation}
The right-hand side is finite because \(B_s\) is bounded and
\(\rho_0\otimes\sigma_s\) is trace-class. Hence the infimum defining
\(g_{\mathbf A,\mathbf E}(B)\) is bounded below for every
\(B\in\Dom(g_{\mathbf A,\mathbf E})\). Thus
\(g_{\mathbf A,\mathbf E}\) never takes the value \(-\infty\).

Finally, setting $\bigl((U_s^{(\ell)})_{s\in[N],\,\ell\in[k]},
(V_s)_{s\in[N]},(\eta_\ell)_{\ell\in[k]}\bigr)$ to be all zero shows that \(g_{\mathbf A,\mathbf E}(0)\le0\). Applying the preceding
lower bound \eqref{eq:lb} with \((B_s)_{s\in[N]}=\boldsymbol{0}\) gives \(g_{\mathbf A,\mathbf E}(0)\ge0\). Hence $g_{\mathbf A,\mathbf E}(\boldsymbol 0)=0$ and \(g_{\mathbf A,\mathbf E}\) is proper.
\end{proof}

\begin{prop}[Identification of the Fenchel--Rockafellar left-hand side]
\label{prop:FR-left-identification-constrained}
Let \(f,g_{\mathbf A,\mathbf E}:\prod_{s=1}^N \cB_S(\cH_0\otimes\cH_s)\to(-\infty,+\infty]\) be defined in \eqref{eq:def:f} and \eqref{eq:defg}, respectively.

Then
\begin{equation}
\label{eq:FR-left-energy-constrained}
\begin{aligned}
&\inf_{(B_s)_{s\in[N]}\in
\prod_{s=1}^N \cB_S(\cH_0\otimes\cH_s)}
\left\{
f((B_s)_{s\in[N]})
+
g_{\mathbf A,\mathbf E}((B_s)_{s\in[N]})
\right\}
\\
&\qquad
=
-
\sup
\left\{
\sum_{s=1}^N \alpha_s\Tr(V_s\sigma_s)
-
\sum_{\ell\in[k]}\eta_\ell E_\ell
:
((U_s,V_s)_{s\in[N]},(\eta_\ell)_{\ell\in[k]})
\in
\mathsf{Pot}_{\mathbf A,\mathbf E}^{\mathbf C,\boldsymbol\alpha}
(\boldsymbol\sigma)
\right\},
\end{aligned}
\end{equation}
where $\mathsf{Pot}_{\mathbf A,\mathbf E}^{\mathbf C,\boldsymbol\alpha}
(\boldsymbol\sigma)$ is defined in \eqref{eq:pot}.
\end{prop}
\begin{proof}
By definition, \(f((B_s)_{s\in[N]})<+\infty\) if and only if
$B_s\le C_s$ for all $s\in[N]$, and in this case \(f((B_s)_{s\in[N]})=0\). For such \((B_s)_{s\in[N]}\), the definition of
\(g_{\mathbf A,\mathbf E}\) gives the infimum of
$-\sum_{s=1}^N\alpha_s\Tr(V_s\sigma_s)+\sum_{\ell=1}^k\eta_\ell E_\ell$
over all representations $B_s=
\left(\sum_{\ell=1}^k U_s^{(\ell)}\right)\oplus V_s$ with
$\eta_\ell\ge0$ and $-\sum_{s=1}^N\alpha_s U_s^{(\ell)}
\le
\eta_\ell A_\ell$ for all $\ell\in[k]$.

Since \(B_s\le C_s\), eliminating the auxiliary variables \(B_s\) is
equivalent to imposing directly $
\left(\sum_{\ell=1}^k U_s^{(\ell)}\right)\oplus V_s
\le C_s$ together with the same energy-potential constraints. This is precisely the
condition that $((U_s,V_s)_{s\in[N]},(\eta_\ell)_{\ell\in[k]})
\in
\mathsf{Pot}_{\mathbf A,\mathbf E}^{\mathbf C,\boldsymbol\alpha}
(\boldsymbol\sigma)$, where \(U_s=\sum_{\ell=1}^k U_s^{(\ell)}\).

Therefore
\[
\begin{aligned}
&\inf_{(B_s)_{s\in[N]}}
\left\{
f((B_s)_{s\in[N]})
+
g_{\mathbf A,\mathbf E}((B_s)_{s\in[N]})
\right\}
\\
&\quad =
\inf
\left\{
-\sum_{s=1}^N\alpha_s\Tr(V_s\sigma_s)
+
\sum_{\ell=1}^k\eta_\ell E_\ell
:
((U_s,V_s)_{s\in[N]},(\eta_\ell)_{\ell\in[k]})
\in
\mathsf{Pot}_{\mathbf A,\mathbf E}^{\mathbf C,\boldsymbol\alpha}
(\boldsymbol\sigma)
\right\}
\\
&\quad =
-
\sup
\left\{
\sum_{s=1}^N\alpha_s\Tr(V_s\sigma_s)
-
\sum_{\ell=1}^k\eta_\ell E_\ell
:
((U_s,V_s)_{s\in[N]},(\eta_\ell)_{\ell\in[k]})
\in
\mathsf{Pot}_{\mathbf A,\mathbf E}^{\mathbf C,\boldsymbol\alpha}
(\boldsymbol\sigma)
\right\}.
\end{aligned}
\]
This proves \eqref{eq:FR-left-energy-constrained}.
\end{proof}

\begin{prop}[Conjugate of the cost-constraint indicator]
\label{prop:conjugate-f-cost-constraint}
Let \(f:\prod_{s=1}^N \cB_S(\cH_0\otimes\cH_s)\to(-\infty,+\infty]\) be defined in \eqref{eq:def:f}. Then, for every $\Lambda=(\Lambda_s)_{s\in[N]}\in \prod_{s=1}^N \cB_S(\cH_0\otimes\cH_s)^*$, one has the Fenchel conjugate $f^*:\prod_{s=1}^N \cB_S(\cH_0\otimes\cH_s)^*\to(-\infty,+\infty]$:
\[
f^*(\Lambda)
=
\begin{cases}
\displaystyle
\sum_{s=1}^N \Lambda_s(C_s),
&
\Lambda_s\ge0\ \text{for all }s\in[N],
\\[2mm]
+\infty,
&
\text{otherwise}
\end{cases}\rev{.}
\]
Here \(\Lambda_s\ge0\) means that $\Lambda_s(B)\ge0$ for all $B\in\cB_\geq(\cH_0\otimes\cH_s)$, and for a positive functional \(\Lambda_s\), its extended evaluation on the possibly unbounded cost \(C_s\) is understood as in \eqref{eq:ext-eval}.
\end{prop}
\begin{proof}
Suppose that \(\Lambda_{s_0}\) is not positive for some
\(s_0\in[N]\). Then there exists $D\in\cB_\geq(\cH_0\otimes\cH_{s_0})$ such that \(\Lambda_{s_0}(D)<0\). For every \(t>0\), set $B_{s_0}:=-tD$ and $B_s:=0$ for \(s\neq s_0\). Since \(C_s\ge0\), this choice satisfies \(B_s\le C_s\) for all \(s\).
Hence
\[
f^*(\Lambda)
\ge
\sum_{s=1}^N\Lambda_s(B_s)
=
-t\Lambda_{s_0}(D)\to+\infty
\quad \text{as }t\to+\infty.
\]
Thus \(f^*(\Lambda)=+\infty\).

It remains to consider the case in which every $\Lambda_s$ is
positive. By the definition of the Fenchel conjugate and the
independence of the constraints in $s$,
\[
f^*(\Lambda)
=
\sum_{s=1}^N
\sup_{\substack{B_s\in\cB_S(\cH_0\otimes\cH_s)\\ B_s\le C_s}}
\Lambda_s(B_s)
=
\sum_{s=1}^N\Lambda_s(C_s),
\]
where the last equality is precisely the extended evaluation from
\eqref{eq:ext-eval}.
\end{proof}
\begin{prop}[Conjugate of the constrained additive-potential functional]
\label{prop:conjugate-gAE}
Assume that \(\alpha_s>0\) for all \(s\in[N]\). Let $g_{\mathbf A,\mathbf E}$ be defined as in \eqref{eq:defg}. For every $\Lambda=(\Lambda_s)_{s\in[N]}
\in
\prod_{s=1}^N
\cB_S(\cH_0\otimes\cH_s)^*$, we have the Fenchel conjugate $g_{\mathbf A,\mathbf E}^*:\prod_{s=1}^N \cB_S(\cH_0\otimes\cH_s)^*\to(-\infty,+\infty]$:
\[
g_{\mathbf A,\mathbf E}^*(\Lambda)
=
\iota_{\mathcal C_{\mathrm{bar}}^{\mathbf A,\mathbf E}}(\Lambda)
=
\begin{cases}
0,
&
\Lambda\in\mathcal C_{\mathrm{bar}}^{\mathbf A,\mathbf E},
\\[1mm]
+\infty,
&
\rev{\text{otherwise}}
\end{cases},
\]
where \(\mathcal C_{\mathrm{bar}}^{\mathbf A,\mathbf E}\) is defined by:
\[
\begin{aligned}
\mathcal C_{\mathrm{bar}}^{\mathbf A,\mathbf E}
:=
\Bigl\{
&
\Lambda=(\Lambda_s)_{s\in[N]}
\in
\prod_{s=1}^N
\cB_S(\cH_0\otimes\cH_s)^*
:
\\
&
\exists\,
\Lambda_{\mathrm{bar}}\in\cB_S(\cH_0)^*, \Lambda_{\mathrm{bar}}\geq0,
\ \text{such that}
\\
&
\Lambda_{\mathrm{bar}}(A_\ell)\le E_\ell,
\qquad \ell\in[k],
\\
&
\Lambda_s(U\otimes \I_{\cH_s})
=
-\alpha_s\Lambda_{\mathrm{bar}}(U),
\qquad
U\in\cB_S(\cH_0),\ s\in[N],
\\
&
\Lambda_s(\I_{\cH_0}\otimes V)
=
-\alpha_s\Tr(V\sigma_s),
\qquad
V\in\cB_S(\cH_s),\ s\in[N]
\Bigr\}.
\end{aligned}
\]
\end{prop}

\begin{proof}
By the definition of \(g_{\mathbf A,\mathbf E}\), its Fenchel conjugate is
the supremum:
\[
\begin{aligned}
g_{\mathbf A,\mathbf E}^*(\Lambda)
=
\sup
\Bigg\{
&
\sum_{\ell=1}^k\sum_{s=1}^N
\Lambda_s(U_s^{(\ell)}\otimes\I_{\cH_s})
+
\sum_{s=1}^N
\left[
\Lambda_s(\I_{\cH_0}\otimes V_s)
+
\alpha_s\Tr(V_s\sigma_s)
\right]
-
\sum_{\ell=1}^k\eta_\ell E_\ell
\Bigg\},
\end{aligned}
\]
over all $\eta_\ell\ge0$, $V_s\in\cB_S(\cH_s)$, $U_s^{(\ell)}\in\cB_S(\cH_0)$ such that $-\sum_{s=1}^N\alpha_s U_s^{(\ell)}
\le
\eta_\ell A_\ell$ for all $\ell\in[k]$.

The variables \(V_s\) are unconstrained. Therefore the supremum over
\(V_s\) is finite if and only if $\Lambda_s(\I_{\cH_0}\otimes V)=-\alpha_s\Tr(V\sigma_s)$ for every $V\in\cB_S(\cH_s)$ and $s\in[N]$.

Assume from now on that the above condition holds. Then the \(V_s\)-part
contributes \(0\). For the \(U_s^{(\ell)}\)-variables, fix
\(\ell\in[k]\). Taking \(\eta_\ell=0\) and restricting to variations with
$\sum_{s=1}^N\alpha_s U_s^{(\ell)}=0$ shows that finiteness forces the map
\[
(U_s^{(\ell)})_{s\in[N]}
\longmapsto
\sum_{s=1}^N\Lambda_s(U_s^{(\ell)}\otimes \I_{\cH_s})
\]
to vanish on this kernel. Hence there exists
$\Lambda_{\mathrm{bar}}\in\cB_S(\cH_0)^*$ such that
\[
\Lambda_s(U\otimes \I_{\cH_s})
=
-\alpha_s\Lambda_{\mathrm{bar}}(U),
\qquad
U\in\cB_S(\cH_0),\ s\in[N],
\]
and the constraint permits
$-\Lambda_{\mathrm{bar}}(\sum_{s=1}^N\alpha_s U_s^{(\ell)})\leq0$ for all $\sum_{s=1}^N\alpha_s U_s^{(\ell)}\ge0$, which requires $\Lambda_{\mathrm{bar}}\geq0$.

Under these conditions, set
$X_\ell:=-\sum_{s=1}^N\alpha_s U_s^{(\ell)}\in\cB_S(\cH_0)$. Because every $\alpha_s$ is positive, every bounded self-adjoint $X_\ell$ can be realized by a choice of the $U_s^{(\ell)}$. For $\eta_\ell=0$, the supremum over $X_\ell\leq0$ is zero. For $\eta_\ell>0$, positive homogeneity of the extended evaluation in~\eqref{eq:ext-eval} gives the following expression without any $0\cdot(+\infty)$ ambiguity:
\[
\begin{aligned}
&\sup_{\substack{\eta_\ell\ge0,\ X_\ell\in\cB_S(\cH_0)\\
X_\ell\le\eta_\ell A_\ell}}
\bigl\{\Lambda_{\mathrm{bar}}(X_\ell)-\eta_\ell E_\ell\bigr\}
=
\max\left\{0,
\sup_{\eta_\ell>0}
\eta_\ell\bigl(\Lambda_{\mathrm{bar}}(A_\ell)-E_\ell\bigr)
\right\}
\\
&\qquad=
\begin{cases}
0,&\Lambda_{\mathrm{bar}}(A_\ell)\le E_\ell,\\
+\infty,&\Lambda_{\mathrm{bar}}(A_\ell)>E_\ell.
\end{cases}
\end{aligned}
\]
Here the second case includes $\Lambda_{\mathrm{bar}}(A_\ell)=+\infty$; thus the argument applies without assuming that $A_\ell$ is bounded.

Therefore
$g_{\mathbf A,\mathbf E}^*(\Lambda)=0$ on \(\mathcal C_{\mathrm{bar}}^{\mathbf A,\mathbf E}\), and it is
\(+\infty\) outside this set.
\end{proof}

\begin{prop}[Identification of the Fenchel--Rockafellar right-hand side]
\label{prop:FR-right-identification-constrained}
Let \(f,g_{\mathbf A,\mathbf E}:\prod_{s=1}^N \cB_S(\cH_0\otimes\cH_s)\to(-\infty,+\infty]\) be defined in \eqref{eq:def:f} and \eqref{eq:defg}.

Assume that there exists $s_0\in[N]$ such that either $C_{s_0}$ has compact resolvent or $C_{s_0}$ is quadrature-generated and its associated first one-body square-sum has compact resolvent. Then
\[
\begin{aligned}
&\sup_{\Lambda\in
\prod_{s=1}^N\cB_S(\cH_0\otimes\cH_s)^*}
\left\{
-f^*(-\Lambda)-g_{\mathbf A,\mathbf E}^*(\Lambda)
\right\} =
-
\inf_{(\pi_s)_s\in
\Gamma_{\mathrm{joint}}^{\mathbf A,\mathbf E}(\boldsymbol\sigma)}
\sum_{s=1}^N
\alpha_s\Tr(\pi_s C_s).
\end{aligned}
\]
\end{prop}
\begin{proof}
By Propositions~\ref{prop:conjugate-f-cost-constraint} and
\ref{prop:conjugate-gAE}, the expression $-f^*(-\Lambda)-g_{\mathbf A,\mathbf E}^*(\Lambda)$ is finite only if $-\Lambda_s\ge0$ for all $s\in[N]$ and
$\Lambda\in\mathcal C_{\mathrm{bar}}^{\mathbf A,\mathbf E}$. Then there exists
\(\Lambda_{\mathrm{bar}}\in\cB_S(\cH_0)^*,\ \Lambda_{\mathrm{bar}}\geq0\) such that $\Lambda_{\mathrm{bar}}(A_\ell)\le E_\ell$ for all $\ell\in[k]$, and
\[
-\Lambda_s(U\otimes \I_{\cH_s})
=
\alpha_s\Lambda_{\mathrm{bar}}(U),
\qquad
-\Lambda_s(\I_{\cH_0}\otimes V)
=
\alpha_s\Tr(V\sigma_s).
\]
Moreover, on this feasible set,
Proposition~\ref{prop:conjugate-f-cost-constraint} gives
\[
-f^*(-\Lambda)-g_{\mathbf A,\mathbf E}^*(\Lambda)
=-\sum_{s=1}^N(-\Lambda_s)(C_s),
\]
where every extended evaluation is applied to the positive functional
$-\Lambda_s$.

If the above value is finite, then
\((-\Lambda_s)(C_s)<+\infty\) for every \(s\). Set $\Xi_s:=-\alpha_s^{-1}\Lambda_s$.
Then $\Xi_s\ge0$ and, by the second marginal identity,
$\Xi_s(\I_{\cH_0}\otimes V)=\Tr(V\sigma_s)$ for every
$V\in\cB_S(\cH_s)$.

We use the following normality observation. Let $\Xi\in\cB_S(\cH\otimes\cK)^*$ be positive. If both marginal functionals $U\mapsto\Xi(U\otimes\I_{\cK})$ and $V\mapsto\Xi(\I_{\cH}\otimes V)$ are normal, then $\Xi$ is normal. Indeed, choose increasing finite-rank projections $P_n\uparrow\I_{\cH}$ and $Q_n\uparrow\I_{\cK}$ such that both normal marginal tails converge to zero. Since
\[
\I_{\cH\otimes\cK}-P_n\otimes Q_n
\leq
(\I_{\cH}-P_n)\otimes\I_{\cK}
+
\I_{\cH}\otimes(\I_{\cK}-Q_n),
\]
one has $\Xi(\I_{\cH\otimes\cK}-P_n\otimes Q_n)\to0$. Write $\Xi(B)=\Tr(TB)+L(B)$ as in Lemma~\ref{lemma:b*decomposition}. Because $P_n\otimes Q_n$ is finite rank, $L(P_n\otimes Q_n)=0$, and hence
\[
0\leq L(\I_{\cH\otimes\cK})
=L(\I_{\cH\otimes\cK}-P_n\otimes Q_n)
\leq\Xi(\I_{\cH\otimes\cK}-P_n\otimes Q_n)\longrightarrow0.
\]
Thus $L=0$, proving the observation.

Since $\Xi_{s_0}(C_{s_0})<+\infty$, Proposition~\ref{prop:B*confine} in the compact-resolvent case and Proposition~\ref{prop:type2costoperator} in the quadrature-generated case show that $\Xi_{s_0}$ is normal. Its first marginal is $\Lambda_{\mathrm{bar}}$, so there exists $\rho\in\cS(\cH_0)$ such that $\Lambda_{\mathrm{bar}}(U)=\Tr(U\rho)$ for every $U\in\cB_S(\cH_0)$. For every $s\in[N]$, the first marginal of $\Xi_s$ is therefore the normal state induced by $\rho$, while its second marginal is the normal state induced by $\sigma_s$. The normality observation gives a state $\pi_s\in\cS(\cH_0\otimes\cH_s)$ such that $-\Lambda_s(\cdot)=\alpha_s\Tr(\pi_s\,\cdot)$.
The marginal identities give $\Tr_{\cH_s}\pi_s=\rho$ and $\Tr_{\cH_0}\pi_s=\sigma_s$. By~\eqref{eq:ext-eval} and Proposition~\ref{prop:spectral-trunc-positive}, the bound
\(\Lambda_{\mathrm{bar}}(A_\ell)\le E_\ell\) is precisely
$\Tr(A_\ell\rho)\le E_\ell$ for all $\ell\in[k]$, also when $A_\ell$ is unbounded.
The same extended-evaluation identity gives $\Xi_s(C_s)=\Tr(\pi_sC_s)$ for every $s\in[N]$.

Thus $(\pi_s)_s\in
\Gamma_{\mathrm{joint}}^{\mathbf A,\mathbf E}(\boldsymbol\sigma)$, and $-\sum_{s=1}^N(-\Lambda_s)(C_s)=-\sum_{s=1}^N\alpha_s\Tr(\pi_sC_s)$.

Conversely, if $(\pi_s)_s\in
\Gamma_{\mathrm{joint}}^{\mathbf A,\mathbf E}(\boldsymbol\sigma)$ has common barycentric marginal \(\rho\), define $-\Lambda_s(\cdot):=\alpha_s\Tr(\pi_s\,\cdot)$ and set $\Lambda_{\mathrm{bar}}(\cdot):=\Tr(\rho\,\cdot)$. Then $\Lambda_{\mathrm{bar}}\geq0$, \(-\Lambda_s\ge0\), \(\Lambda\in
\mathcal C_{\mathrm{bar}}^{\mathbf A,\mathbf E}\), and $-f^*(-\Lambda)-g_{\mathbf A,\mathbf E}^*(\Lambda)=-\sum_{s=1}^N\alpha_s\Tr(\pi_sC_s)$.

Therefore the Fenchel--Rockafellar right-hand side is exactly the supremum
of the negatives of the constrained joint primal costs:
\[
\begin{aligned}
&\sup_{\Lambda\in
\prod_{s=1}^N\cB_S(\cH_0\otimes\cH_s)^*}
\left\{
-f^*(-\Lambda)-g_{\mathbf A,\mathbf E}^*(\Lambda)
\right\}=
-
\inf_{(\pi_s)_s\in
\Gamma_{\mathrm{joint}}^{\mathbf A,\mathbf E}(\boldsymbol\sigma)}
\sum_{s=1}^N
\alpha_s\Tr(\pi_s C_s).
\end{aligned}
\]
\end{proof}
 
\section{Technical arguments for Gaussian rigidity}
\label{subsec:faithful-gaussian-global-uniqueness-app}
We fix $m,N\in\N$, weights $\alpha_s>0$, $s\in[N]$, satisfying $\sum_{s\in[N]}\alpha_s=1$, and covariances $\{\Sigma_s\}_{s\in[N]}\subset\mathrm{Cov}_m^{\mathrm Q}$.
\subsection{Fenchel conjugates for the covariance dual}
\label{app:faithful-covariance-duality}
\begin{prop}
\label{prop:dualfunccov}
Recall $\lcc$, $\fc^\tau$ and $\gc^\tau$ defined in Section~\ref{subsubsec:faithful-covariance-dual}. Then $\fc^\tau$ and $\gc^\tau$ are proper convex functions. Moreover,
for $\mathbf Z=(Z_r)_{r\in[k]}$ with
$Z_r\in\cM_{4m,\mathrm{sa}}(\C)$, let $(Z_r)_{ij}$ denote its
$(i,j)$-th $2m\times2m$ block. Then
\[
\lcc^*(\mathbf Z)
=
\left(
\sum_{r=1}^k\Real (Z_r)_{11},\,
2\Real (Z_1)_{12},\ldots,2\Real (Z_k)_{12}
\right),
\]
and
\[
(\fc^\tau)^*(\mathbf Z)
=
\begin{cases}
\displaystyle
-\sum_{r=1}^k
\Real\Tr\left[
Z_r^\dagger
\begin{pmatrix}
\frac{i\hbar}{2}\Omega_m&0\\
0&\Sigma_r+\frac{i\tau\hbar}{2}\Omega_m
\end{pmatrix}
\right],
&
Z_r\preceq0,\quad r\in[k],
\\
+\infty,&\text{otherwise}.
\end{cases}
\]

If $\Real (Z_r)_{12}\neq-\alpha_r\I_{2m}$ for some $r\in[k]$, then $(\gc^\tau)^*(\mathbf Z)=+\infty$. Otherwise, when $k<N$,
\[
\begin{aligned}
(\gc^\tau)^*(\mathbf Z)
=
\sup_{\substack{
\Sigma\in\mathrm{Cov}_m^{\mathrm Q},\\
A_s\in\mathrm{Corr}_\tau(\Sigma,\Sigma_s),\
s=k+1,\ldots,N}}
\left\{
\Tr\left[
\left(
\sum_{r=1}^k\Real (Z_r)_{11}-\I_{2m}
\right)\Sigma
\right]
+
2\sum_{s=k+1}^N\alpha_s\Tr(A_s)
\right\}.
\end{aligned}
\]
If $k=N$, then
\[
(\gc^\tau)^*(\mathbf Z)
=
\sup_{\Sigma\in\cM_{2m,\mathrm{sym}}(\R)}
\Tr\left[
\left(
\sum_{r=1}^N\Real (Z_r)_{11}-\I_{2m}
\right)\Sigma
\right].
\]
\end{prop}
\begin{proof}
We first verify that $\fc^\tau$ and $\gc^\tau$ are proper and convex.
For any $t>\hbar/2$, set $\mathbf V_0:=\lcc(t\I_{2m},0,\ldots,0)$. By \eqref{eq:faithful-inputs},
\[
(\mathbf V_0)_r+
\begin{pmatrix}
\frac{i\hbar}{2}\Omega_m&0\\
0&\Sigma_r+\frac{i\tau\hbar}{2}\Omega_m
\end{pmatrix}
\succ0,
\quad r\in[k],
\]
so $\fc^\tau$ is proper and is convex by the convexity of $\ccov^\tau$.

The domain of $\gc^\tau$ is nonempty, since for $t>\hbar/2$, $0\in\mathrm{Corr}_\tau(t\I_{2m},\Sigma_s)$ for all $s\in[N]$. For fixed $\mathbf V$, the relation
$\lcc(\Sigma,A_1,\ldots,A_k)=\mathbf V$ fixes
$\Sigma,A_1,\ldots,A_k$. Moreover, for each $s>k$,
$\mathrm{Corr}_\tau(\Sigma,\Sigma_s)$ is bounded: by the standard factorization of a positive block matrix \cite[Prop.~1.3.2]{Bhatia2007PDM}, $A_s\in\mathrm{Corr}_\tau(\Sigma,\Sigma_s)$ if and only if there exists $K_s\in\cM_{2m}(\C)$ such that $\|K_s\|_{\mathrm{op}}\leq1$ and 
\[ 
A_s=(\Sigma+\frac{i\hbar}{2}\Omega_m)^{1/2}K_s(\Sigma_s+\frac{i\tau\hbar}{2}\Omega_m)^{1/2}. 
\]
Therefore $\mathrm{Corr}_\tau(\Sigma,\Sigma_s)$ is compact since it is closed by the defining linear matrix inequality. Thus every nonempty fiber of $\gc^\tau$ is compact, and the linear objective attains a finite minimum. Convexity follows directly from
the convexity of the sets
$\mathrm{Corr}_\tau(\Sigma,\Sigma_s)$ in the joint variables
$(\Sigma,A_s)$ and the linearity of the objective.

For $\mathbf Z=(Z_r)_{r\in[k]}$, a direct block computation gives
\[
\Real\sum_{r=1}^k
\Tr\left[
Z_r^\dagger
\begin{pmatrix}
\Sigma&A_r\\
A_r^\top&0
\end{pmatrix}
\right]
=
\Tr\left[
\left(\sum_{r=1}^k\Real (Z_r)_{11}\right)\Sigma
\right]
+
2\sum_{r=1}^k
\Tr\left[(\Real (Z_r)_{12})^\top A_r\right],
\]
which proves the formula for $\lcc^*$.

By definition of $\fc^\tau$,
\[
\begin{aligned}
(\fc^\tau)^*(\mathbf Z)
&=
-\sum_{r=1}^k
\Real\Tr\left[
Z_r^\dagger
\begin{pmatrix}
\frac{i\hbar}{2}\Omega_m&0\\
0&\Sigma_r+\frac{i\tau\hbar}{2}\Omega_m
\end{pmatrix}
\right]
+
\sum_{r=1}^k
\sup_{X_r\succeq0}\Real\Tr(Z_r^\dagger X_r).
\end{aligned}
\]
The last supremum is zero if $Z_r\preceq0$ for every $r\in[k]$,
since then $\Real\Tr(Z_r^\dagger X_r)\leq0$ for all $X_r\succeq0$,
with equality at $X_r=0$. If $Z_r\not\preceq0$ for some $r$, choosing
$X_r=tuu^\dagger$ along a direction $u$ with $u^\dagger Z_ru>0$ and
letting $t\to\infty$ shows that the supremum is $+\infty$. This proves
the stated formula for $(\fc^\tau)^*$.

Finally,
\[
\begin{aligned}
(\gc^\tau)^*(\mathbf Z)
=
\sup\quad&
\Tr\left[
\left(
\sum_{r=1}^k\Real (Z_r)_{11}-\I_{2m}
\right)\Sigma
\right]
+
2\sum_{r=1}^k
\Tr\left[
\bigl(\Real (Z_r)_{12}+\alpha_r\I_{2m}\bigr)^\top A_r
\right]
\\
&+
2\sum_{s=k+1}^N\alpha_s\Tr(A_s).
\end{aligned}
\]
When $k<N$, the supremum is taken over
$\Sigma\in\mathrm{Cov}_m^{\mathrm Q}$ and
$A_s\in\mathrm{Corr}_\tau(\Sigma,\Sigma_s)$ for $s>k$, while
$A_1,\ldots,A_k$ are unrestricted. When $k=N$, it is taken over
$\Sigma\in\cM_{2m,\mathrm{sym}}(\R)$ and unrestricted
$A_1,\ldots,A_N$. 

Hence finiteness requires $\Real (Z_r)_{12}=-\alpha_r\I_{2m}$ for $r\in[k]$. Under these conditions, the terms involving $A_1,\ldots,A_k$ vanish. For $k<N$, this gives
\[
(\gc^\tau)^*(\mathbf Z)
=
\sup_{\substack{
\Sigma\in\mathrm{Cov}_m^{\mathrm Q},\\
A_s\in\mathrm{Corr}_\tau(\Sigma,\Sigma_s),\
s=k+1,\ldots,N}}
\left\{
\Tr\left[
\left(
\sum_{r=1}^k\Real (Z_r)_{11}-\I_{2m}
\right)\Sigma
\right]
+
2\sum_{s=k+1}^N\alpha_s\Tr(A_s)
\right\}.
\]
If $k=N$, it reduces to
\[
(\gc^\tau)^*(\mathbf Z)
=
\sup_{\Sigma\in\cM_{2m,\mathrm{sym}}(\R)}
\Tr\left[
\left(
\sum_{r=1}^N\Real (Z_r)_{11}-\I_{2m}
\right)\Sigma
\right].
\]
This proves the stated formulas.
\end{proof}

\subsection{Symplectic coordinates for the covariance slackness relation}
\label{app:faithful-gaussian-channel-representation}

We denote by $U(m)$ and $O(m)$ the unitary and orthogonal matrix
groups, respectively. The following proposition is adapted from
\cite[Prop.~2.12]{deGosson2006}, with the shuffle permutation
$P_m^{\mathrm{shuf}}$ accounting for our interleaved quadrature
ordering.

\begin{prop}[Unitary--orthosymplectic identification]
\label{prop:unitary-realification}
Let $P^{\mathrm{shuf}}_m$ be the permutation matrix such that $(P^{\mathrm{shuf}}_m)^\top
(x_1,y_1,\ldots,x_m,y_m)^\top
=
(x_1,\ldots,x_m,y_1,\ldots,y_m)^\top$.

For $V\in U(m)$, define
\[
O_V
:=
P^{\mathrm{shuf}}_m
\begin{pmatrix}
\Real V&-\Imag V\\
\Imag V&\Real V
\end{pmatrix}
(P^{\mathrm{shuf}}_m)^\top.
\]
Then
$O_V\in\Sp(2m,\R)\cap O(2m)$, and
$V\mapsto O_V$ is a group isomorphism
\[
U(m)
\cong
\Sp(2m,\R)\cap O(2m).
\]
\end{prop}

\begin{remark}
\label{rem:unitary-spectral-coordinates}
Let $U_0=
\frac1{\sqrt2}
\smat{
1&1\\
i&-i}$ and $\mathbb U_m:=\bigoplus_{j=1}^m{}_D U_0$. With the same permutation matrix
$P^{\mathrm{shuf}}_m$ as in
Proposition~\ref{prop:unitary-realification}, one has
\[
O_{\overline V}
=
\mathbb U_mP^{\mathrm{shuf}}_m
\bigl(V\oplus_D\overline V\bigr)
(P^{\mathrm{shuf}}_m)^\top
\mathbb U_m^\dagger.
\]
\end{remark}
\subsection{Gaussian-channel representation under covariance slackness}

We first recall a source representation obtained from the pure
Williamson modes of a covariance matrix.

\begin{lemma}[Vacuum-mode representation from a Williamson decomposition]
\label{lem:Williamson-Gaussian-source}
Let $\hat\pi\in\cS_2^R(L_2(\R^{2m}))$ have mean vector
\rev{$\mathbf m_{\hat\pi}\in\R^{4m}$} and covariance matrix
$\Sigma\in\mathrm{Cov}_{2m}^{\mathrm Q}$. Let $p$ be the number of
symplectic eigenvalues of $\Sigma$ equal to $\hbar/2$.

For every integer $q$ with $0\leq q\leq p$, there exist
$S\in\Sp(4m,\R)$, a Gaussian unitary $\hat U_S$
implementing $S$, and a state
$\hat\eta\in\cS_2^R(L_2(\R^{2m-q}))$ with zero mean such that the map
$\Phi_G:\mathcal T^{(1)}(L_2(\R^{2m-q}))
\rightarrow\mathcal T^{(1)}(L_2(\R^{2m}))$ defined by
\begin{equation}
\begin{aligned}
\Phi_G(\hat\rho)
:=
&W_{\mathrm{CCR}}\!\left(
\frac{\Omega_{2m}\rev{\mathbf m_{\hat\pi}}}{\hbar}
\right)
\hat U_S
\left(
|0\rangle\langle0|^{\otimes q}
\otimes\hat\rho
\right)
\hat U_S^\dagger
W_{\mathrm{CCR}}\!\left(
\frac{\Omega_{2m}\rev{\mathbf m_{\hat\pi}}}{\hbar}
\right)^\dagger
\end{aligned}
\label{eq:faithful-Williamson-Gaussian-channel}
\end{equation}
is a quantum channel satisfying $\hat\pi=\Phi_G(\hat\eta)$.
\end{lemma}
\begin{proof}
Choose a Williamson decomposition $\Sigma=SDS^\top$ with
$D=\frac{\hbar}{2}\I_{2p}\oplus_D
\bigoplus_{j=p+1}^{2m}{}_D\nu_j\I_2$ and $\nu_j>\hbar/2$ for $j>p$.
Let $\hat U_S$ be a Gaussian unitary implementing $S$, and define
$\hat\pi_0:=\hat U_S^\dagger
W_{\mathrm{CCR}}(\frac{\Omega_{2m}\rev{\mathbf m_{\hat\pi}}}{\hbar})^\dagger
\hat\pi
W_{\mathrm{CCR}}(\frac{\Omega_{2m}\rev{\mathbf m_{\hat\pi}}}{\hbar})
\hat U_S$.
Then $\hat\pi_0$ has zero mean and covariance $D$.

Each of its first $p$ one-mode marginals has zero mean and covariance
$(\hbar/2)\I_2$, and hence is the one-mode Fock ground state
$|0\rangle\langle0|$. Since a state with a pure marginal factorizes,
there exists a state $\hat\eta_0$ on $2m-p$ modes such that
$\hat\pi_0=|0\rangle\langle0|^{\otimes p}\otimes\hat\eta_0$.

For any $q\leq p$, set
$\hat\eta:=|0\rangle\langle0|^{\otimes(p-q)}\otimes\hat\eta_0$.
Then $\hat\eta$ is a state on $2m-q$ modes with zero first moments,
and
\[
\hat\pi
=
W_{\mathrm{CCR}}\!\left(
\frac{\Omega_{2m}\rev{\mathbf m_{\hat\pi}}}{\hbar}
\right)
\hat U_S
\left(
|0\rangle\langle0|^{\otimes q}
\otimes\hat\eta
\right)
\hat U_S^\dagger
W_{\mathrm{CCR}}\!\left(
\frac{\Omega_{2m}\rev{\mathbf m_{\hat\pi}}}{\hbar}
\right)^\dagger.
\]
\end{proof}

\begin{lemma}[Gaussian channel induced by a vacuum ancilla]
\label{lem:fock-ancilla-Gaussian-channel}
In the setting of
Lemma~\ref{lem:Williamson-Gaussian-source}, the channel $\Phi_G$ is
Gaussian. More precisely, write the column partition
$S=(S_0\ S_1)$, where
$S_0\in\R^{4m\times2q}$ and
$S_1\in\R^{4m\times2(2m-q)}$ according to the Fock-ground-state and
input modes. Then
\begin{equation}
\chi_{\Phi_G(\hat\rho)}(\boldsymbol\xi)
=
\chi_{\hat\rho}(S_1^\top\boldsymbol\xi)
\exp\left(
i\rev{\mathbf m_{\hat\pi}}^\top\boldsymbol\xi
-
\frac12
\boldsymbol\xi^\top
\left(
\frac{\hbar}{2}S_0S_0^\top
\right)
\boldsymbol\xi
\right).
\label{eq:fock-ancilla-Gaussian-channel}
\end{equation}
Thus, in the notation of Definition~\ref{def:gaussian-channel},
$\mathbb X_G=S_1$,
$\mathbb Y_G=\frac{\hbar}{2}S_0S_0^\top$, and
$\mathbf m_G=\rev{\mathbf m_{\hat\pi}}$.
\end{lemma}
\begin{proof}
The Weyl transformation rule gives
\[
\rev{W_{\mathrm{CCR}}(\Omega_{2m}\mathbf m_{\hat\pi}/\hbar)^\dagger
W_{\mathrm{CCR}}(\boldsymbol\xi)
W_{\mathrm{CCR}}(\Omega_{2m}\mathbf m_{\hat\pi}/\hbar)
=e^{i\mathbf m_{\hat\pi}^\top\boldsymbol\xi}W_{\mathrm{CCR}}(\boldsymbol\xi)}.
\]
The Gaussian-unitary transformation rule gives
$\hat U_S^\dagger W_{\mathrm{CCR}}(\boldsymbol\xi)\hat U_S
=W_{\mathrm{CCR}}(S^\top\boldsymbol\xi)$.
Since
$S^\top\boldsymbol\xi
=(S_0^\top\boldsymbol\xi,S_1^\top\boldsymbol\xi)$,
the characteristic function factorizes as
\[
\chi_{\Phi_G(\hat\rho)}(\boldsymbol\xi)
=e^{i\rev{\mathbf m_{\hat\pi}}^\top\boldsymbol\xi}
\chi_{|0\rangle\langle0|^{\otimes q}}(S_0^\top\boldsymbol\xi)
\chi_{\hat\rho}(S_1^\top\boldsymbol\xi).
\]
The $q$-mode Fock ground state has covariance
$(\hbar/2)\I_{2q}$. Substituting its characteristic function in the
preceding identity gives \eqref{eq:fock-ancilla-Gaussian-channel}.
\end{proof}

\begin{lemma}[Symplectic normal form under covariance slackness]
\label{lem:faithful-dual-Williamson}
Let $\Sigma$ and $p$ be as in
Lemma~\ref{lem:Williamson-Gaussian-source}. Suppose that
$Z\in\cM_{4m,\geq}(\C)$ satisfies
$(\Sigma+\frac{i\hbar}{2}\Omega_{2m})Z=0$.
Then $\Real Z\succeq0$ and
$\rank(\Real Z)$ is even. Writing
$q:=\frac12\rank(\Real Z)$, one has $q\leq p$.

The parameters in
Lemmas~\ref{lem:Williamson-Gaussian-source}
and~\ref{lem:fock-ancilla-Gaussian-channel}
can be chosen with this value of $q$ and with a symplectic matrix
$S\in\Sp(4m,\R)$ such that for some $d_1,\ldots,d_q>0$ and $\rev{\nu_{p+1},\ldots,\nu_{2m}}>\frac{\hbar}{2}\rev{,}$
\[
\Sigma=SDS^\top,\quad
D=\frac{\hbar}{2}\I_{2p}\oplus_D
\bigoplus_{j=p+1}^{2m}{}_D\nu_j\I_2,\quad
S^\top ZS
=
\bigoplus_{j=1}^q{}_D
\frac{d_j}{2}
\smat{1&-i\\i&1}
\oplus_D0_{2(2m-q)},
\]
and $\ker(\Real Z)=S\left(\{0\}^{2q}\oplus\R^{2(2m-q)}\right)$.

Writing the column partition
$S=(S_0\ S_1)$, where
$S_0\in\R^{4m\times2q}$ and
$S_1\in\R^{4m\times2(2m-q)}$, one has
$\Ran(S_1)=\ker(\Real Z)$. In particular, $S_1$ has
full column rank.
\end{lemma}
\begin{proof}
Since $Z\succeq0$ is Hermitian, $\Real Z\succeq0$. Let $\Sigma=S_{\mathrm W}DS_{\mathrm W}^\top$ be a Williamson
decomposition of $\Sigma$ as in
Theorem~\ref{thm:williamson}, where
$D=\frac{\hbar}{2}\I_{2p}\oplus_D
\bigoplus_{j=p+1}^{2m}{}_D\nu_j\I_2$ and $\nu_j>\hbar/2$.
Since $S_{\mathrm W}\in\Sp(4m,\R)$, we have
$\Sigma+\frac{i\hbar}{2}\Omega_{2m}
=S_{\mathrm W}(D+\frac{i\hbar}{2}\Omega_{2m})S_{\mathrm W}^\top$,
so the assumed relation is equivalent to
\[
\left(
D+\frac{i\hbar}{2}\Omega_{2m}
\right)
S_{\mathrm W}^\top ZS_{\mathrm W}
=
0.
\]

Recall
$U_0=\frac1{\sqrt2}\smat{1&1\\i&-i}$ and
$\mathbb U=\bigoplus_{j=1}^{2m}{}_D U_0$ from
Remark~\ref{rem:unitary-spectral-coordinates}. Since $U_0^\dagger
\left(
c\I_2+\frac{i\hbar}{2}\Omega_1
\right)
U_0
=
\smat{c-\hbar/2&0\\0&c+\hbar/2}$, we have
\[
\mathbb U^\dagger
\left(
D+\frac{i\hbar}{2}\Omega_{2m}
\right)
\mathbb U
=
\bigoplus_{j=1}^p{}_D
\smat{0&0\\0&\hbar}
\oplus_D
\bigoplus_{j=p+1}^{2m}{}_D
\smat{\nu_j-\hbar/2&0\\0&\nu_j+\hbar/2}.
\]

Let $P_p^{\mathrm{shuf}}$ be the permutation matrix from
Proposition~\ref{prop:unitary-realification}. We have
\[
(P_p^{\mathrm{shuf}}\oplus_D\I_{2(2m-p)})^\top[
\bigoplus_{j=1}^p{}_D
\smat{0&0\\0&\hbar}
\oplus_D
\bigoplus_{j=p+1}^{2m}{}_D
\smat{\nu_j-\hbar/2&0\\0&\nu_j+\hbar/2}](P_p^{\mathrm{shuf}}\oplus_D\I_{2(2m-p)})=0_p\oplus_D\Xi,
\]
where $\Xi
:=
\hbar\I_p
\oplus_D
\bigoplus_{j=p+1}^{2m}{}_D
\smat{\nu_j-\hbar/2&0\\0&\nu_j+\hbar/2}
\succ0$.

The covariance slackness condition $(\Sigma+\frac{i\hbar}{2}\Omega_{2m})Z=0$ finally gives
\[
(0_p\oplus_D\Xi)B=0,
\]
where $B:=(P_p^{\mathrm{shuf}}\oplus_D\I_{2(2m-p)})^\top
\mathbb U^\dagger
S_{\mathrm W}^\top ZS_{\mathrm W}
\mathbb U
(P_p^{\mathrm{shuf}}\oplus_D\I_{2(2m-p)})\succeq0$. Consequently,
$\Ran B\subseteq\ker(0_p\oplus_D\Xi)
=\C^p\oplus\{0\}^{4m-p}$. Since $B$ is Hermitian, there exists
$M_Z\in\cM_{p,\geq}(\C)$ such that
$B=M_Z\oplus_D0_{4m-p}$.

Let $r=\rank M_Z\leq p$, and choose $V\in U(p)$ such that
\[
V^\dagger M_ZV
=
\operatorname{diag}(d_1,\ldots,d_r,0,\ldots,0),
\quad
d_1,\ldots,d_r>0.
\]
By Remark~\ref{rem:unitary-spectral-coordinates}, let $O_{\overline{V}}$ be the orthogonal symplectic matrix corresponding to $\overline{V}$. $O_{\overline V}\oplus_D\I_{2(2m-p)}
\in\Sp(4m,\R)\cap O(4m)$ and preserves $D$. Hence
$S:=S_{\mathrm W}(O_{\overline V}\oplus_D\I_{2(2m-p)})$
gives $\Sigma=SDS^\top$ and $\Sigma+\frac{i\hbar}{2}\Omega_{2m}
=S(D+\frac{i\hbar}{2}\Omega_{2m})S^\top$, and
\[
S^\top ZS
=
\bigoplus_{j=1}^r{}_D
U_0
\smat{d_j&0\\0&0}
U_0^\dagger
\oplus_D0_{2(2m-r)}
=
\bigoplus_{j=1}^r{}_D
\frac{d_j}{2}
\smat{1&-i\\i&1}
\oplus_D0_{2(2m-r)}.
\]
The last equality is obtained by $U_0\smat{d&0\\0&0}U_0^\dagger
=\frac d2\smat{1&-i\\i&1}$. Moreover,
\[
S^\top(\Real Z)S
=
\frac12
\operatorname{diag}
\left(
d_1\I_2,\ldots,d_r\I_2,
0_{2(2m-r)}
\right).
\]
Therefore $r=\frac12\rank(\Real Z)=q$, and
$\ker(\Real Z)=S\left(\{0\}^{2q}\oplus\R^{2(2m-q)}\right)$ follows.
Since $\Sigma=SDS^\top$ with the same $D$ as above and $q\leq p$,
Lemma~\ref{lem:Williamson-Gaussian-source} applies with this $q$
and this $S$. Writing $S=(S_0\ S_1)$ gives
$\Ran(S_1)=\ker(\Real Z)$. Since $S$ is invertible, $S_1$ has full
column rank.
\end{proof}

\section{Isotropic symmetrization of the covariance matrix}
The proof below is based on an averaging argument over the Pauli matrices, a technique sometimes referred to in the literature as Pauli twirling.
\begin{prop}[Reduction to the isotropic one-mode case]
\label{prop:covidm}
Let $\tau\in\{\pm1\}$, $\Sigma_{\rho},\Sigma_{\sigma}
\in\cM_{2,\mathrm{sym}}(\mathbb{R})$ and $X\in\cM_2(\mathbb{R})$. \rev{If}
\[
M_\tau:=
\begin{pmatrix}
\Sigma_{\rho} & X\\
X^{\top} & \Sigma_{\sigma}
\end{pmatrix}
+\frac{i\hbar}{2}\bigl(\omega\oplus_D(\tau\omega)\bigr)
\succeq0,
\]
then
\[
\begin{pmatrix}
\frac{\Tr\Sigma_{\rho}}{2}\I_2
&
\frac{\Tr X}{2}\I_2
\\[1mm]
\frac{\Tr X}{2}\I_2
&
\frac{\Tr\Sigma_{\sigma}}{2}\I_2
\end{pmatrix}
+\frac{i\hbar}{2}\bigl(\omega\oplus_D(\tau\omega)\bigr)
\succeq0.
\]
\end{prop}
\begin{proof}
Consider the Pauli matrices
\[
P_1=
\begin{pmatrix}
0&1\\
1&0
\end{pmatrix},
\qquad
P_2=
\begin{pmatrix}
0&-i\\
i&0
\end{pmatrix},
\qquad
P_3=
\begin{pmatrix}
1&0\\
0&-1
\end{pmatrix},\qquad 
P_4=\I_2.
\]
Then $P_j\omega P_j
=(-1)^j\omega$ and $
\frac14\sum_{j=1}^4P_j AP_j
=
\frac{\Tr A}{2}\I_2$ for all $A\in\cM_2(\mathbb R)$. Therefore, for $j\in[4]$,
\[
(P_j\oplus_D P_j) M_{\tau}(P_j\oplus_D P_j)
=
\begin{pmatrix}
P_j\Sigma_{\rho}P_j
&
P_j XP_j
\\
(P_j XP_j)^{\top}
&
P_j\Sigma_{\sigma}P_j
\end{pmatrix}
+
\frac{i(-1)^j\hbar}{2}\bigl(\omega\oplus_D(\tau\omega)\bigr).
\]
Since each $P_j\oplus_D P_j$ is unitary and Hermitian, and since both unitary congruence and transposition preserve positive semidefiniteness, we have
\begin{equation}
\begin{split}
&\frac14\left[\left(\sum_{j=1,3}(P_j\oplus_D P_j)^\dagger M_{\tau}(P_j\oplus_D P_j)\right)^\top+\sum_{j=2,4}(P_j\oplus_D P_j)^\dagger M_{\tau}(P_j\oplus_D P_j)\right]\\
=&\frac14\left[\left(\sum_{j=1,3}(P_j\oplus_D P_j)M_{\tau}(P_j\oplus_D P_j)\right)^\top+\sum_{j=2,4}(P_j\oplus_D P_j)M_{\tau}(P_j\oplus_D P_j)\right]\\
=&\begin{pmatrix}
\frac14\sum_{j=1}^4P_j\Sigma_{\rho}P_j
&
\frac14\sum_{j=1}^4P_j XP_j
\\
\frac14\sum_{j=1}^4(P_j XP_j)^{\top}
&
\frac14\sum_{j=1}^4P_j\Sigma_{\sigma}P_j
\end{pmatrix}+
\frac{i\hbar}{2}\bigl(\omega\oplus_D(\tau\omega)\bigr)\\
=&\begin{pmatrix}
\frac{\Tr\Sigma_{\rho}}{2}\I_2
&
\frac{\Tr X}{2}\I_2
\\[1mm]
\frac{\Tr X}{2}\I_2
&
\frac{\Tr\Sigma_{\sigma}}{2}\I_2
\end{pmatrix}
+
\frac{i\hbar}{2}\bigl(\omega\oplus_D(\tau\omega)\bigr)
\succeq0.
\end{split}
\end{equation}
The resulting matrix is positive semidefinite, since it is an average of positive semidefinite matrices and their transposes.
\end{proof}
\begin{coro}[Multimode extension]
\label{coro:multimode-isotropic-reduction}
Let $\tau\in\{\pm1\}$, $\Sigma_\rho,\Sigma_\sigma
\in\cM_{2m,\mathrm{sym}}(\mathbb R)$,
and $X\in\cM_{2m}(\mathbb R)$. \rev{If}
\[
\begin{pmatrix}
\Sigma_\rho & X\\
X^{\top} & \Sigma_\sigma
\end{pmatrix}
+
\frac{i\hbar}{2}
\bigl(\Omega_m\oplus_D(\tau\Omega_m)\bigr)
\succeq0,
\]
then
\[
\begin{pmatrix}
\frac{\Tr\Sigma_\rho}{2m}\I_{2m}
&
\frac{\Tr X}{2m}\I_{2m}
\\[1mm]
\frac{\Tr X}{2m}\I_{2m}
&
\frac{\Tr\Sigma_\sigma}{2m}\I_{2m}
\end{pmatrix}
+
\frac{i\hbar}{2}
\bigl(\Omega_m\oplus_D(\tau\Omega_m)\bigr)
\succeq0.
\]
\end{coro}
\begin{proof}
For \(j\in[m]\), let $\Sigma_{\rho}^{(j)}$, $X^{(j)}$, and $
\Sigma_{\sigma}^{(j)}$ denote the $2\times2$ diagonal blocks corresponding to rows and columns $2j-1$ and $2j$ of $\Sigma_{\rho}$, $X$, and $
\Sigma_{\sigma}$, respectively. Since principal compressions preserve positive semidefiniteness and $\Omega_m=\bigoplus_{j=1}^m{}_D\omega$, we obtain
\[
\begin{pmatrix}
\Sigma_{\rho}^{(j)} & X^{(j)}\\
(X^{(j)})^{\top} & \Sigma^{(j)}_{\sigma}
\end{pmatrix}
+
\frac{i\hbar}{2}
\bigl(\omega\oplus_D(\tau\omega)\bigr)
\succeq0.
\]
Applying the one-mode result, Proposition~\ref{prop:covidm}, gives
\[
M_j:=
\begin{pmatrix}
\frac{\Tr\Sigma_{\rho}^{(j)}}2\I_2
&
\frac{\Tr X^{(j)}}2\I_2
\\[1mm]
\frac{\Tr X^{(j)}}2\I_2
&
\frac{\Tr\Sigma_{\sigma}^{(j)}}2\I_2
\end{pmatrix}
+
\frac{i\hbar}{2}
\bigl(\omega\oplus_D(\tau\omega)\bigr)
\succeq0.
\]
Averaging over all modes yields
\[
\overline M
:=
\frac1m\sum_{j=1}^m M_j
=
\begin{pmatrix}
\frac{\Tr\Sigma_\rho}{2m}\I_2
&
\frac{\Tr X}{2m}\I_2
\\[1mm]
\frac{\Tr X}{2m}\I_2
&
\frac{\Tr\Sigma_\sigma}{2m}\I_2
\end{pmatrix}
+
\frac{i\hbar}{2}
\bigl(\omega\oplus_D(\tau\omega)\bigr)
\succeq0,
\]
where we used $\sum_{j=1}^m\Tr\Sigma_{\rho}^{(j)}=\Tr\Sigma_\rho$, $\sum_{j=1}^m\Tr X^{(j)}=\Tr X$ and $
\sum_{j=1}^m\Tr\Sigma_{\sigma}^{(j)}=\Tr\Sigma_\sigma$. Therefore, the positive semidefiniteness is preserved for the Kronecker product $\I_m\otimes\overline M\succeq0$. After permuting the coordinates from the mode-wise ordering to the subsystem-wise ordering, we obtain
\[
\begin{pmatrix}
\frac{\Tr\Sigma_\rho}{2m}\I_{2m}
&
\frac{\Tr X}{2m}\I_{2m}
\\[1mm]
\frac{\Tr X}{2m}\I_{2m}
&
\frac{\Tr\Sigma_\sigma}{2m}\I_{2m}
\end{pmatrix}
+
\frac{i\hbar}{2}
\bigl(\Omega_m\oplus_D(\tau\Omega_m)\bigr)
\succeq0.
\]
\end{proof}
\section{Proof of Proposition \ref{prop:qotcovid1}}\label{app:wassersteindistance}
\begin{proof}
We prove the formula for the quantum-state formulation. The formula for the quantum-channel formulation
follows from \cite[Thm.~3]{DPaTre19}, after restoring $\hbar$ and using
the cost normalization in \eqref{eq:tau-cost}.

By Theorem~\ref{thm:Gaussian-opt-plan}, we may restrict to Gaussian optimal couplings.
For a Gaussian optimal plan $\hat\pi_{\mathrm{opt}}^{G}\mapsto(\hat\rho^{\rm th},\hat\sigma^{\rm th})$, its mean is $(\mathbf{m}_{\hat{\rho}},\mathbf{m}_{\hat{\sigma}})$ by Proposition \ref{prop:reduction-moments}. We write its covariance as
\[
\Sigma_{\hat\pi_{\mathrm{opt}}^{G}}
=
\begin{pmatrix}
w_{\hat\rho}\,\I_{2m} & A\\
A^\top & w_{\hat\sigma}\,\I_{2m}
\end{pmatrix},
\qquad
\Sigma_{\hat\pi_{\mathrm{opt}}^{G}}+\frac{i\hbar}{2}\,(\Omega_m\oplus_D\Omega_m)\succeq0,
\]
with $A\in\cM_{2m}(\R)$. By Corollary~\ref{coro:multimode-isotropic-reduction}, feasibility of the above implies feasibility of
\[
\begin{pmatrix}
w_{\hat\rho}\,\I_{2m} & \frac{\Tr(A)}{2m}\,\I_{2m}\\[2pt]
\frac{\Tr(A)}{2m}\,\I_{2m} & w_{\hat\sigma}\,\I_{2m}
\end{pmatrix}
+\frac{i\hbar}{2}\,(\Omega_m\oplus_D\Omega_m)\succeq0.
\]
Let $\hat\pi_{\mathrm{iso}}^{G}$ be the Gaussian state with the same bipartite mean and this block
covariance. Since the cost depends only on first and second moments by Proposition~\ref{prop:stateTrC}, we have $\Tr(\hat C^{(+1)}\,\hat\pi_{\mathrm{iso}}^{G})
=\Tr(\hat C^{(+1)}\,\hat\pi_{\mathrm{opt}}^{G})$. Hence we may restrict to the one-parameter family
\[
\left\{
\begin{pmatrix}
w_{\hat\rho}\,\I_{2m} & c\,\I_{2m}\\[2pt]
c\,\I_{2m} & w_{\hat\sigma}\,\I_{2m}
\end{pmatrix}
\ :\ c\in\R,\ 
\begin{pmatrix}
w_{\hat\rho}\,\I_{2m} & c\,\I_{2m}\\
c\,\I_{2m} & w_{\hat\sigma}\,\I_{2m}
\end{pmatrix}
+\frac{i\hbar}{2}\,(\Omega_m\oplus_D\Omega_m)\succeq0
\right\}.
\]

Using the decomposition and
$\Omega_m=\I_m\otimes\omega$,
the RS constraint splits into $m$ identical $4\times4$ blocks; positivity is equivalent to
\[
M(c):=
\begin{pmatrix}
w_{\hat\rho} & \tfrac{i\hbar}{2} & c & 0\\
-\tfrac{i\hbar}{2} & w_{\hat\rho} & 0 & c\\
c & 0 & w_{\hat\sigma} & \tfrac{i\hbar}{2}\\
0 & c & -\tfrac{i\hbar}{2} & w_{\hat\sigma}
\end{pmatrix}\succeq0.
\]
The eigenvalues of $M(c)$ are
\[
\lambda_{1,2}=-\frac{\hbar}{2}+\frac{w_{\hat\rho}+w_{\hat\sigma}}{2}
\pm \frac{1}{2}\sqrt{(w_{\hat\rho}-w_{\hat\sigma})^{2}+4c^{2}},
\quad
\lambda_{3,4}=\frac{\hbar}{2}+\frac{w_{\hat\rho}+w_{\hat\sigma}}{2}
\pm \frac{1}{2}\sqrt{(w_{\hat\rho}-w_{\hat\sigma})^{2}+4c^{2}}.
\]
Thus $M(c)\succeq0$ iff $\lambda_{\min}(M(c))\ge0$, i.e.
\[
\sqrt{(w_{\hat\rho}-w_{\hat\sigma})^{2}+4c^{2}}
\le w_{\hat\rho}+w_{\hat\sigma}-\hbar,
\]
which yields the maximal feasible $c_{\max}
=\sqrt{\Big(w_{\hat\rho}-\frac{\hbar}{2}\Big)\Big(w_{\hat\sigma}-\frac{\hbar}{2}\Big)}$.

By Proposition~\ref{prop:stateTrC}, we have
\[
\begin{aligned}
\QW_{2,+1}^2(\hat\rho^{\rm th},\hat\sigma^{\rm th})
={}&\|\mathbf{m}_{\hat\rho}-\mathbf{m}_{\hat\sigma}\|^2+\min_{c\ \text{feasible}}\ 2m\,(w_{\hat\rho}+w_{\hat\sigma}-2c)
\\
={}&\|\mathbf{m}_{\hat\rho}-\mathbf{m}_{\hat\sigma}\|^2+2m\Big(w_{\hat\rho}+w_{\hat\sigma}-2c_{\max}\Big).
\end{aligned}
\]
Substituting $c_{\max}$ and completing the square gives
\[
\QW_{2,+1}^2(\hat\rho^{\rm th},\hat\sigma^{\rm th})
=2m\Big(\sqrt{w_{\hat\rho}-\tfrac{\hbar}{2}}
-\sqrt{w_{\hat\sigma}-\tfrac{\hbar}{2}}\Big)^{2}+\|\mathbf{m}_{\hat\rho}-\mathbf{m}_{\hat\sigma}\|^2+2m\hbar,
\]
as claimed.
\end{proof}
 
\section{Optimal couplings for the partially pure boundary example}
\label{app:nonfaithful-boundary-couplings}
We verify the couplings used in
Example~\ref{ex:nonfaithful-boundary} directly at the covariance
level.
\begin{lemma}[Boundary family of barycenters]
\label{lem:nonfaithful-boundary-couplings}
For the states in Example~\ref{ex:nonfaithful-boundary}, every
$\hat\rho_t^\tau$, with $t\in[-1,1]$ and
$\tau\in\{\pm1\}$, is an equal-weight $2$-QW barycenter of
$\hat\sigma_1$ and $\hat\sigma_2$.
\end{lemma}

\begin{proof}
By Lemma~\ref{lem:tensorization-barycenter},
$\hat\rho_0^{+1}=\left(\hat\sigma_{3\hbar/4}^{\mathrm{th}}\right)^{\otimes2}$ and
$\hat\rho_0^{-1}=\left(\hat\sigma_{\hbar}^{\mathrm{th}}\right)^{\otimes2}$ are barycenters
for $\tau=+1$ and $\tau=-1$, respectively, and therefore determine the
corresponding minimum barycenter values.

For $\tau=+1$, set
$A_{1,t}^{+}=\frac{\hbar}{2}\smat{\I_2&0\\ t\I_2&0}$ and
$A_{2,t}^{+}=\frac{\hbar}{2}\smat{0&t\I_2\\ 0&\I_2}$.
For $\tau=-1$, set
$A_{1,t}^{-}=\hbar\smat{\I_2&0\\ t\I_2&0}$ and
$A_{2,t}^{-}=\hbar\smat{0&t\I_2\\ 0&\I_2}$. For every $s\in\{1,2\}$ and $|t|\leq1$, direct substitution gives $\smat{\overline\Sigma_t^\tau&A_{s,t}^\tau\\
(A_{s,t}^\tau)^\top&\Sigma_s}
+\frac{i\hbar}{2}\bigl(\Omega_2\oplus_D(\tau\Omega_2)\bigr)\succeq0$. Thus $A_{s,t}^\tau\in
\mathrm{Corr}_\tau(\overline\Sigma_t^\tau,\Sigma_s)$. Let
$\hat\pi_{s,t}^\tau\in\Gamma_\tau(\hat\rho_t^\tau,\hat\sigma_s)$
be the centered Gaussian coupling with covariance
$\smat{\overline\Sigma_t^\tau&A_{s,t}^\tau\\
(A_{s,t}^\tau)^\top&\Sigma_s}$.

Since $\Tr(\overline\Sigma_t^\tau)$ and $\Tr(A_{s,t}^\tau)$ are
independent of $t$, Proposition~\ref{prop:stateTrC} gives $\Tr\!\left(\hat C^{(\tau)}\hat\pi_{s,t}^\tau\right)=\Tr\!\left(\hat C^{(\tau)}\hat\pi_{s,0}^\tau\right)$. At $t=0$, $\hat\rho_0^\tau$ is the tensor-product barycenter above,
and $\hat\pi_{s,0}^\tau$ is the tensor product of the corresponding one-mode optimal couplings. Hence $\frac12\sum_{s=1}^2
\Tr\!\left(\hat C^{(\tau)}\hat\pi_{s,t}^\tau\right)
=
B_\tau(\hat{\boldsymbol{\sigma}}\mid\boldsymbol{\alpha})$. Therefore $\hat\rho_t^\tau$ attains the minimum and is an equal-weight
$2$-QW barycenter of $\hat\sigma_1$ and $\hat\sigma_2$.
\end{proof}

 \bibliographystyle{plain}
\bibliography{references}

\setcounter{section}{0}

\renewcommand{\thesection}{\Alph{section}}

\end{document}